\documentclass{article}
\usepackage{kr}

\usepackage{times}
\usepackage{soul}
\usepackage{url}
\usepackage[hidelinks]{hyperref}
\usepackage[utf8]{inputenc}
\usepackage[small]{caption}
\usepackage{graphicx}
\usepackage{amsmath}
\usepackage{amsthm}
\usepackage{booktabs}
\usepackage{algorithm}
\usepackage{algorithmic}
\usepackage{wrapfig}
\usepackage{amsthm,amsmath,amssymb,amsfonts}
\usepackage{algorithmic}
\usepackage{totcount}
\usepackage[inline,shortlabels]{enumitem}
\usepackage{refcount}
\usepackage{graphicx}
\usepackage{textcomp}
\usepackage{bm}
\usepackage{xcolor}
\usepackage{newfile}
\usepackage{enumitem}
\usepackage{thm-restate}
\def\BibTeX{{\rm B\kern-.05em{\sc i\kern-.025em b}\kern-.08em
    T\kern-.1667em\lower.7ex\hbox{E}\kern-.125emX}}

\usepackage[capitalize]{cleveref}
\usepackage{bm}
\usepackage{todonotes}
\usepackage{stmaryrd}
\usepackage{tikz}
\usepackage{mathdots}

\newtheorem{theorem}{Theorem}[section]
\newtheorem{definition}[theorem]{Definition}
\newtheorem{lemma}[theorem]{Lemma}

\newtheorem{proposition}[theorem]{Proposition}
\newtheorem{claim}[theorem]{Claim}
\newtheorem{corollary}[theorem]{Corollary}

\theoremstyle{remark}

\newtheorem{remark}[theorem]{Remark}

\crefname{observation}{Observation}{Observations}
\crefname{appsec}{Appendix}{Appendix}

\newcommand{\set}[1]{\ensuremath{\mathsf{set}\big({#1}\big)}}
\newcommand{\N}{\mathbb{N}}
\newcommand{\Z}{\mathbb{Z}}
\newcommand{\Q}{\mathbb{Q}}

\newcommand{\fdiv}{\mathsf{div}}

\newcommand{\PA}{\mathsf{PA}}
\newcommand{\PFnS}{\mathsf{PFnS}}
\newcommand{\BFnS}{\mathsf{BFnS}}
\newcommand{\PASyn}{\mathsf{PSyNF}}
\newcommand{\DNNF}{\mathsf{PSySyNF}}

\newcommand{\bigO}[1]{\ensuremath{\mathcal{O}\big({#1}\big)}}
\newcommand{\circuit}{\mathcal{C}} 
\newcommand{\ckt}{\mathcal{C}} 
\newcommand{\ite}{\mathsf{ite}}
\newcommand{\lexists}{\exists^{\mathsf{local}}}

\newcommand{\coNEXP}{\mathsf{coNEXP}}
\newcommand{\coNP}{\mathsf{coNP}}
\newcommand{\NP}{\mathsf{NP}}

\newcommand{\Cp}{\mathsf{C}}
\newcommand{\Eq}{\mathsf{E}}

\DeclareMathOperator{\lcm}{lcm}

\begin{document}
\title{Presburger Functional Synthesis: Complexity and Tractable Normal Forms \thanks{Full version of paper at KR 2025 (22nd International Conference on Principles of Knowledge Representation and Reasoning).}}
\author{%
	S. Akshay$^1$\and
	A. R. Balasubramanian$^2$\and
	Supratik Chakraborty$^1$\and
	Georg Zetzsche$^2$\\
	\affiliations
	$^1$Indian Institute of Technology Bombay, Mumbai, India \\
	$^2$Max Planck Institute for Software Systems (MPI-SWS), Germany
}

\maketitle

\begin{abstract}
Given a relational specification between inputs and outputs as a logic
formula, the problem of functional synthesis is to automatically
synthesize a function from inputs to outputs satisfying the
relation. Recently, a rich line of work has emerged tackling this
problem for specifications in different theories, from Boolean to
general first-order logic. In this paper, we launch an investigation
of this problem for the theory of Presburger Arithmetic, that we call
Presburger Functional Synthesis (PFnS).
We show that PFnS can be solved in EXPTIME and provide a matching
exponential lower bound. This is unlike the case for Boolean
functional synthesis (BFnS), where only conditional exponential lower
bounds are known. Further, we show that PFnS for one input and one
output variable is as hard as BFnS in general. We then identify a
special normal form, called PSyNF, for the specification formula that
guarantees poly-time and poly-size solvability of PFnS. We prove
several properties of PSyNF, including how to check and compile to
this form, and conditions under which any other form that guarantees
poly-time solvability of PFnS can be compiled in poly-time to
PSyNF. Finally, we identify a syntactic normal form that is easier to
check but is exponentially less succinct than PSyNF.

\end{abstract}

\section{Introduction}
\label{sec:introduction}

Automated synthesis, often described as a holy grail of computer
science, deals with the problem of automatically generating correct
functional implementations from relational
specifications. Specifications are typically presented as relations,
encoded as first-order logic (FOL) formulas over a set of free
variables that are partitioned into inputs and outputs.  The goal of
automated functional synthesis is to synthesize a function from inputs
to outputs such that for every valuation of the inputs, if it is
possible to satisfy the specification, then the valuation of outputs
produced by the function also satisfies it.  The existence of such
functions, also called \emph{Skolem functions}, is well-known from the
study of first-order logic~\cite{Enderton1972,spec-logic-variants}%
. However, it is not always possible to obtain succinct representations or
efficiently executable descriptions of Skolem
functions~\cite{DBLP:conf/mfcs/ChakrabortyA22}. This has motivated researchers
to study the complexity of functional synthesis in different first-order
theories, and investigate specific normal forms for specifications that
enable efficient functional synthesis.

In the simplest setting of Boolean (or propositional) specifications,
Boolean functional synthesis (henceforth called $\BFnS$) has received
significant attention in the recent past~\cite{JSCTA15,RS16,FTV16,CFTV18,manthan,DBLP:conf/cav/AkshayCJ23,LTV24} among
others. Even for this restricted class, functional synthesis cannot
be done efficiently unless long-standing complexity theoretic
conjectures are falsified~\cite{AkshayCGKS21}. Nevertheless, several
practical techniques have been developed, including counter-example
guided approaches~\cite{JSCTA15,AkshayCGKS21,manthan,manthantwo},
input-output separation based approaches~\cite{CFTV18}, machine
learning driven approaches~\cite{manthan,manthantwo}, BDD and ZDD
based approaches~\cite{FTV16,LTV22,LTV24}. Researchers
have also studied {\em knowledge representations} or normal forms for
specifications that guarantee efficient
$\BFnS$~\cite{AkshayCGKS21,AACKRS19,DBLP:conf/cav/AkshayCJ23,DBLP:journals/amai/AkshayCS24},
with~\cite{SBAC21} defining a form that precisely characterizes when
$\BFnS$ can be solved in polynomial time and space.

Compared to $\BFnS$, work on functional synthesis in theories beyond
Boolean specifications has received far less attention, even though such theories are
widely applicable in real-life specifications. One such important
extension is to theories of linear arithmetic over reals and integers.
The work of~\cite{KMPS10,KMPS13} deals with complete functional
synthesis for quantifier-free linear real arithmetic (QF\_LRA) and
linear integer arithmetic (QF\_LIA).
Similarly,~\cite{DBLP:conf/cav/Jiang09} goes beyond Boolean
specifications, and points out that Skolem functions may not always be
expressible as terms in the underlying theory of the specification,
necessitating an extended vocabulary. For specifications in
QF\_LIA,~\cite{DBLP:conf/cp/FedyukovichG19,FGG19} build tools for
synthesizing (or extracting) Skolem functions as terms.

In this paper, our goal is to study functional synthesis from specifications in Presburger arithmetic ($\PA$ for short), that extends QF\_LIA with modular constraints. $\PA$ has been extensively studied in the literature (see~\cite{DBLP:journals/siglog/Haase18} for a survey) and admits multiple interpretations, including geometric and logic-based interpretations; see, e.g.~\cite{DBLP:conf/fsttcs/000124}. Recent work has shown significant improvements in the complexity of quantifier elimination for $\PA$; see e.g.~\cite{DBLP:conf/icalp/HaaseKMMZ24,DBLP:conf/icalp/0001MS24}.  Since $\PA$ admits effective quantifier elimination, it follows from~\cite{DBLP:conf/mfcs/ChakrabortyA22} that for every $\PA$ specification, Skolem functions for all outputs can be synthesized as halting Turing machines.  Unfortunately, this does not give good complexity bounds on the time required to compute Skolem functions. Our focus in this paper is to fill this gap by providing {\em optimal complexity results for $\PFnS$ as well as normal forms for tractable synthesis}. %

Before we proceed further, let us see an example of a $\PA$
specification, and an instance of $\PFnS$. Consider a factory with two
machines $M_1$ and $M_2$. Suppose $M_1$ must pre-process newly arrived
items before they are further processed by $M_2$. Suppose further that
$M_1$ can start pre-processing an item at any integral time instant
$k$ (in appropriate time units), and takes one time unit for
pre-processing. $M_2$, on the other hand, can start processing an item
only at every 2nd unit of time, and takes one time unit to
process. Suppose items $I_1, \ldots I_n$ arrive at times
$t_{1}, \ldots t_{n}$ respectively, and we are told that the job
schedule must satisfy three constraints. First, $M_1$ must finish
pre-processing each item exactly $1$ time unit before $M_2$ picks it up
for processing; otherwise, the item risks being damaged while waiting
for $M_2$.  In general, this requires delaying the start-time of
pre-processing $I_i$ by $\delta_i~(\ge 0)$ time units so that the
end-time of pre-processing aligns with one time instant before $2r$,
for some $r \in \mathbb{N}$.
Second, the (pre-)processing windows for different items must not
overlap. Third, the total weighted padded delay
must not exceed a user-provided cap $\Delta$, where the weight
for item $i$ is $i$.
Formalizing the above constraints in $\PA$, we obtain the specification
$\varphi~\equiv~ \varphi_1 \wedge \varphi_2 \wedge \varphi_3$, where
$\varphi_1 \equiv \bigwedge_{i=1}^n \big(t_i + \delta_i + 1 \equiv 1 \pmod 2\big)$, 
$\varphi_2 \equiv \bigwedge_{1 \le i < j \le n} \big((t_i + \delta_i + 1 < t_j + \delta_j) \vee (t_j + \delta_j + 1 < t_i + \delta_i)\big)$, 
$\varphi_3 \equiv \bigwedge_{i=1}^n (\delta_i \ge 0) \wedge \big(\sum_{i=1}^n i. \delta_i \leq \Delta\big)$.
Here, $t_1, \ldots t_n$ and $\Delta$ are input variables, while
$\delta_1, \ldots \delta_n$ are output variables.  The {\em functional
synthesis problem} then asks us to synthesize the delays, i.e., functions $f_1, \ldots f_n$ that take $t_1, \ldots t_n, \Delta$ as inputs and produce values of $\delta_1, \ldots \delta_n$ such that $\varphi$ is satisfied, whenever possible.   

\noindent{\bf Our contributions.} As a first step, we need a representation for Skolem functions, for which we propose \emph{Presburger circuits}, constructed by composing basic Presburger ``gates''. We identify a (minimal) collection of these gates such that every Presburger-definable function (closely related to those defined in~\cite{DBLP:journals/siamcomp/IbarraL81}) can be represented as a circuit made of these gates. %
Using Presburger circuits as representations for Skolem functions, we examine the complexity of $\PFnS$ and develop knowledge representations that make $\PFnS$ tractable. Our main contributions are:%
\begin{enumerate}
  \item We provide a {\em tight complexity-theoretic
    characterization} for $\PFnS$. Specifically:
    \begin{enumerate}
    \item We show that for every $\PA$ specification
	    $\varphi(\bar{x},\bar{y})$, we can construct in
    $\bigO{2^{|\varphi|^{\mathcal{O}(1)}}}$ time a Presburger
    circuit of size $\bigO{2^{|\varphi|^{\mathcal{O}(1)}}}$
    that represents a Skolem function for $\bar{y}$. This exponential upper bound 
    significantly improves upon earlier
    constructions~\cite{DBLP:journals/siamcomp/Cherniavsky76,DBLP:journals/siamcomp/IbarraL81}
    for which we argue that the resulting Presburger circuits
    would be of at least triply- or even quadruply-exponential size,
    respectively.
    \item We show that the exponential blow-up above is unavoidable, by exhibiting a family $(\mu_n)_{n\ge
        0}$ of $\PA$ specifications of size polynomial in $n$, such
      that any Presburger circuit for any Skolem function for 
      $\mu_n$ must have size at least $2^{\Omega(n)}$.  This
      unconditional lower bound for $\PFnS$ stands in contrast to the Boolean case ($\BFnS$), where lower bounds are conditional on long-standing conjectures from complexity theory.%
    \item We show that $\PFnS$ from one-input-one-output
      specifications is already as hard as $\BFnS$ in general. As a
      corollary, unless $\mathsf{NP} \subseteq 
      \mathsf{P}/\mathsf{poly}$, the size of Skolem functions for 
      one-input-one-output specifications must grow super-polynomially in
      the size of the specification in the worst-case.
    \end{enumerate}
  \item The above results imply that efficient $\PFnS$ algorithms do not exist,
    and so, we turn to {\em knowledge representations}, i.e.,
    studying normal forms of $\PA$ specifications that admit efficient Skolem
    function synthesis.  

  \begin{enumerate}
  \item For one-output $\PA$ specifications, we define the
    notion of {\em modulo-tameness}, and prove that every $y$-modulo
    tame specification $\varphi(\bar{x},y)$ admits polynomial-time
    synthesis of Presburger circuits for a Skolem function.
  \item We lift this to $\PA$ specifications with multiple
    output variables, and provide a {\em semantic normal form} called
    $\PASyn$ that enjoys the following properties:
    \begin{enumerate}
    \item $\PASyn$ is universal: every $\PA$ specification can be
      compiled to $\PASyn$ in worst-case exponential time (unavoidable
      by our hardness results above).
    \item $\PASyn$ is good for existential quantification and
      synthesis: Given any specification in $\PASyn$, we can
      effectively construct Presburger circuits for Skolem functions in
      time polynomial in the size of the specification. Additionally,
      we can also existentially quantify output variables in
      polytime.
    \item $\PASyn$ is effectively checkable: Given any $\PA$
      specification, checking if it is in $\PASyn$ in
      $\coNP$-complete. As a byproduct of independent interest, we obtain that the
      $\exists^*\forall$ fragment of $\langle\Z;+,<,0,1\rangle$ is $\NP$-complete.
    \item $\PASyn$ is optimal for one output: For every universal normal form
      of single-output $\PA$ specifications that admits polynomial-time existential quantification of the output, we can compile
      formulas in that form to $\PASyn$ in polynomial time.
    \end{enumerate}
  \item We provide a {\em syntactic normal form} for $\PA$ specifications,
    called $\DNNF$, that is universal and efficiently checkable (in
    time linear in the size of the formula), but is exponentially less
    succinct than $\PASyn$.
  \end{enumerate}
\end{enumerate}

{\em Structure.} The paper is organized as follows. In
Section~\ref{sec:prelim}, we start with preliminaries and define
the problem statement and representations in Section~\ref{sec:problem}. %
Our main complexity results for $\PFnS$ are in Section~\ref{sec:synthesis}. We
present our semantic normal forms in Section~\ref{sec:semantic-nf} and
syntactic forms in Section~\ref{sec:syntactic-nf} and conclude in Section~\ref{sec:conclusion}. 
Due to lack of space, many of the proofs and some more details have been provided in the supplementary material.

{\em Related Work.} A circuit representation similar to ours, but using
a slightly different set of gates, was (implicitly) studied
in~\cite[Theorem 6]{DBLP:journals/siamcomp/IbarraL81} in the context
of representing Presburger-definable functions. However, their
formalism is closely tied to the setting of natural numbers, making it
somewhat cumbersome in the setting of integers, for which our circuit
representation appears more natural. In addition, specialized
programming languages for describing Presburger-definable functions
have been studied in the literature, examples being
SL~\cite{DBLP:journals/jacm/GurariI81} and
$L_+$~\cite{DBLP:journals/siamcomp/Cherniavsky76}, among others.
However, because of the loopy nature of these programming languages,
such programs do not guarantee as efficient evaluation of the
functions as circuits do.

The problem of functional synthesis is intimately
related to that of quantifier elimination, and our work leverages
recent advances in quantifier elimination for $\PA$~\cite{DBLP:conf/icalp/HaaseKMMZ24,DBLP:conf/icalp/0001MS24}.
However, being able to effectively eliminate quantifiers does not
automatically yield an algorithm for synthesizing Presburger circuits. Hence, although our work bootstraps on recent results in
quantifier elimination for $\PA$, and draws inspiration from knowledge
representation for Boolean functional synthesis, the core techniques
for synthesizing Skolem functions are new. In fact, our knowledge compilation results yield a new alternative approach to quantifier elimination from $\PA$ formulas, that can result in sub-exponential (even polynomial) blow-up in the size of the original formula, if the formula is in a special form.  This is in contrast to state-of-the-art quantifier
elimination techniques~\cite{DBLP:conf/icalp/HaaseKMMZ24,DBLP:conf/icalp/0001MS24} that always yield an exponential
blow-up. %

\section{Preliminaries}
\label{sec:prelim}
\noindent \emph{Presburger Arithmetic:}
Presburger arithmetic (PA) is the first-order theory of the structure
$\langle \mathbb{Z},+,<,0,1 \rangle$. Presburger arithmetic is
well-known to admit quantifier elimination, as originally shown by Moj\.zesz Presburger in 1929~\cite{presburger-original} (see~\cite{DBLP:journals/siglog/Haase18} for a modern survey). That is, every formula in PA with quantifiers can be converted into an equivalent one without
quantifiers, at the cost of introducing \emph{modulo constraints},
which are constraints of the form $\sum_{i=1}^n a_i x_i \equiv r \pmod M$, where
$x_1,\dots,x_n$ are variables, and $a_1, \ldots a_n, r, M$ are integer
constants with $0 \le r < M$.  The constraint $\sum_{i=1}^n a_i x_i
\equiv r \pmod M$ is semantically equivalent to $\exists k \in \mathbb{Z}:~
\sum_{i=1}^n a_i x_i = kM + r$.  We say $M$ is the modulus of the constraint,
and $r$ its residue.  For notational convenience, we sometimes use
$\sum_{i=1}^m a_i x_i \equiv_M r$ for $\sum_{i=1}^m a_i x_i \equiv
r \pmod M$.  Hence, technically, we are working over the structure
$\langle \mathbb{Z},+,<, (\equiv_M)_{M \in \mathbb{Z}}, 0,1 \rangle$.
For variables $\bar{x}=(x_1,\ldots,x_n)$ and vectors $\bar{r}=(r_1,\ldots,r_n)$ of constants $r_1,\ldots,r_n\in [0,M-1]$, we will use the shorthand $\bar{x}\equiv\bar{r}\pmod{M}$ to mean $\bigwedge_{i=1}^n x_i\equiv r_i\pmod{M}$.

A linear inequality is a formula of the form $a_1 x_1 + \dots + a_n
x_n + b \ge 0$ (or equivalently, $a_1 x_1 + \dots + a_n x_n + b + 1 >
0$) for variables $x_1,\dots,x_n$ and constants $a_1,\dots,a_n,b
\in \mathbb{Z}$. An atomic formula is either a linear inequality or a
modulo constraint. Note that every quantifier-free formula over
$\langle \mathbb{Z},+,<, (\equiv_M)_{M \in \mathbb{Z}}, 0,1 \rangle$
is simply a Boolean combination of atomic formulas.  Throughout this
paper, we assume that all constants appearing in formulas are encoded
in binary.  We use variables with a bar at the top, viz.  $\bar{x}$,
to denote a tuple of variables, such as $(x_1,\dots,x_n)$. With abuse
of notation, we also use $\bar{x}$ to denote the underlying set of
variables, when there is no confusion.

A quantifier-free $\PA$ formula $\varphi(\bar{x})$ is said to be in
\emph{negation normal form (NNF)} if no sub-formulas other than atomic
sub-formulas, are negated in $\varphi$.  By applying DeMorgan's rules,
a quantifier-free $\PA$ formula can be converted to NNF in time linear
in the size of the formula.  Therefore, we assume all quantifier-free
$\PA$ formulas are in NNF.  We represent such a formula as
a \emph{tree} in which each internal node is labeled by $\land$ or
$\lor$, and each leaf is labeled by a linear inequality of the form
$\sum_{k=1}^n a_kx_k + b \ge 0$, or by a modulo constraint of the form
$\sum_{k=1}^n a_kx_k \bowtie r \pmod M$, where
$\mathord{\bowtie} \in \{\equiv, \not\equiv\}$, $a_1, \ldots a_n$,
$b$, $r$ and $M$ are integers, and $0 \le r < M$. We identify every
node $v$ in the tree with the sub-formula of $\varphi$ represented by
the sub-tree rooted at $v$.  Specifically, the root of the tree is
identified with the formula $\varphi$.  The \emph{size} of a
quantifier-free $\PA$ formula $\varphi$, denoted $|\varphi|$, is the
sum of the number of nodes in the tree representation of $\varphi$,
the number of variables, and the number of bits needed to encode each
constant in the atomic formulas in the leaves.
As an example, Fig.~\ref{fig:tree-rep} shows a tree representing the
formula $\big((3x - 2 \ge 0)\land(4x + 5y \equiv 2 \pmod 3)\big)
\lor \big((-2x + 5y + 7 \ge 0) \land (y \equiv 5 \pmod
6)\big)$.

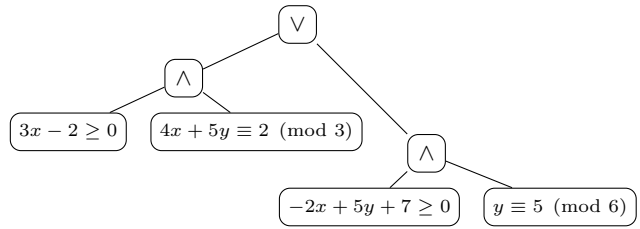
\begin{figure}
\begin{center}
	\begin{tikzpicture}[
		node distance=2mm and 5mm,
		every node/.style={draw, rounded corners, text centered, minimum size=5mm, font=\footnotesize},
		edge from parent/.style={thick}
		]
		\node (lor) {$\lor$};

		\node[below left=of lor,xshift=-5mm] (land1) {$\land$};
		\node[below left=of land1] (f1) {\scriptsize $3x - 2 \ge 0$};
		\node[below right=of land1,xshift=-12mm] (f2) {\scriptsize $4x + 5y \equiv 2\pmod{3}$};

		\node[below right=of lor, xshift = 7mm, yshift=-10mm] (land2) {$\land$};
		\node[below left=of land2,xshift=12mm] (f3) {\scriptsize $-2x + 5y + 7 \ge 0$};
		\node[below right=of land2] (not) (f4) {\scriptsize $y \equiv 5 \pmod{6}$};

		\draw (lor) -- (land1);
		\draw (lor) -- (land2);
		
		\draw (land1) -- (f1);
		\draw (land1) -- (f2);
		
		\draw (land2) -- (f3);
		\draw (land2) -- (f4);
	\end{tikzpicture}
\end{center}
\caption{Tree representation of $\big((3x - 2 \ge 0)\land(4x + 5y \equiv 2 \pmod 3)\big)
	\lor \big((-2x + 5y + 7 \ge 0) \land (y \equiv 5 \pmod
	6)\big)$. }
\label{fig:tree-rep}
\end{figure}

\section{Presburger Functional Synthesis}
\label{sec:problem}
The central problem in this paper is \emph{Presburger functional
synthesis} ($\PFnS$).  Intuitively, we have a tuple of input variables
$\bar{x}$ and a tuple of output variables $\bar{y}$, with each
variable ranging over $\mathbb{Z}$. In addition, we are also given a
quantifier-free $\PA$ formula $\varphi(\bar{x},\bar{y})$ that we
interpret as a relational specification between the inputs and
outputs. Our task is to find (and represent) a function $f$ with
inputs $\bar{x}$ and outputs $\bar{y}$ such that the specification
$\varphi$ is satisfied by this function, whenever possible. Such a
function is called a \emph{Skolem function}. More formally:
\begin{definition}
	Let $\varphi(\bar{x},\bar{y})$ be a quantifier-free $\PA$
        formula, where $\bar{x}$ denotes $(x_1, \ldots x_n)$ and
        $\bar{y}$ denotes $(y_1, \ldots y_m)$.  A function $f\colon
        \mathbb{Z}^n \rightarrow \mathbb{Z}^m$ is called a
        \emph{Skolem function for $\bar{y}$ in $\forall \bar{x}
        \exists \bar{y}\colon \varphi(\bar{x},\bar{y})$} if for every
	value $\bar{u} \in \mathbb{Z}^n$ of $\bar{x}$, $\exists
	\bar{y}\colon\varphi(\bar{u},\bar{y})$ holds if and only if $\varphi(\bar{u},f(\bar{u}))$ holds.
\end{definition}

\vspace{-0.5cm}
\paragraph{A syntax for Skolem functions}
Since our goal is to synthesize Skolem functions, we need a syntax to represent them.
We introduce such a syntax, called {\em Presburger circuits}, which are a variant of a syntax studied implicitly by Ibarra and Leininger~\shortcite{DBLP:journals/siamcomp/IbarraL81}. The notion of Presburger circuits is designed to achieve two key properties:
\begin{enumerate}
	\item Efficient evaluation: Given a Presburger circuit for a function $f\colon\Z^n\to\Z$ and a vector $\bar{u}\in\Z^n$, one can compute $f(\bar{u})$ in polynomial time.
	\item Completeness: Every Presburger formula has a Skolem function defined by some Presburger circuit.
\end{enumerate}
Let us describe Presburger circuits in detail. A Presburger circuit consists of a set of \emph{gates}, each of which computes a function from a small set called \emph{atomic functions}. The atomic functions are
\begin{enumerate}
	\item linear functions with integer coefficients, i.e.\ $\Z^n\to\Z$, $(u_1,\ldots,u_n)\mapsto a_0+\sum_{i=1}^n a_iu_i$ for $a_0,\ldots,a_n\in\Z$.
	\item the maximum function $\max\colon\Z\times\Z\to\Z$.
	\item the equality check function, i.e.\ $\Eq\colon\Z\times\Z\to\Z$ with
		$\Eq(x,y)=y$ if $x=0$ and $\Eq(x,y)=0$ if $x\ne 0$.
	\item division functions $\fdiv_m\colon x\mapsto \lfloor x/m\rfloor$ for $m\in\N\setminus\{0\}$.
\end{enumerate}
More formally, a \emph{Presburger circuit} is a collection of gates, each
labeled either with an atomic function or with an input variable $x_i$, $i=1,\ldots,n$. If a
gate $g$ is labeled by an atomic function $f\colon\Z^k\to\Z$, then there are
$n$ edges $e_1,\ldots,e_k$, each connecting some gate $g_i$ to $g$.
Intuitively, these edges provide the input to the gate $g$. Hence,
$g_1,\ldots,g_k$ are called the \emph{input gates} of $g$. Finally, there is a
list of $m$ distinguished \emph{output gates} $g_1,\ldots,g_m$. The output gates are the ones that compute the
output vector $\in\Z^m$ of the Presburger circuit.  

A Presburger circuit must be
\emph{acyclic}, meaning the edges between gates form no cycle. This acyclicity
allows us to evaluate a Presburger circuit for a given input
$(u_1,\ldots,u_n)$: First, the gates labeled by input variables evaluate to the
respective value. Then, a gate labeled with an atomic function
$f\colon\Z^n\to\Z$ evaluates to $f(u_1,\ldots,u_n)$, where $u_i$ is the result
of evaluating the $i$-th input gate of $g$. Finally, the \emph{output} of the
Presburger circuit is the output (i.e.\ evaluation result) of the distinguished
output gates. Overall, the Presburger circuit computes a function $\Z^n\to\Z^m$.

To simplify terminology, Presburger circuits that compute Skolem functions will
also be called \emph{Skolem circuits}.

\paragraph{Properties of Presburger circuits}
First, it is obvious that a Presburger circuit can be evaluated in polynomial
time.  Moreover, we will show that Presburger circuits are expressively
complete for Skolem functions in Presburger arithmetic. Indeed, the following
is a direct consequence of \cref{thm:exponential-presburger-circuit}, which
will be shown in \cref{sec:synthesis}:
\begin{restatable}{theorem}{presburgerCircuitsCompletenessSkolem}\label{presburger-circuits-completeness-skolem}
	For every quantifier-free formula $\varphi(\bar{x},\bar{y})$, there
	exists a Skolem circuit for
	$\bar{y}$ in
	$\forall\bar{x}\exists\bar{y}\colon\varphi(\bar{x},\bar{y})$.
\end{restatable}
Equivalently, Presburger circuits describe exactly those functions that can be
defined in Presburger arithmetic. Formally, a function $f\colon\Z^n\to\Z^m$ is
\emph{Presburger-definable} if there exists a Presburger formula
$\varphi(\bar{x},\bar{y})$, $\bar{x}=(x_1,\ldots,x_n)$,
$\bar{y}=(y_1,\ldots,y_m)$, such that for all $\bar{u}\in\Z^n$ and
$\bar{v}\in\Z^m$, we have $\varphi(\bar{u},\bar{v})$ if and only if
$f(\bar{u})=\bar{v}$. The following is an alternative characterization:
\begin{restatable}{theorem}{presburgerCircuitsCompletenessDefinable}\label{presburger-circuits-completeness-definable}
	A function $\Z^n\to\Z^m$ is computable by a Presburger circuit if and
	only if it is Presburger-definable.
\end{restatable}
Note that \cref{presburger-circuits-completeness-definable} follows directly
from \cref{presburger-circuits-completeness-skolem}: If a function
$f\colon\Z^n\to\Z^m$ is Presburger-definable by some Presburger formula
$\varphi(\bar{x},\bar{y})$, then clearly $f$ is the only possible Skolem
function for $\bar{y}$ in $\forall\bar{x}\exists\bar{y}\colon\varphi(\bar{x},\bar{y})$. Hence, the circuit provided by
\cref{presburger-circuits-completeness-skolem} must compute $f$. Conversely,
given a Presburger circuit $\circuit$, it is easy to construct a Presburger
formula that defines the function $\circuit$ computes.  

\begin{remark}\label{presburger-circuit-minimality}
	In \cref{app:problem},
we also show that if we restrict the division functions to those of the form
$\fdiv_p$ for primes $p$, then (i)~one can still express the same functions and
(ii)~the set of atomic functions is \emph{minimal}. This means, removing any of
the functions $\max$, $\Eq$, or $\fdiv_p$ will result in some
Presburger-definable functions being not expressible. 
\end{remark}
\vspace{-0.3cm}
\paragraph{Presburger functional synthesis, formally}
We are ready to state our main problem of interest. 
\emph{Presburger functional synthesis} ($\PFnS$) is the following task: 
\begin{description}
	\item[Given] A quantifier-free Presburger formula
          $\varphi(\bar{x},\bar{y})$ representing a relational
          specification between $\bar{x}$ and $\bar{y}$.
	\item[Output] A Presburger circuit $\circuit$ that
          computes a Skolem function for $\bar{y}$ in
		$\forall \bar{x} \exists \bar{y}\colon \varphi(\bar{x},\bar{y})$.
\end{description}

For clarity of exposition, we also call the Presburger circuit
referred to above as a \emph{Skolem circuit}.  Intuitively, for every
possible value $\bar{u} \in \mathbb{Z}^n$ of $\bar{x}$, a Skolem circuit
$\circuit$ produces $\circuit(\bar{u}) \in \mathbb{Z}^m$ with the following
guarantee: The relational specification $\varphi(\bar{u},\circuit(\bar{u}))$ is
true iff there is some $\bar{v} \in \mathbb{Z}^m$ for which $\varphi(\bar{u},\bar{v})$
is true. Hence, the value of $\circuit(\bar{u})$ matters only when
$\exists \bar{y}:
\varphi(\bar{u}, \bar{y})$ holds.  If, however, there is no $\bar{v} \in
\mathbb{Z}^m$ with $\varphi(\bar{u},\bar{v})$, then any value produced
by $\circuit(\bar{u})$ is fine.

\begin{restatable}{remark}{definableSkolem}\label{definable-skolem}
	Every Presburger specification admits a Presburger-definable function as a Skolem function.
\end{restatable}
See \cref{app:definable-skolem} for a proof.

\section{Presburger Functional Synthesis for General Formulas}
\label{sec:synthesis}
In this section, we consider Presburger functional synthesis for arbitrary quantifier-free relational specifications. Our main results here are an exponential upper bound, as well as an exponential lower bound.

\paragraph{An exponential upper bound}
Our first main result regarding $\PFnS$ is that, given a quantifier-free formula $\varphi(\bar{x},\bar{y})$, we can synthesize a Skolem function circuit for $\bar{y}$ in $\forall \bar{x} \exists \bar{y} \varphi(\bar{x},\bar{y})$ in exponential time.

\begin{restatable}{theorem}{thmExponentialPresburgerCircuit}\label{thm:exponential-presburger-circuit}
	Given a quantifier-free formula $\varphi(\bar{x},\bar{y})$, there
	exists a Skolem circuit for $\bar{y}$ in $\forall \bar{x} \exists
	\bar{y}\colon \varphi(\bar{x},\bar{y})$.  Moreover, this circuit 
	can be constructed in time $2^{|\varphi|^{\mathcal{O}(1)}}$. 
\end{restatable}

This exponential upper bound result improves significantly on existing methods related to constructing Presburger Skolem functions. The two related lines of work that we are aware of, namely, Presburger functions defined by Ibarra and Leininger~\cite{DBLP:journals/siamcomp/IbarraL81} and translation of Presburger-definable functions into $L_+$-programs by Cherniavsky~\cite[Thm.~5]{DBLP:journals/siamcomp/Cherniavsky76} would yield, respectively, \emph{quadruply-exponential} and \emph{triply-exponential} upper bounds. See \cref{app:comparison} for an analysis.

\Cref{thm:exponential-presburger-circuit} can also be deduced from our normal
form results (i.e.\ by using \cref{construct-circuit-for-semantic-nf} and
either \cref{construct-semantic-nf}, or \cref{construct-syntactic-nf}).
However, we find it instructive to provide a direct proof without conversion
into normal forms.

We now present a sketch of the construction of \cref{thm:exponential-presburger-circuit}. Full details can be found in \cref{app:exponential-presburger-circuit}. The crux of our approach is to use the geometric insight underlying a recent quantifier elimination technique in \cite{DBLP:conf/icalp/HaaseKMMZ24}. This geometric insight refines solution bounds to systems $A\bar{x}\le\bar{b}$ of linear inequalities. Standard bounds provide a solution that is small compared to $\|A\|$ and $\|\bar{b}\|$. The bound from \cite{DBLP:conf/icalp/HaaseKMMZ24} even applies when $\|\bar{b}\|$ itself cannot be considered small. Instead, the result
provides a solution that is ``not far from $\bar{b}$'': The solution can be
expressed as an affine transformation of $\bar{b}$ with small coefficients. 
\newcommand{\binnorm}[1]{\|#1\|_{\mathsf{bin}}}
\newcommand{\fracnorm}[1]{\|#1\|_{\mathsf{frac}}}
To state the result, we need some notation. For a rational number $r\in\Q$, its
\emph{fractional norm} $\fracnorm{r}$ is defined as $|a|+|b|$, where
$\tfrac{a}{b}=r$ is the unique representation with co-prime $a,b$. The
fractional norm of vectors and matrices, written $\fracnorm{A}$ and
$\fracnorm{\bar{x}}$, is then the maximum of the fractional norms of all
entries. The geometric insight is the following, which appeared in \cite[Prop.~4.1]{DBLP:conf/icalp/HaaseKMMZ24}.
\begin{proposition}\label{affine-transformations}
	Let $A\in\Z^{\ell\times n}$ and $\bar{b}\in\Z^\ell$, and let $\Delta$
	be an upper bound on the absolute values of the subdeterminants of $A$.
	If the system $A\bar{x}\le\bar{b}$ has an integral solution, then it
	has an integral solution of the form $D\bar{b}+\bar{d}$, where
	$D\in\Q^{n\times\ell}$, $\bar{d}\in\Q^{n}$ with
	$\fracnorm{D},\fracnorm{\bar{d}}\le n\Delta^2$. 
\end{proposition}
Crucially, the bound $n\Delta^2$ only depends on $A$, not on $\bar{b}$.
By the Hadamard bound for the determinant~\cite{Had93}, this means the
number of bits in the description of $D$ and $\bar{d}$ is polynomial in the
number of bits in $A$.

\begin{proof}[Proof sketch of \cref{thm:exponential-presburger-circuit}]
	(A full proof can be found in \cref{app:exponential-presburger-circuit}.) To apply \cref{affine-transformations}, we first remove modulo constraints in
	$\varphi$, in favor of new output variables. For example, a constraint
	$x_1\equiv a\bmod{b}$ is replaced with $x_1=by'+a$, where $y'$ is a fresh
	output variable. These new output variables can just be ignored in the end, to
	yield a circuit for the original formula. By bringing $\varphi$ into DNF,
	we may assume that $\varphi$ is a disjunction of $r$-many systems of
	inequalities $A_i\bar{y}\le B_i\bar{x}+\bar{c}_i$%
	. Here, $r$ is at most
	exponential, and each $A_i$, $B_i$, and $\bar{c}_i$ has at most polynomially
	many bits.
	
	Now for each $i\in[1,r]$, \cref{affine-transformations} yields $s$-many
	candidate pairs $(D_{i,j},\bar{d}_{i,j})$ for solutions $\bar{y}$ to
	$A_i\bar{y}\le B_i\bar{x}+\bar{c}_i$. Here, $s$ is at most exponential, and we
	know that if the system has a solution for a given $\bar{x}$, then it has one
	of the form $D_{i,j}(B_i\bar{x}+\bar{c}_i)+\bar{d}_{i,j}$ for some
	$j=1,\ldots,s$.

	Our circuit works as follows. The idea is to try for each $(i,j)$, in
	lexicographical order, whether
	$\sigma_{i,j}(\bar{x}):=D_{i,j}(B_i\bar{x}+\bar{c}_i)+\bar{d}_{i,j}$ is
	an integral solution to $A_i\bar{y}\le B_i\bar{x}+\bar{c}_i$. In this
	case, let us say that $(i,j)$ is a \emph{solution}. If $(i,j)$ is a
	solution,  then our circuit outputs $\sigma_{i,j}(\bar{x})$. In order
	to check if $(i,j)$ is a solution, we  need to check two things:
	whether (a)~$\sigma_{i,j}(\bar{x})$ is an integer vector and
	(b)~whether it satisfies $A_i\sigma_{i,j}(\bar{x})\le
	B_i\bar{x}+\bar{c}_i$. Note that (a) is necessary because $D_{i,j}$ and
	$\bar{d}_{i,j}$ are over the rationals. However, we can check
	integrality of $\sigma_{i,j}(\bar{x})$ by way of $\fdiv$ gates. To
	check (b), our circuit computes all entries of the vector
	$B_i\bar{x}+\bar{c}_i-A_i\sigma_{i,j}(\bar{x})$. Using summation, $\max$, and
	$\Eq$ gates, it then computes the number of entries that are $\ge 0$.
	If this number is exactly the dimension of the vector (which can be
	checked with an $\Eq$ gate), $(i,j)$ is a solution.
	
	To implement the lexicographic traversal of all $(i,j)$, we have for
	each $(i,j)\in[r,s]$ a circuit that computes the function
	$F_{i,j}(\bar{x})$, which returns $1$ if and only if (i)~$(i,j)$ is a
	solution, and (ii)~for all $(r,s)$ that are lexicographically smaller
	than $(i,j)$, the pair $(r,s)$ is not a solution. Based on this, we can
	compute the function $S_{i,j}(\bar{x})$, which returns
	$\sigma_{i,j}(\bar{x})$ if $(i,j)$ is a solution, and zero otherwise.
	Note that $S_{i,j}(\bar{x})$ is non-zero for at most one pair
	$(i,j)$.  Finally, we define $f(\bar{x})$ to sum
	up $S_{i,j}(\bar{x})$ over all $(i,j)\in[1,r]\times[1,s]$.
	Then, $f$ is clearly a Skolem function for $\varphi$.
\end{proof}
\vspace{-0.3cm}
\begin{remark}
	Our construction even yields a circuit of polynomial depth,
	and where all occurring coefficients (in linear combination gates) have
	at most polynomially many bits. %
\end{remark}

\paragraph{An exponential lower bound}
The second main result of this section is a matching exponential lower bound.
\begin{restatable}{theorem}{thmExpLowerBound}\label{thm:exp-lower-bound}
	There are quantifier-free formulas $(\mu_n)_{n\ge 0}$ such that
	any Skolem circuit for $\mu_n$ has size at
	least $2^{\Omega(n)}$
\end{restatable}

\newcommand{\PTIME}{\mathsf{P}}
\newcommand{\Ppoly}{\mathsf{P}/\mathsf{poly}}
Let us point out that usually it
is extremely difficult to prove lower bounds for the size of circuits. Indeed,
proving an (unconditional) exponential lower bound for the size of circuits for
Boolean functional synthesis is equivalent to one of the major open problems in
complexity theory---whether the class $\NP$ is included in $\Ppoly$ (which, in
turn, is closely related to whether $\PTIME$ equals $\NP$):
\newcommand{\impl}[2]{``\labelcref{#1}$\Rightarrow$\labelcref{#2}''}
\newcommand{\implc}[2]{``#1$\Rightarrow$#2''}
\begin{restatable}{observation}{booleanCircuitSize}\label{obs:Bool:short}
The following are equivalent:
		(i)~Every Boolean formula $\varphi$ has a Skolem function computed by a Boolean circuit of size polynomial in $|\varphi|$.
			(ii)~$\NP\subseteq\Ppoly$.
\end{restatable}
Here, $\Ppoly$ is the class of all problems solvable in polynomial time with a
polynomial amount of advice (see e.g.,~\cite{AB}).  The implication 
\implc{(i)}{(ii)} had been shown in \cite[Theorem 1]{AkshayCGKS21}. We prove \implc{(ii)}{(i)} in \cref{app:Bool}.

Nevertheless, we will prove an exponential lower bound for circuits for
Presburger functional synthesis.
\paragraph{Proof sketch exponential lower bound}
For space reasons, we can only provide a rough sketch---with a full proof
in \cref{app:exp-lower-bound}.  Following a construction from \cite[Section
6]{DBLP:conf/icalp/HaaseKMMZ24}, we choose $\mu_n\equiv \forall
x\colon\exists\bar{y}\colon\psi_n(x,\bar{y})$ so that the formula
$\exists\bar{y}\colon\psi_n(x,\bar{y})$ defines a subset $S\subseteq\Z$ whose
minimal period is doubly exponential. Here, the \emph{minimal period} of $S$,
is the smallest $p$ so that for all but finitely many $u$, we
have $u+p\in S$ if and only if $u\in S$.  Moreover, we show that if $\mu_n$ had
a Skolem function with a Presburger circuit $\circuit_n$ with $e_n$-many
$\fdiv$-gates, and $M$ is the least common multiple of all divisors occurring
in that gate, then $p$ must divide $M^{e_n}$.  This proves that $e_n$ is at least
exponential, hence $\circuit_n$ must contain an exponential number
of $\fdiv$-gates.

\vspace{-0.4cm}

\paragraph{Hardness for one input, one output}

We now show that synthesizing Skolem functions for Presburger specifications with even just one input and one output variable is as hard as the general Boolean functional synthesis problem:
\begin{restatable}{observation}{hardnessOneOneLong}\label{hardness-one-one-long}
	Suppose every one-input one-output quantifier-free Presburger formula
	has a polynomial size Skolem circuit. Then every Boolean formula has a
	polynomial size Skolem circuit---impossible unless $\NP\subseteq\Ppoly$.
\end{restatable}
This follows using the ``Chinese Remaindering'' technique, by which one can encode an assignment of $n$ Boolean variables in a single integer: in the residues modulo the first $n$ primes. See \cref{app:hardness-one-one-long} for details.

\section{Semantic Normal Form for $\PFnS$}
\label{sec:semantic-nf}
We now present a semantic normal form for $\PA$ specifications, called
$\PASyn$, that guarantees efficient Skolem function synthesis. The
normal form definition has two key ingredients, (i)~modulo-tameness
and (ii)~local quantification.

\paragraph*{Ingredient I: Modulo-tameness}
Recall that we represent quantifier-free $\PA$ formulas as trees.  A $\land$-labeled node in the tree representing $\varphi$ is said to be a \emph{maximal
conjunction} if there are no $\land$-labeled ancestors of the node in
the tree.  A subformula is \emph{maximal conjunctive}\label{maximal-conjunctive} if it is the sub-formula rooted at some maximal
conjunction. %

\begin{definition}\label{def:y-mod-tame}
A quantifier-free $\PA$ formula $\varphi(\bar{x},y)$ is called
\emph{$y$-modulo-tame}, if it is in NNF, and for every maximal conjunctive sub-formula
$\psi$ of $\varphi$, there is an integer $M^\psi$ such that all modulo constraints
involving $y$ in $\psi$ are of the form $y\equiv r\pmod {M^\psi}$
for some $r\in[0,M^\psi-1]$.
\end{definition}
Hence, the definition admits $y \equiv r_1 \pmod M$ and $y \equiv
r_2 \pmod M$ in the same maximal conjunctive sub-formula, even if $r_1
\neq r_2$. It does not admit $y \equiv r_1 \pmod {M_1}$ and
$y \equiv r_2 \pmod {M_2}$ in a maximal conjunctive sub-formula, if
$M_1 \neq M_2$.  The value of $M$ can vary from one maximal
conjunctive sub-formula to another; so the definition admits $y \equiv
r_1 \pmod {M_1}$ and $y \equiv r_2 \pmod {M_2}$ in sub-trees rooted at
two different maximal conjunctions.

As an example, the formula represented in Fig.~\ref{fig:tree-rep} is
not $y$-modulo tame.  This is because the maximal conjunctive
sub-formula to the left of the root has the atomic formula $4x + 5y
\equiv 2 \pmod 3$, which is not of the form $y \equiv r \pmod M$.
However, if we replace $4x + 5y \equiv 2 \pmod 3$ by the
semantically equivalent formula $\big((y \equiv 0 \pmod 3 \,\land\, 4x
\equiv 2 \pmod 3) \lor (y \equiv 1 \pmod 3 \,\land\, 4x \equiv 0 \pmod
3) \lor (y \equiv 2 \pmod 3 \,\land\, 4x \equiv 1 \pmod 3)\big)$,
then every maximal conjunctive sub-formula in the new formula
satisfies the condition of Definition~\ref{def:y-mod-tame}.  Hence,
the resulting semantically equivalent formula is $y$-modulo tame.

Checking if a given formula $\varphi(\bar{x},y)$ is $y$-modulo-tame is
easy: look at each maximal conjunction in the tree
representation of $\varphi$ and check if all modulo constraints
involving $y$ are of the form $y \equiv r \pmod M$ for the same
modulus $M$. Furthermore, this form is universal: any formula can be
made $y$-modulo tame for any $y$, albeit at the cost of a worst-case
exponential blow-up (with proof in
\cref{app:ModTameUniv}):
\begin{restatable}{proposition}{propModTameUniv}\label{prop:ModTameUniv}
Given a quantifier-free formula $\varphi(\bar{x},y)$, let
  $\mathfrak{M}$ be the set of all moduli appearing in
  modular constraints involving $y$.  We can construct an 
  equivalent $y$-modulo-tame formula in
  $\bigO{|\varphi|.\big(\prod_{M \in \mathfrak{M}} M\big)}$ time.
\end{restatable}
Since the moduli $M$'s in $\varphi$ are represented in binary,
Proposition~\ref{prop:ModTameUniv} implies an exponential blow-up in
the formula size, when making it $y$-modulo tame.  This blow-up is
however unavoidable, by virtue of the hardness result in
Observation~\ref{hardness-one-one-long} and a key result of this
section (Theorem~\ref{presburger-circuit-one-output}).

\paragraph{Ingredient II: Local quantification}
For $\PASyn$, we also need the concept of local quantification, which we introduce now.
For a quantifier-free $\varphi(\bar{x},y)$ in NNF and $y$-modulo-tame, we define $\lexists y\colon\varphi(\bar{x},y)$ as the formula obtained by replacing each atomic subformula in $\varphi$ that mentions $y$ with $\top$.
Clearly, $\exists y\colon\varphi(\bar{x},y)$ implies $\lexists y\colon \varphi(\bar{x},y)$.
\paragraph{Definition of $\PASyn$}
Suppose $\varphi(\bar{x},\bar{y})$ is a quantifier-free Presburger
formula in NNF with free variables $\bar{x}=(x_1,\ldots,x_n)$ and
$\bar{y}=(y_1,\ldots,y_m)$. We define {\em $\varphi$ to be in $\PASyn$ w.r.t.\ the ordering $y_1\preceq\cdots\preceq y_m$}, if (i)~$\varphi$ is $y_i$-modulo-tame for each $i\in[1,m]$ and (ii)~for every $i\in[1,m-1]$, the formula
\begin{align}
  \forall \bar{x}\forall y_1,\ldots,y_i\colon (\lexists y_{i+1},\ldots,y_m\colon\varphi(\bar{x},\bar{y})\to\nonumber\\
	\exists y_{i+1},\ldots,y_m\colon \varphi(\bar{x},\bar{y}))\label{condition-local-quantification}
\end{align}
denoted $\varphi^{(i)}$, holds.
In the following, we assume that every specification formula is annotated with
an ordering on the output variables, and that $\PASyn$ is w.r.t.\ that
ordering.

To see an example of a $\PASyn$ specification, consider a variant of
the job scheduling problem discussed in
Section~\ref{sec:introduction}. In this variant, we have only two
items, and we require item $2$ to be (pre-)processed before item $1$.
The variant specification is
$\psi \equiv \psi_1 \wedge \psi_2 \wedge \psi_3$, where
$\psi_1 \equiv \bigwedge_{i=1}^2(t_i + \delta_i + 1 \equiv 1 \pmod
2)$, $\psi_2 \equiv t_2 + \delta_2 + 1 < t_1 + \delta_1$ and
$\psi_3 \equiv \bigwedge_{i=1}^2(\delta_i \ge 0) \wedge (\delta_1 +
2\delta_2 \leq \Delta)$.  It is easy to see that $\psi$ is not
$\delta_i$-modulo tame for any $\delta_i$; hence it is not in
$\PASyn$. If we replace $\psi_1$ with the equivalent formula
$\psi_1' \equiv \bigwedge_{i=1}^2 \bigvee_{r=0}^1 \big((\delta_i \equiv
r \pmod 2)\wedge (t_i \equiv r \pmod 2)\big)$, the resulting
specification $\psi' \equiv \psi_1' \wedge \psi_2 \wedge \psi_3$ is
$\delta_i$-modulo tame for $i \in \{1, 2\}$.  However $\psi'$ is not
in $\PASyn$ w.r.t. any ordering of $\delta_1, \delta_2$, since local
quantifier elimination replaces the constraint $\delta_1 +
2\delta_2 \le \Delta$ with $\top$, removing the cap on the cumulative
weighted delay. To remedy this situation, consider
$\psi'' \equiv \psi_1' \wedge \psi_2 \wedge \psi_3 \wedge
(\psi_4 \vee \psi_5)$, where $\psi_4 \equiv (t_2 + \delta_2 + 1 <
t_1) \wedge \bigvee_{r=0}^1 \big( (t_1 \equiv r \pmod 2) \wedge
(2\delta_2 + r \le \Delta)\big)$, and $\psi_5 \equiv (t_2 + \delta_2 +
1 \ge t_1) \wedge (t_2 - t_1 + 3\delta_2 + 2\le \Delta)$. It can be
verified that $\psi''$ is semantically equivalent to $\psi$, and
satisfies all conditions for $\PASyn$ w.r.t. the order
$\delta_1 \prec \delta_2$ (but not w.r.t.
$\delta_2 \prec \delta_1$). Note that $(\psi_4 \vee \psi_5)$
constrains $t_1, t_2, \delta_2, \Delta$ in such a way that
$(\psi_4 \vee \psi_5) \wedge \lexists \delta_1 \psi ~\leftrightarrow~
\exists \delta_1 \psi$ holds.

\paragraph{Main results about $\PASyn$.}
The first main result is that for formulas in $\PASyn$, we can easily
solve $\PFnS$.
\begin{theorem}\label{construct-circuit-for-semantic-nf}
	Given a Presburger formula $\varphi(\bar{x},\bar{y})$ in $\PASyn$, we
        can compute in time polynomial in the size of $\varphi$, a
        Skolem circuit for each $y_i$ in
        $\forall \bar{x} \exists \bar{y}\, \varphi(\bar{x},\bar{y})$.
\end{theorem}
Second, every formula has an equivalent in $\PASyn$, albeit with an exponential blow-up
(unavoidable by Thm.~\ref{construct-circuit-for-semantic-nf}, \ref{thm:exp-lower-bound}):
\begin{theorem}\label{construct-semantic-nf}
	For every quantifier-free formula $\varphi(\bar{x},\bar{y})$, there is an equivalent
	formula $\psi$ in $\PASyn$, such that $\psi$ is at most 
	exponential in the size of $\varphi$.
\end{theorem}

As a third important result, we show that checking whether a formula is in $\PASyn$ has reasonable complexity:
\begin{theorem}\label{check-synthesis-form}
	Given a quantifier-free formula $\varphi(\bar{x},\bar{y})$ in NNF, it is $\coNP$-complete to decide whether $\varphi$ is in $\PASyn$.
\end{theorem}

Finally, we have a corollary of independent interest:
\begin{corollary}\label{exists-star-forall}
	The $\exists^*\forall$ fragment of formulas over the structure $\langle
	\Z;+,<,0,1\rangle$ is $\NP$-complete.
\end{corollary}
In \cref{exists-star-forall}, it is crucial that the input formula is over the
structure $\langle\Z;+,<,0,1\rangle$, meaning it cannot contain modulo
constraints.  Indeed, a reduction similar to \cref{hardness-one-one-long} shows that
with modulo constraints, even the $\exists\forall$ fragment is
$\Sigma_2^p$-hard. \cref{exists-star-forall} is somewhat surprising, since the $\forall\exists^*$ fragment of $\langle\Z;+,<,0,1\rangle$ is
$\coNEXP$-complete~\cite[Thm.~1]{DBLP:conf/csl/Haase14} (the lower bound was
already shown in \cite[Thm.~4.2]{DBLP:journals/apal/Gradel89}). Hence, in this setting, allowing an unbounded number of inner quantifiers is more expensive than allowing an unbounded number of outer quantifiers.
Furthermore, \cref{exists-star-forall} complements a result of \citeauthor{DBLP:journals/mst/Schoning97} (\citeyear{DBLP:journals/mst/Schoning97}, Corollary), which states that the $\exists\forall$ fragment for the structure $\langle\Z;+,<,0,1\rangle$ is $\NP$-complete: Together, \cref{exists-star-forall} and Sch\"{o}ning's result imply that for every $i\ge 1$, the fragment $\exists^i\forall$ is $\NP$-complete.

The remainder of this section is devoted to proving \cref{construct-circuit-for-semantic-nf,construct-semantic-nf,check-synthesis-form,exists-star-forall}.
Of these proofs, \cref{construct-circuit-for-semantic-nf} is the most involved. It is shown in two steps: First, we prove \cref{construct-circuit-for-semantic-nf} in the case of one output variable (i.e.\ $m=1$). Then, we show that this procedure can be used repeatedly to solve $\PFnS$ in general in polynomial time.
\paragraph{The case of one output}
We first prove \cref{construct-circuit-for-semantic-nf} when $m=1$. In this setting, $\PASyn$ is equivalent to modulo-tameness w.r.t.\ the only output variable.
\begin{theorem}\label{presburger-circuit-one-output}
Let $\varphi(\bar{x},y)$ be a $y$-modulo-tame quantifier-free $\PA$
formula.  A Skolem circuit for $y$ in $\forall \bar{x} \exists
y: \varphi(\bar{x}, y)$ can be computed in time polynomial in
$|\varphi|$.
\end{theorem}
Below, we give an outline of the proof of
\cref{presburger-circuit-one-output}, leaving the details to
\cref{app:presburger-circuit-one-output}.
\subparagraph{Step I: Simplify modulo constraints}
We restrict ourselves to the case of $\varphi$ being conjunctive
(i.e.\ its top-most connective is a conjunction): else, one can first
compute a Skolem circuit for each maximal conjunctive subformula, and
then easily combine these circuits into a Skolem circuit for
$\varphi$. Since $\varphi$ is modulo-tame, there is an $M\in\N$ such
that all modulo constraints on $y$ in $\varphi$ are of the form
$y\equiv r\pmod M$ for some $r\in\N$. Let $R$ denote the set of all
such $r$; clearly, $|R|\le|\varphi|$. Now, it suffices to construct a
Skolem circuit $\circuit_r$ for each formula $\varphi_r:= \big(\varphi\wedge
y\equiv r\pmod{M}\big)$ for $r\in R$: This is because from these $|R|$
circuits, we can easily construct one for $\varphi$: Just compute
$\circuit_r(u)$ for each $r\in R$, and check
whether $\varphi(u,\circuit_r(u))$ holds; if it does, then output
$\circuit_r(u)$ (if no $\circuit_r(u)$ works, then the output can be
arbitrary).

However, $\varphi\wedge y\equiv r\pmod{M}$ is equivalent to a formula
$\varphi'\wedge y\equiv r\pmod{M}$, where $\varphi'$ contains no
modulo constraints on $y$. Indeed, a modulo constraint on $y$ in $\varphi$ is
either consistent with $y\equiv r\pmod{M}$ and can be replaced with
$\top$, or it is inconsistent with $y\equiv r\pmod{M}$ and can be
replaced with $\bot$.  Thus, we may assume that our input formula is
of the form $\varphi\wedge y\equiv r\pmod{M}$, where $\varphi$
is \emph{$y$-modulo-free}, meaning $\varphi$ contains no modulo
constraints on $y$.
\paragraph{Step II: Computing interval ends}
First, note that for any $\bar{u}\in\Z^n$, the set $V_{\bar{u}}$ of
all $v\in\Z$ for which $\varphi(\bar{u},v)$ holds can be represented
as a finite union of intervals. This is because $\varphi(\bar{x},y)$
has no modulo constraints on $y$, and thus every atomic formula is an
inequality that either yields (i)~an upper bound or (ii)~a lower bound
on $y$, given a value of $\bar{x}$.

Next, we \emph{construct Presburger circuits that compute the ends of
these finitely many intervals}. Once we do this, it is easy to
construct a Skolem circuit for $\varphi\wedge y\equiv r\pmod{M}$: For
each interval in some order, check (using $\fdiv_M$) whether it
contains a number $\equiv r\pmod{M}$, and if so, output one.

\newcommand{\Powerset}[1]{2^{#1}}
Let us describe more precisely how a circuit computes the interval
union $V_{\bar{u}}$.  An \emph{interval-computing circuit} is a
Presburger circuit $\circuit$ that computes a function $\Z^n\to
(\Z\times\Z)^{k+2}$ for some $k\in\N$. It induces a function
$F_\circuit\colon \Z^n\to\Powerset{\Z}$ as follows. If
$\circuit(\bar{u})=\langle
r_0,s_0,r_1,s_1,\ldots,r_{k+1},s_{k+1}\rangle$ for some
$\bar{u}\in\Z^n$, then we set
$F_\circuit(\bar{u}):=I~\cup~J_1~\cup~\cdots~\cup~J_k~\cup~K$, where
$J_i$ is the closed interval $[r_i,s_i]=\{v\in\Z \mid r_i\le v\le
s_i\}$; and $I$ is the left-open interval $(-\infty,s_0]$ if $r_0=1$
and $I=\emptyset$ if $r_0\ne 1$; and $K$ is the right-open interval
$K=[r_{k+1},\infty)$ if $s_{k+1}=1$ and $K=\emptyset$ if $s_{k+1}\ne
1$. Thus, while $r_1, s_1, \ldots r_k, s_k$ represent end-points of
$k$ (possibly overlapping) intervals, $r_0$ (resp.\ $s_{k+1}$)
serves as a flag indicating whether the left-open interval $(-\infty,
s_0]$ (resp.\ right-open interval $[r_{k+1}, \infty)$) is to be
included in $V_{\bar{u}}$.

For a formula $\varphi(\bar{x},y)$ with one output $y$ and no
modulo-constraints on $y$, we say that an interval-computing circuit $\circuit$
is \emph{equivalent to $\varphi$} if for every $\bar{u}\in\Z^n$ and every $v\in\Z$,
$\varphi(\bar{u},v)$ holds \emph{if and only if} $v\in F_\circuit(\bar{u})$.
The most technical ingredient in our construction is to show:
\vspace{-0.1cm}
\begin{claim}\label{comput-interval-endpoints}
	Given a quantifier-free $y$-modulo-free Presburger formula, we can
	compute in polynomial time an equivalent interval-computing circuit.
\end{claim}
\begin{proof}[Proof sketch]
	We build the circuit by structural induction, beginning with atomic
	formulas. Each atomic formula imposes either a lower bound or an upper
	bound on $y$, which can be computed using linear functions and $\fdiv$.
	For example, if the formula is $-x_1+3x_2+5y\ge 0$, then this is
	equivalent to $y\ge\tfrac{1}{5}(x_1-3x_2)$, and thus we compute
	$\fdiv_5(x_1-3x_2)$ as the only lower bound.

	Building the circuit for a disjunction
	$\varphi_1\vee\varphi_2$ is easy: Starting from circuits
	$\circuit_1$ and $\circuit_2$, we simply output all the closed
	intervals output by each circuit.  The open intervals output
	by the circuits are combined slightly differently depending on
	the values of the flags. For example, if $\circuit_1(u)$ and
	$\circuit_2(u)$ include intervals $[t_1,\infty)$ and
	$[t_2,\infty)$, then the new circuit will produce the interval
	$[\min(t_1,t_2),\infty)$.

	The difficult step is to treat conjunctions
	$\varphi_1\wedge\varphi_2$.  Here, we follow a strategy
	inspired from sorting networks~\cite{CLRS,AKS83}
	to \emph{coalesce-and-sort} the intervals output by each of
	$\circuit_1$ and $\circuit_2$. A
	basic \emph{coalesce-and-sort} gadget takes as input a pair of
	(possibly overlapping) intervals $[r, s]$ and $[r', s']$, and
	coalesces them into one interval if they overlap; otherwise it
	leaves them unchanged.  The gadget outputs two disjoint
	intervals $[t, u]$ and $[t', u']$, with $[t, u]$ ``ordered
	below'' $[t',u']$, such that $[t,u] \cup [t',u'] = [r,s] \cup
	[r',s']$, and either $[t, u] = \emptyset$ or $u < t'$.  Thus,
	empty intervals are ordered below non-empty ones, and
	non-empty intervals are ordered by their end-points.  A
	coalesce-and-sort network is a sorting network built using
	these gadgets. If $\circuit_i$ outputs $q_i$ (possibly
	overlapping) intervals, feeding these to a coalesce-and-sort
	network yields at most $q_i$ disjoint sorted intervals.  The
	interval-computing circuit for $\varphi_1\wedge\varphi_2$ now
	computes the $q_1q_2$ pairwise intersections of these disjoint
	intervals, coalesce-and-sorts the resulting intervals, and
	returns the $\max(q_1, q_2)$ intervals at the top of the
	sorted order. This is sound because intersecting the union of
	$q_1$ disjoint intervals with the union of some other $q_2$
	disjoint intervals yields at most $\max(q_1, q_2)$
	non-empty disjoint intervals.

	To keep the interval-computing circuit size under check, our
	construction maintains carefully chosen size
	invariants. Specifically, we ensure that the number of
	interval endpoints at the output of, and indeed the total size
	of the interval-computing circuit for
	$\varphi_1 \vee \varphi_2$ or $\varphi_1 \wedge \varphi_2$ is
	always bounded by a polynomial in $|\varphi_1| + |\varphi_2|$.
	Intuitively, since each interval endpoint at the output of the
	interval-computing circuit must originate from an atomic
	formula at a leaf in the tree representation of the
	specification, there are atmost a polynomial number of
	interval endpoints to track.

        The reader is referred to \cref{app:semantic-nf} for
        details of the proof, and a pictorial depiction.
\end{proof}
\paragraph{The case of multiple output variables}
It remains to prove \cref{construct-circuit-for-semantic-nf} in the general case (i.e.~$m\ge 1$).
\begin{proof}[Proof of \cref{construct-circuit-for-semantic-nf}]
 Let $\widehat{\varphi}^{(i)}$ denote $\exists y_{i+1} \ldots y_m\colon$ $\varphi(\bar{x}, \bar{y})$, for $i \in \{1, \ldots m-1\}$.  For each $i$ in $m$ down to $1$, we obtain a Presburger circuit for a Skolem function $f_i$ for $y_i$ in $\forall \bar{x} \forall y_1, \ldots y_{i-1} \exists y_i\colon\widehat{\varphi}^{(i)}(\bar{x},y_1, \ldots y_i)$ using Theorem~\ref{presburger-circuit-one-output} for single output specifications. Each such $f_i$ expresses $y_i$ in terms of $\bar{x}$ and $y_1, \ldots y_{i-1}$.  It is easy to see that by composing the resulting Presburger circuits, we can obtain Presburger circuits for Skolem functions for all $y_i$'s in $\forall \bar{x} \exists \bar{y}\colon\varphi(\bar{x}, \bar{y})$.  Each resulting Skolem function is of course expressed only in terms of $\bar{x}$. 
\end{proof}
\paragraph{Achieving $\PASyn$}
We now prove \cref{construct-semantic-nf}.
using (either of) the recent QE procedures.%
\begin{proof}[Proof of \cref{construct-semantic-nf}]
	For each $i\in[1,m-1]$, let $\psi_i$ be a quantifier-free equivalent to 
	$\exists y_{i+1},\ldots,y_{m}\colon \varphi(\bar{x},\bar{y})$.
	By recent results on
	quantifier-elimination~\cite{DBLP:conf/icalp/0001MS24,DBLP:conf/icalp/HaaseKMMZ24},
	we can obtain such a $\psi_i$ of at most exponential size in
	$|\varphi|$. The formula
        $\eta=\varphi\wedge \bigwedge_{i\in[1,m-1]} \psi_i$%
	is equivalent to $\varphi$, and satisfies the equivalence condition
	regarding local and global quantification. It remains to achieve
	modulo-tameness. For this, we notice that both recent QE procedures,
	\cite[Thm.~3]{DBLP:conf/icalp/0001MS24} and
	\cite[Thm.~3.1]{DBLP:conf/icalp/HaaseKMMZ24} produce an exponential
	disjunction of polynomial-sized formulas. We may thus assume that
	$\psi_i=\bigvee_{j=1}^s\psi_{i,j}$ for some exponential $s$ for each
	$i\in[1,m-1]$. We can now write $\eta$
        equivalently as
	$\bigvee_{f\in F} \left(\varphi\wedge\bigwedge_{i\in[1,m-1]} \psi_{i,f(i)}\right)$, where $F$ is the set of functions $f\colon [1,m-1]\to[1,s]$. Observe that each formula $\tau_f:=\varphi\wedge\bigwedge_{i\in[1,m-1]}
	\psi_{i,f(i)}$ is of polynomial size, and thus product $M$ of all
	moduli ocurring in $\tau_f$ is at most exponential.  We thus
	rewrite all modulo constraints in $\tau_f$ for variables $y_k$ w.r.t.\
	$M$, yielding an exponential-sized equivalent of $\tau_f$ which is
	$y_k$-modulo-tame for all $k$. The resulting formula has at most
	exponential size and is in $\PASyn$.
\end{proof}
\paragraph{Checking $\PASyn$: Step I}
Finally, we prove \cref{check-synthesis-form}. We begin with an auxiliary result:
\begin{theorem}\label{modulo-tame-conp}
	Given a $y$-modulo-tame quantifier-free formula $\varphi(\bar{x},y)$,
it is $\coNP$-complete to decide whether $\forall\bar{x}\exists
y\colon\varphi(\bar{x},y)$ holds. 
\end{theorem}
Note that \cref{modulo-tame-conp} implies \cref{exists-star-forall}, since a
formula over the signature $\langle\Z;+,<,0,1\rangle$ is automatically
$y$-modulo-tame for each variable $y$. 
\begin{proof}[Proof of \cref{modulo-tame-conp}]
	Since $\varphi$ is
	$y$-modulo-tame, \cref{presburger-circuit-one-output} allows
	us to compute a polynomial-sized Presburger circuit $\ckt$
	that computes a Skolem function $f$ for $y$ in
	$\forall \bar{x} \exists y\colon\varphi$. %
        Now, we can build a
	polynomial-sized circuit $\ckt'$ for the function $g$ with
	$g(\bar{x})=1$ if $\varphi(\bar{x},f(\bar{x}))$, and
	$g(\bar{x})=0$ otherwise (see \Cref{prop:pa-formula-as-function} for details). Then, we have $\forall\bar{x}\exists
	y\colon\varphi(\bar{x},y)$ if and only if the circuit $\ckt'$
	returns $1$ true for every vector $\bar{x}$. Equivalently,
	$\forall\bar{x}\exists y\colon\varphi(\bar{x},y)$ does not
	hold if and only if there is $\bar{x}$ such that $\ckt'$
	evaluates to $0$.
	The existence of such an $\bar{x}$ can be decided in $\NP$ by a
	reduction to existential Presburger arithmetic: Given $\ckt'$, we introduce
	a variable for the output of each gate, and require 
	that (i)~each gate is evaluated correctly and (ii)~the circuit outputs $0$.
\end{proof}
\paragraph{Checking $\PASyn$: Step II}
We can now show Thm.~\ref{check-synthesis-form}:
\begin{proof}[Proof of \cref{check-synthesis-form}]
	We can clearly check whether $\varphi$ is in NNF and whether $\varphi$ is $y_i$-modulo-tame for every $i\in[1,m]$.
	It remains to check whether $\varphi^{(i)}$ in eq.~\eqref{condition-local-quantification} holds for every $i\in[1,m-1]$.
	This is the case iff each formula 
	\begin{align*} \varphi^{\dagger(i)}:=\forall \bar{x}&\forall y_1,\ldots,y_i\colon (\lexists y_{i+1},\ldots,y_m\colon \varphi(\bar{x},\bar{y})\\
          &\to \exists y_{i+1}\colon \lexists y_{i+2},\ldots,y_m\colon \varphi(\bar{x},\bar{y}))
          \end{align*}
	holds for $i$ in $m-1$ down to $1$.  Indeed, since we know
	from $\varphi^{\dagger(m-1)}=\varphi^{(m-1)}$ that $y_m$
	can be eliminated locally, we can plug that equivalence into
	$\varphi^{(m-2)}$ to obtain $\varphi^{\dagger(m-2)}$. By repeating this
	argument, we can see that the conjunction of all $\varphi^{\dagger(i)}$
	implies the conjunction of all $\varphi^{(i)}$.

	Note that $\varphi^{\dagger(i)}$ belongs to the $\forall^*\exists$
	fragment, and the formula is modulo-tame w.r.t.\ the existentially
	quantified variable. By \cref{modulo-tame-conp}, we can decide the
	truth of $\varphi^{\dagger(i)}$ in $\coNP$.
	For $\coNP$-hardness, note that an NNF formula $\varphi$
	with free variables in $\bar{x}$ is in $\PASyn$ if
	and only if $\forall \bar{x}\colon\varphi(\bar{x})$.  Moreover,
	universality for NNF formulas is $\coNP$-hard.
\end{proof}

Our final result in this section is that the $\PASyn$ normal form is ``optimal'' for existential quantification and synthesis for single-output
modulo-tame specifications. Specifically,
\begin{restatable}{theorem}{thmOptimal}
\label{thm:optimal}
  Let $\mathfrak{S}$ be a class of quantifier-free $\PA$-formulas in NNF on 
  free variables $\bar{x}$ and $y$ such that:
  \begin{enumerate}
  \item  $\mathfrak{S}$ is universal, i.e. for every quantifier-free $\PA$
  formula $\psi(\bar{x},y)$, there is a semantically equivalent formula in  $\mathfrak{S}$
  \item For every formula $\varphi(\bar{x}, y)$ in $\mathfrak{S}$,
  \begin{itemize}
  \item $\varphi$ is $y$-modulo tame, and
  \item There is a poly (in $|\varphi|$) time algorithm for
  computing a quantifier-free formula equivalent to $\exists y: \varphi(\bar{x},y)$.
  \end{itemize}
\end{enumerate}
Then there exists a poly (in $|\varphi|$) time algorithm that
compiles $\varphi(\bar{x},y) \in \mathfrak{S}$ to $\varphi'$ that is in $\PASyn$ wrt $y$.
\end{restatable}
The proof is in \cref{app:optimal}. Assumption 2) is weaker than requiring $\PFnS$ to be efficiently solvable. This is due to the difference in vocabulary between Presburger formulas and circuits (unlike the Boolean case). 

\section{Syntactic Normal Form for $\PFnS$}
\label{sec:syntactic-nf}

  We now present a \emph{syntactic normal form} for $\PFnS$. This means, it has three properties:
  \begin{enumerate*}[(i)]
	  \item It is \emph{syntactic}, meaning one can check in polynomial time whether a given formula is in this normal form,
	  \item it facilitates $\PFnS$, meaning for formulas in normal form, $\PFnS$ is in polynomial time, and
	  \item every formula can be brought into normal-form (and even in exponential time). \end{enumerate*}
  We call our normal form $\DNNF$.
  \vspace{-0.3cm}
  \paragraph{Definition of $\DNNF$} Recall that an \emph{affine transformation}
  (from $\Q^k$ to $\Q^\ell$) is a map $\Q^k\to\Q^\ell$ of the form
  $\bar{x}\mapsto Bx+\bar{b}$, where $B\in\Q^{k\times\ell}$ is a $k\times\ell$
  matrix over $\Q$ and $\bar{b}\in\Q^\ell$ is a vector in $\Q^\ell$. In
  particular, the affine transformation is described by the entries of $B$ and
  $\bar{b}$.
   Consider a quantifier-free PA formula
  $\varphi(\bar{x},\bar{y})$, $\bar{x}=(x_1,\ldots,x_n)$,
  $\bar{y}=(y_1,\ldots,y_m)$. To simplify notation, we define for any vector $(\bar{u},\bar{v})\in\Z^{n+m}$ with $\bar{u}\in\Z^n$, $\bar{v}\in\Z^m$:
\begin{align*}
	\bar{u}^{i}&:=(u_1,\ldots,u_n,v_1,\ldots,v_i), &
	\bar{v}^i&:=(v_{i+1},\ldots,v_m).
\end{align*}
  The idea of $\DNNF$ is to encode Skolem functions in the
  formula: Each maximal conjunctive subformula (see \cref{maximal-conjunctive}) is annotated
  with affine transformations $A_1,\ldots,A_m$, where $A_i$ could serve as
  Skolem functions for this subformula, when $\bar{y}^i$ are considered as
  output variables. This can be viewed as an analogue of the wDNNF in the
  Boolean setting~\cite{AkshayCGKS21}, where each maximal conjunctive subformula provides
  for each output variable a truth value for a Skolem function. Instead
  of concrete truth values, $\DNNF$ has affine transformations in $\bar{x}^i$.

  We say that $\varphi$ is in \emph{$\DNNF$} (``syntactic synthesis normal form'') if 
for every maximal conjunctive subformula $\varphi'$, there exists an $M\in\Z$ and for every $i\in[1,m]$, there exists an affine transformation $A_i\colon\Q^{n+i}\to\Q^{m-i}$ such that (a)~$\varphi$ is $y_i$-modulo-tame for every $i\in[1,m]$ and (b)~every denominator in the coefficients in $A_i$ divides $M$ and (c)~$\varphi'$ is a positive Boolean combination of formulas of the form
	\begin{multline}
		\left(\psi(\bar{x},\bar{y})\vee\bigvee_{i=0}^m \bar{y}^{i}=A_i(\bar{x}^{i})\right)\wedge  \\
		\bigwedge_{i=0}^m \psi(\bar{x}^{i},A_i(\bar{x}^{i}))\wedge (\bar{x},\bar{y})\equiv (\bar{r},\bar{s})\pmod{M}, \label{piece-dnnf}
	\end{multline}

where $\psi(\bar{x},\bar{y})$ is an atomic formula and where $A_i(\bar{r}^i)\in\Z^{m-i}$. Note that assuming (b) and $(\bar{x},\bar{y})\equiv (\bar{r},\bar{s})\pmod{M}$, the condition $A_i(\bar{r}^i)\in\Z^{m-i}$ is equivalent to $A_i(\bar{x}^i)\in\Z^{m-i}$ (see \cref{app:syntactic-nf-preserve-integrality}).
\vspace{-0.3cm}
\paragraph{Properties of $\DNNF$} 
Let us now show that the $\DNNF$ indeed has the properties (i)--(iii) above. 

First, one can easily check (in polynomial time) whether a formula is in $\DNNF$: In each parenthesis, the disjunction over $\bar{y}^i=A_i(\bar{x}^i)$ means the formula explicitly contains all coefficients of the affine transformation $A_i$, for every $i\in[1,m]$. Once these are looked up, one can verify that the subformulas $\psi(\bar{x}^i,A_i(\bar{x}^i))$ are obtained by plugging $A_i(\bar{x}^i)$ into $\psi(\bar{x}^i,\bar{y}^i)$ in place of $\bar{y}^i$.

Property (ii) is due to $\DNNF$ implying $\PASyn$:
\begin{restatable}{theorem}{syntacticImpliesSemanticNF}\label{syntactic-implies-semantic-nf}
	Every formula in $\DNNF$ is also in $\PASyn$.
\end{restatable}
Essentially, this is because the annotated affine transformations yield
valuations for satisfying globally quantified subformulas. For space reasons, the proof is in \cref{app:syntactic-implies-semantic-nf}. Thus, \cref{construct-circuit-for-semantic-nf}
yields a polynomial-time algorithm for $\PFnS$ for $\DNNF$ formulas.

Finally, we have property (iii): $\DNNF$ can be achieved with at most an exponential blow-up:
\begin{restatable}{theorem}{constructSyntacticNF}\label{construct-syntactic-nf}
	Every quantifier-free PA formula has an equivalent in $\DNNF$\ of at most exponential size.
\end{restatable}
Note that 
\cref{syntactic-implies-semantic-nf,construct-syntactic-nf} yield an alternative proof for
\cref{construct-semantic-nf}. While \cref{construct-semantic-nf} could be shown using
either of the two recent quantifier elimination
techniques~\cite{DBLP:conf/icalp/0001MS24,DBLP:conf/icalp/HaaseKMMZ24},
\cref{construct-syntactic-nf} depends on the specific geometric insight from
\cite{DBLP:conf/icalp/HaaseKMMZ24}, namely \cref{affine-transformations}.

Roughly speaking, the idea for proving \cref{construct-syntactic-nf} is to bring the formula into a DNF where each co-clause only contains linear inequalities. For each co-clause we can then apply \cref{affine-transformations} to yield exponentially many affine transformations $A_i$ that yield candidate assignments for $\bar{y}$. From these $A_i$, we then construct the subformulas of the form \eqref{piece-dnnf}. The proof is in \cref{app:construct-syntactic-nf}.
\vspace{-0.4cm}
\paragraph{Succinctness}
We have seen that compared to $\PASyn$, the form $\DNNF$ has the advantage that
it is  syntactic (i.e.\ easy to check). However, as we show now, $\PASyn$ has
the advantage that it can be \emph{exponentially more succinct}. More
specifically, there are formulas in $\PASyn$ whose smallest equivalent in
$\DNNF$ are exponentially larger:
\begin{restatable}{theorem}{succinctnessSemanticSyntactic}\label{succinctness-semantic-syntactic}
	There is a family $(\Psi_n)_{n\ge 0}$ of $\PASyn$ formulas such
	that any equivalent $\DNNF$\ has size $2^{\Omega(|\Phi_n|)}$.
\end{restatable}
One can take $\Psi_n(x,y):=x<y\le x+2^n \wedge y\equiv
0\pmod{2^n}$.  For each $x$, there is exactly one $y\in[x,x+2^n]$ with
$\Psi_n(x,y)$, and there are exponentially many ($2^n$) possible differences
between $x$ and $y$. One can argue that for each such difference, a separate
affine map has to appear in any $\DNNF$. A full proof is in
\cref{app:succinctness-semantic-syntactic}.

\section{Discussion and Conclusion}
\label{sec:conclusion}
Our work maps out the landscape of functional synthesis for Presburger specifications, setting up a new research agenda towards normal forms for such specifications and compilation to them. In doing so, it exposes fundamental differences between functional synthesis from Boolean and Presburger specifications. Specifically the complexity upper bounds for $\PFnS$ match the best known algorithms
for $\BFnS$ (EXPTIME), though for one-output specifications, $\BFnS$ is known
to be poly-time solvable, while $\PFnS$ is at least $\NP$-hard. 
This makes it necessary to design new normal forms for $\PA$ specifications using  new concepts of modulo-tameness and local quantification.  Interestingly,
the condition of local quantification may be viewed as a generalization of the SynNNF form used in $\BFnS$~\cite{AACKRS19}. It is also surprising that we can characterize the space of all Skolem functions for Presburger specifications using a set of intervals that can be represented by Presburger circuits, while the corresponding characterization of the space of all Boolean Skolem functions using Skolem basis in~\cite{DBLP:conf/cav/AkshayCJ23} has the flavour of an on-set and a don't-care set.  A priori, there doesn't seem to be a natural connection between these two representations, although Boolean functional synthesis can be encoded as Presburger functional synthesis.  This warrants further investigation into the relation between these representations.

There are further questions that arise from the comparison with Boolean specifications. For instance, in~\cite{SBAC21}, it was shown that there exists a precise characterization for polynomial-time and size solvable Boolean functional synthesis. We do not have such a characterization for Presburger arithmetic in general, though Theorem~\ref{thm:optimal} does provide such a result for the restricted case of single-output specifications. We leave the development of such a necessary and sufficient condition for polynomial time synthesis for Presburger formulas as a challenging open problem.

As another such instance, Conjunctive Normal Form (CNF) is well-accepted as a standard form for Boolean formulas, and state-of-the-art SAT-solvers are often highly optimized for CNF formulas. Some Boolean functional synthesis engines also exploit the CNF representation for efficient processing. However, for Presburger Arithmetic, there is no such dominant representation form that we are aware of. For example, QF\_LIA (quantifier-free linear integer arithmetic) benchmarks used by the SMT community are Presburger formulas sans modulo constraints, that are not always represented in CNF. We remark that not all knowledge compilation based approaches for synthesis require CNF representation to start with. For example, in the Boolean case~\cite{DBLP:conf/cav/AkshayCJ23,DBLP:journals/amai/AkshayCS24} work with formula in Negation Normal Form directly.

Finally, we would like to improve our constructions to make them more efficient in theory and in practice. For instance, the modulo-tameness definition currently can lead to blowups that can potentially be avoided by finding alternate characterizations and normal forms. Thus, we expect our results to be a stepping stone towards practical implementability of Skolem function synthesis algorithms for Presburger arithmetic via knowledge compilation in the future and their wider use within the KR and SMT communities.

\section*{Acknowledgments}
We are grateful to Christoph Haase for answering questions about the literature
around \cref{exists-star-forall} and to Dmitry Chistikov for pointing out the work of
\citeauthor{DBLP:journals/mst/Schoning97} (\citeyear{DBLP:journals/mst/Schoning97}).

\begin{wrapfigure}[3]{l}{0.1\textwidth}
        \vspace{-15pt}
        \begin{center}
                \includegraphics[width=1.8cm]{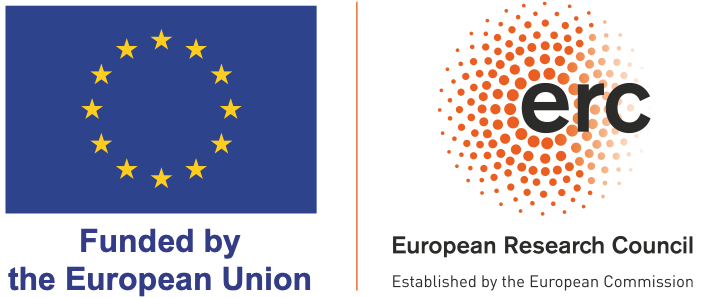}
        \end{center}
\end{wrapfigure}
Funded by the European Union (ERC, FINABIS, 101077902).  Views and opinions
expressed are however those of the authors only and do not necessarily reflect
those of the European Union or the European Research Council Executive Agency.
Neither the European Union nor the granting authority can be held responsible
for them.

\label{beforebibliography}
\newoutputstream{pages}
\openoutputfile{main.pages.ctr}{pages}
\addtostream{pages}{\getpagerefnumber{beforebibliography}}
\closeoutputstream{pages}

\bibliographystyle{kr}
\bibliography{bib}

\clearpage
\appendix
\crefalias{section}{appsec}
\crefalias{subsection}{appsec}

\gdef\thesubsection{\Roman{subsection}.}%
\gdef\thesubsectiondis{\Roman{subsection}.}%
\renewcommand\thesubsection{\thesection.\Roman{subsection}}

\section{Additional material on Section~\ref{sec:problem}}
\label[appsec]{app:problem}
In this section, we prove \cref{presburger-circuit-minimality,definable-skolem}.

In \cref{app:formal-statement-presburger-circuit-minimality}, we provide a formal statement for \cref{presburger-circuit-minimality} (namely, \cref{basis-complete-minimal}), prove a key lemma in \cref{app:minimal-period}, and finally prove \cref{basis-complete-minimal} in \cref{app:proof-formal-statement-presburger-circuit-minimality}. 

Then in \cref{app:definable-skolem}, we prove \cref{definable-skolem}.

\subsection{Remark~\ref{presburger-circuit-minimality}: Formal statement}\label[appsec]{app:formal-statement-presburger-circuit-minimality}
To formally state it, we need some terminology.

\begin{definition}\label{def:f-circuit}
	For a set of functions $F$, we use $F^{\circ}$ to denote the 
	set of all functions computed by circuits having as gates i) integer linear functions, and ii) functions from $F$. We say that
	$F$ is \emph{Presburger-complete} if $F^\circ$ is exactly the set of Presburger-definable functions.
\end{definition}

A circuit constructed as in Definition~\ref{def:f-circuit} is called an $F$-circuit. If there is no danger of confusion, we also write ``$\circuit\in F^\circ$'' to mean that $\circuit$ is an $F$-circuit.
We say that a set of Presburger-complete functions $F$ is \emph{minimally Presburger-complete} if for every $f\in F$, the set $F\setminus f$ is not Presburger-complete.

\newcommand{\fBasis}{\mathcal{B}}
\newcommand{\fBasisPrime}{\mathcal{B}^{\mathsf{p}}}
In the rest of this section, we consider two such sets $F$. We define
\begin{align*}
\fBasis &= \{\max, E, \fdiv_m \mid m\in\Z \}, \\
\fBasisPrime &= \{\max, E, \fdiv_p \mid \text{$p\in\Z$ is a prime}\}.
\end{align*}
Hence, a $\fBasis$-circuit is exactly what we define as a Presburger circuit. 
In this section, we also consider $\fBasisPrime$ to complete the picture of Presburger-completeness.
Of course in practice, one would not use $\fBasisPrime$ for PFnS, but rather
$\fBasis$. We will prove the following:
\begin{theorem}\label{basis-complete-minimal}
The set $\fBasis$ is Presburger-complete, and the set $\fBasisPrime$ is
minimally Presburger-complete. Moreover, there is no finite Presburger-complete
set of functions. 
\end{theorem}
Here, Presburger-completeness of $\fBasis$ is already mentioned in \cref{presburger-circuits-completeness-definable}, which will be proven later in \cref{thm:exponential-presburger-circuit}, where
this will also be accompanied by complexity bounds.

Note that if we prove that $\fBasisPrime$ is minimally Presburger-complete then the second statement follows immediately: If there were a
finite Presburger-complete set $F$, then a finite subset of $\fBasisPrime$
would suffice to compute all functions in $F$: This is because each function in
$F$ is expressible using finitely many functions in $\fBasisPrime$. However, a
finite Presburger-complete subset of $\fBasisPrime$ contradicts minimality of
$\fBasisPrime$.

The remainder of this section is therefore devoted to showing minimal
Presburger-completeness of $\fBasisPrime$.

\subsection{Minimal periods of definable sets}\label[appsec]{app:minimal-period}
In the proof of \cref{basis-complete-minimal} (and thus
\cref{presburger-circuit-minimality})---specifically the fact that $\fdiv_p$ is
necessary for each prime $p$ (see \cref{each-prime-necessary})---and also later
in \cref{thm:exp-lower-bound}, we will need a result about the minimal period
of sets definable by Presburger circuits. We prove this here.

We say that a set $S\subseteq\Z^k$ is \emph{definable using $F$} if the
characteristic function of $S$ belongs to $F^\circ$. Recall that every
Presburger-definable set $S\subseteq\Z$ is \emph{ultimately periodic}, meaning
there are $x_0,p\in\N$ such that for every $x\in\Z$, $|x|\ge x_0$, we have
$x+p\in S$ if and only if $x\in S$. In this case, $p$ is a \emph{period} of
$S$. Because of B\'{e}zout's identity, if a set $S$ has periods $p_1$ and
$p_2$, then $\gcd(p_1,p_2)$ is also a period of $S$. Therefore, the smallest
period of $S$ is the greatest common divisor of all periods.

The following \lcnamecref{minimal-period} is the main result of this subsection:
\begin{restatable}{lemma}{minimalPeriod}\label{minimal-period}
Let $D\subseteq\Z$ and $F=\{\max, E, \fdiv_m \mid m\in D\}$. Suppose
$S\subseteq\Z$ can be defined by an $F$-circuit with at most $e$ many $\fdiv$-gates.
Then $\lcm(D)^e$ is a period of $S$.
\end{restatable}

\begin{proof}%
We use
structural induction w.r.t.\ $F$-circuits. Let $M:=\lcm(D)$. 
Given $m\in\Z$, an \emph{$m$-modular assignment} consists of a sequence
$(r_1,I_1,c_1,d_1),\ldots,(r_n,I_n,c_n,d_n)$, where $r_i\in[0,m-1]$,
	$I_i\subseteq\Z$ is an interval (finite or infinite), and $c_i,d_i\in\Z$. Such a conditional assignment
\emph{defines} the function $f\colon\Z\to\Z$, where $f(x)=(c_ix+d_i)/m$, where $i$
is the smallest $j$ such that $x\equiv r_j\pmod{m}$ and $x\in I_j$.

Now we perform a structural induction that shows that every function defined by an $F$-circuit with at most $e$ $\fdiv$-gates is also defined by an $M^e$-modular assignment.

Observe that it suffices to show this in the case that every occurring
$\fdiv$-gate is a $\fdiv_M$-gate: This is because by the identity
\eqref{k-using-kl}, we can replace each $\fdiv_m$-gate for $m\in D$ by a
$\fdiv_M$-gate. 

Moreover, to simplify notation, we assume that the gates computing affine
functions $\Z\to\Z$ are of the following form:
\begin{description}
	\item[Sum] A gate with two inputs that computes $x+y$ for input values $x,y\in\Z$.
	\item[Multiplication with constant] A gate with one input $x$ that computes $a\cdot x$, for some constant $a\in\Z$.
	\item[Constant One] A gate with no input that always yields $1\in\Z$.
\end{description}
Clearly, every gate computing an affine function of its inputs can be
decomposed into a circuit of such gates.

Note that our claim implies the \lcnamecref{minimal-period}: A characteristic
function that has an $M^e$-modular assignment must clearly have period $M^e$.
We prove the statement by induction on the circuit size. We make a case
distinction according to the type of the output gate.
\paragraph{Max}
If the output gate is $\max$, and the two input functions have $M^e$-modular
assignments, then we construct an $M^e$-modular assignment.  Suppose we have
two functions $f,f'\colon\Z\to\Z$ that each have an $M^e$-modular assignment.
We want to construct an $M^e$-modular assignment for the function $g\colon
x\mapsto \max(f(x),f'(x))$. Let $(r_1,I_1,c_1,d_1),\ldots,(r_n,I_n,c_n,d_n)$ be
an $M^e$-modular assignment for $f$, and
$(r'_1,I'_1,c'_1,d'_1),\ldots,(r'_{n'},I'_{n'},c'_{n'},d'_{n'})$ be an
$M^e$-modular assignment for $f'$.

The $M^e$-modular assignment for $g$ will be constructed as follows. For each
$i\in[1,n]$ and $j\in[1,n']$ where $r_i=r'_j$, we consider the interval
$I_i\cap I'_j$ and divide it further according to whether $f$ dominates $f'$:
\begin{align*}
K_{i,j} &= \{x\in I_i\cap I'_j \mid (c_ix+d_i)/{M^e}\ge (c'_jx+d'_j)/{M^e}\}, \\
K'_{i,j} &= \{x\in I_i\cap I'_j \mid (c_ix+d_i)/{M^e}< (c'_jx+d'_j)/{M^e}\}.
\end{align*}
Note that since $x\mapsto (c_ix+d_i)/{M^e}$ and $x\mapsto (c'_jx+d'_j)/{M^e}$
are linear functions, $K_{i,j}$ and $K'_{i,j}$ are intervals with $I_i\cap
I'_j=K_{i,j}\uplus K'_{i,j}$: For Two linear functions, there can be at most
one point where the domination changes from one function to another, and so the
set of points where one function dominates the other is convex. Thus, our new $M^e$-modular assignment for $g$ consists of the assignments
\begin{align*} &(r_i,K_{i,j},c_i,d_i), & &(r_i,K'_{i,j},c'_j,d'_j) \end{align*}
for each $i\in[1,n]$ and $j\in[1,n']$ such that $r_i=r_j$. Thus, we have
constructed an $M^e$-modular assignment for $g\colon x\mapsto
\max(f(x),f'(x))$.

\paragraph{Zero conditional}
If the output gate is $E$, and the two input functions have $M^e$-modular assignments, then we construct an $M^e$-modular assignment. Suppose the gate computes $g(x):=E(f(x), f'(x))$ for functions $f,f'\colon\Z\to\Z$. Let
$(r_1,I_1,c_1,d_1),\ldots,(r_n,I_n,c_n,d_n)$ be an $M^e$-modular assignment for $f$, and $(r'_1,I'_1,c'_1,d'_1),\ldots,(r'_{n'},I'_{n'},c'_{n'},d'_{n'})$ be an $M^e$-modular assignment for $f'$.

The $M^e$-modular assignment for $g$ will be constructed as follows. For each
$i\in[1,n]$ and $j\in[1,n']$ where $r_i=r'_j$, we consider the interval
$I_i\cap I'_j$ and divide it further according to the sign of $f$:
\begin{align*}
K_{i,j,-1} &= \{x\in I_i\cap I'_j \mid (c_ix+d_i)/{M^e}<0\}, \\
K_{i,j,0} &= \{x\in I_i\cap I'_j \mid (c_ix+d_i)/{M^e}=0\}, \\
K_{i,j,1} &= \{x\in I_i\cap I'_j \mid (c_ix+d_i)/{M^e}>0\}.
\end{align*}
Note that since $x\mapsto (c_ix+d_i)/{M^e}$ is a linear function, $K_{i,j,-1}$,
$K_{i,j,0}$, and $K_{i,j,1}$ are intervals with $I_i\cap I'_j=K_{i,j,-1}\uplus
K_{i,j,0}\uplus K_{i,j,1}$: A linear functions can change its sign at most
once, and so the sets $K_{i,j,-1}$, $K_{i,j,0}$, and $K_{i,j,1}$ are convex.
Thus, our new $M^e$-modular assignment for $g$ consists of the assignments
\begin{align*} &(r_i,K_{i,j,-1},0,0), & &(r_i,K_{i,j,0},c'_j,d'_j), &&(r_i,K_{i,j,1},0,0) \end{align*}
for each $i\in[1,n]$ and $j\in[1,n']$ such that $r_i=r_j$. Thus, we have
constructed an $M^e$-modular assignment for $g\colon x\mapsto E(f(x),f'(x))$.

\paragraph{Division function}
If the output is $\fdiv_M$, and the input function has an $M^e$-modular assignment, then we construct an $M^{e+1}$-modular assignment. Suppose overall, we compute $g(x):=\fdiv_M(f(x))$ for a function $f\colon\Z\to\Z$. Let
$(r_1,I_1,c_1,d_1),\ldots,(r_n,I_n,c_n,d_n)$ be an $M^e$-modular assignment for $f$.

Our goal is to construct an $M^{e+1}$-modular assignment to compute $g(x)$. This means, we need to ascertain the following data about $x$:
\begin{itemize}
\item The remainder of $x$ modulo $M^e$, in order to pick the right assignment used for computing $f(x)$. The remainder of $x$ modulo $M^{e}$ clearly only depends on $x$'s remainder modulo $M^{e+1}$. Let $\alpha\colon [0,M^{e+1}-1]\to [0,M^e-1]$ be the function where for every $m\in[0,M^{e+1}]$, the $\alpha(m)$ is the remainder of $m$ modulo $M^e$.

\item If $f(x)$ used the $i$-th assignment, then we need the remainder of $(c_ix+d_i)/M^e$ modulo $M$. This means, if $(c_ix+d_i)/M^e=sM+t$ with $t\in[0,M-1]$, then we want to determine $t$. Note that this implies $c_ix+d_i=sM^{e+1}+tM$. Thus, we can compute $t$ by taking the remainder of $c_ix+d_i$ modulo $M^{e+1}$, and dividing it by $M$. Let $\beta_i\colon[0,M^{e+1}-1]\to[0,M]$ be the function so that if $x\equiv m\bmod{M^{e+1}}$, then $c_ix+d_i\equiv \beta(m)\bmod{M^{e+1}}$. 
\end{itemize}

The $M^{e+1}$-modular assignment for $g$ will be constructed as follows. For each
$m\in[0,M^{e+1}-1]$, and every $i\in[1,n]$ with $r_i=m$, then we add an assignment
\[ (m,I_i,c_i,d_i-M^e\beta_i(m)). \]
This is correct, since $((c_ix+d_i)/M^e-\beta_i(m))/M=(c_ix+d_i-M^e\beta_i(m))/M^{e+1}$.
\paragraph{Sum}
If the output gate is $+$, and the two input functions have an $M^e$-modular assignment, then we construct an $M^e$-modular assignment.
Suppose the output gate computes $g(x):=f(x)+f'(x)$ for functions $f,f'\colon\Z\to\Z$. Let
$(r_1,I_1,c_1,d_1),\ldots,(r_n,I_n,c_n,d_n)$ be an $M^e$-modular assignment for $f$, and $(r'_1,I'_1,c'_1,d'_1),\ldots,(r'_{n'},I'_{n'},c'_{n'},d'_{n'})$ be an $M^e$-modular assignment for $f'$. In the $M^e$-modular assignment for $g$, we proceed as follows. For any $i\in[1,n]$ and $j\in[1,n']$ such that $r_i=r_j$, we include the assignment
\[ (r_i,I_i\cap I'_j,c_i+c'_j,d_i+d_j). \]
Then clearly, the resulting $M^e$-modular assignment computes $g$.
\paragraph{Multiplication with a constant}
If the output performs a multiplication with a constant $a\in\Z$, and the input function has an $M^e$-modular assignment, then we construct an $M^{e}$-modular assignment. Suppose overall, we compute $g(x):=a\cdot f(x)$ for a function $f\colon\Z\to\Z$. Let
$(r_1,I_1,c_1,d_1),\ldots,(r_n,I_n,c_n,d_n)$ be an $M^e$-modular assignment for $f$.
Then clearly, we obtain an $M^e$-modular assignment by using the assignments
\[ (r_i,I_i,a\cdot c_i, a\cdot d_i) \]
for all $i\in[1,n]$.

\paragraph{Constant One}
Finally, suppose the output gate computes the function $g\colon\Z\to\Z$ with $g(x)=1$ for all $x\in\Z$. This is easily achieved by the $M^e$-modular assignment 
\[ (r,\Z,0,M) \]
for every $r\in[0,M-1]$.

\end{proof}

\subsection{Proof of Formal Statement of Remark~\ref{presburger-circuit-minimality}}\label[appsec]{app:proof-formal-statement-presburger-circuit-minimality}
In this subsection, we prove \cref{basis-complete-minimal}, the formal statement of \cref{presburger-circuit-minimality}.

For Presburger-completeness of
$\fBasisPrime$, we first observe the following:
\begin{lemma}
	For every $k,\ell\in\Z$, we have $\{\Eq,\fdiv_{k\ell}\}^\circ=\{\Eq,\fdiv_k,\fdiv_\ell\}^\circ$.
\end{lemma}
\begin{proof}
	To show that $\fdiv_{k\ell}\in \{\Eq,\fdiv_k,\fdiv_\ell\}^\circ$, observe that
	$\fdiv_{k\ell}(x)$ is exactly the following function
$$
\sum_{r\in[0,k\ell-1]} \Eq(\fdiv_\ell(\fdiv_k(x-r))\cdot k\ell-(x-r),~\fdiv_\ell(\fdiv_k(x-r))),
$$
since the first argument to $\Eq$ vanishes if and only if $x\equiv r\bmod{k\ell}$.
Conversely, one can express $\fdiv_k(x)$ as
	\begin{equation*} \sum_{r\in[0,k\ell-1]} \Eq(\fdiv_{k\ell}(x)\cdot k\ell-(x-r), \fdiv_{k\ell}(x)\cdot \ell+\fdiv_k(r)). \label{k-using-kl}\end{equation*}
	As above, the first argument to $\Eq$ vanishes if and only if $x\equiv
	r\pmod{k\ell}$. This means $x=yk\ell+r$ for some $y\in\Z$. If we write $r=sk+t$ for some
	$t\in[0,k-1]$, then $x\equiv t\pmod{k}$, and also
	$\fdiv_{k}(x)=y\ell+s=\fdiv_{k\ell}(x)\cdot\ell+\fdiv_k(r)$. Finally, note that here $\fdiv_k(r)$ is a constant and thus need not invoke $\fdiv_k(\cdot)$. Similarly, we can also express $\fdiv_\ell$ in the same manner.
\end{proof}
Since every integer is the product of primes, the preceding lemma shows that
Presburger-completenss of $\fBasis$ yields Presburger-completeness of
$\fBasisPrime$. Let us now show minimality, i.e.\ that every function in
$\fBasisPrime$ is necessary.

\subsection*{Maximum is necessary}
\begin{proposition}\label{max-necessary}
	$\max$ is not in $\{\Eq, \fdiv_m\mid m\in\Z\}^\circ$.
\end{proposition}
\begin{proof}
Suppose $\max$ belongs to $\{\Eq,\fdiv_m\mid m\in\Z\}^\circ$. Observe that a function $\Z^n \to \Z$ that is expressible using $\Eq$, integer linear functions, and $\fdiv_m$ for $m\in\Z$ is existentially definable in the structure $(\Z;+,0,1)$ (importantly: without $\le$), i.e., 
any such function can be defined by a formula that only has existential quantifiers and no occurrence of $\le$ .
Now since $\max$ is definable existentially, so is the set $\N$ of natural numbers.

However, $\N$ is not existentially definable in $(\Z;+,0,1)$: An existentially definable set in $(\Z;+,0,1)$ is a finite union of projections of solution sets of linear Diophantine equation systems. If such a projection is infinite, it contains a negative number (because infiniteness implies a solution of the homogeneous equation system that is non-zero in our component, thus we can add or subtract it to get a negative entry). Hence, the natural numbers cannot be existentially defined in $(\Z;+,0,1)$.
\end{proof}

\subsection*{Zero conditional $\Eq$ is necessary}
\begin{proposition}\label{E-necessary}
	$\Eq$ is not in $\{\max,\fdiv_m\mid m\in\Z\}^\circ$.
\end{proposition}

A function $f\colon\Z\to\Q$ is \emph{linear} if there are $a,b\in\Q$
with $f(x)=ax+b$ for every $x\in\Z$. We say that $f$ is
\emph{pseudo-linear} if there are constants $M,B\ge 0$ and linear functions
$g,h\in\Z\to\Q$ such that:
\begin{enumerate}
	\item for every $x\le-M$, we have $|f(x)-g(x)|\le B$, and
	\item for every $x\ge M$, we have $|f(x)-h(x)|\le B$.
\end{enumerate}
\begin{restatable}{lemma}{maxDivPseudoLinear}\label{max-div-pseudo-linear}
	Every function in $\{\max,\fdiv_m\mid m\in\Z\}^\circ$ is pseudo-linear.
\end{restatable}
\begin{proof}
To simplify notation, we assume that the gates computing affine
functions $\Z\to\Z$ are of the following form:
\begin{description}
	\item[Sum] A gate with two inputs that computes $x+y$ for input values $x,y\in\Z$.
	\item[Multipilication with constant] A gate with one input $x$ that computes $a\cdot x$, for some constant $a\in\Z$.
	\item[Constant One] A gate with no input that always yields $1\in\Z$.
\end{description}
Clearly, every gate computing an affine function of its inputs can be
decomposed into a circuit of such gates.

We prove the lemma by induction on the size of the circuit, and we make a case
	distinction according to the type of the output gate.

\paragraph{Max}
Suppose the output gate is $\max$, and its two input functions are pseudo-linear.
	Thus, our circuit computes the function $m\colon\Z\to\Z$ with $x\mapsto \max(f(x),f'(x))$, where $f,f'\colon\Z\to\Z$ are pseudo-linear functions. Hence, we have constants $M,B,M',B'\ge 0$ and linear functions $g,h,g',h'\colon\Z\to\Z$ such that:
	\begin{itemize}
		\item for $x\le -M$, we have $|f(x)-g(x)|\le B$,
		\item for $x\ge M$, we have $|f(x)-h(x)|\le B$,
		\item for $x\le -M'$, we have $|f'(x)-g'(x)|\le B'$, and
		\item for $x\ge M'$, we have $|f'(x)-h'(x)|\le B'$.
	\end{itemize}
	We now construct a linear function $h''\colon\Z\to\Q$ and constants $M'',B''\ge 0$ such that $|m(x)-h''(x)|\le B''$ for every $x\ge M''$. By symmetry, this implies one can also construct $M''$ and a linear $g''\colon\Z\to\Q$ with $|m(x)-g''(x)|\le B'$ for $x\le -M''$. Together this, implies that $k$ is pseudo-linear. Thus, we focus on constructing $h''$ and $M''$.

	Since $h,h'\colon\Z\to\Q$ are linear functions, we have $h(x)=ax+b$ and $h'(x)=a'x+b'$ for all $x\in\Z$, for some coefficients $a,a',b,b'\in\Q$. We distinguish three cases:
	\begin{enumerate}
		\item Suppose $a=a'$. In this case, $h(x)$ and $h'(x)$ only ever differ by $b-b'$. In particular, for $x\ge M'':=M$, we have 
			\begin{align*}
				|m(x)-h(x)|&\le \max (|f(x)-h(x)|, |f'(x)-h(x)|) \\
			&\le \max(B,B')+\lceil|b-b'|\rceil. \end{align*}
		Thus, picking $h'':=h$ and $B'':=\max(B,B')+\lceil
		|b-b'|\rceil$ satisfies our conditions.
	\item Suppose $a>a'$. Then for $x>(B+B'+|\lfloor b-b'\rfloor|)/\lceil (a-a')\rceil$, we have 
		\begin{align*}
			h(x)-h'(x)&=(a-a')x+b-b' \\
			&\ge \lceil a-a'\rceil\cdot x-|\lfloor b-b'\rfloor| \\
			&>B+B'
		\end{align*}
		and thus 
		\begin{align*}
			f(x)-f'(x)&\ge (h(x)-B)-(h'(x)+B') \\
			&=h(x)-h'(x)-(B+B')>0
		\end{align*}
		which implies $m(x)=f(x)$. Therefore, we set
		$M''=\max(M,(B+B'+|\lfloor b-b'\rfloor|)/\lceil
		(a-a')\rceil)$, and $h''=h$, and $B''=B$. Then, as we
		have seen above, $x\ge M''$ implies $m(x)=f(x)$ and
		thus $|m(x)-h''(x)|=|f(x)-h(x)|\le B=B''$, as desired.
		\item Suppose $a<a'$. This is symmetric to the case $a>a'$, so it can be shown the same way.
	\end{enumerate}

\paragraph{Division}
	Suppose the output gate is $\fdiv_k$, and the function computed by the circuit below its input gate is the pseudo-linear $f\colon\Z\to\Z$. Let $M,B\ge 0$ and let $g,h\colon\Z\to\Q$ be linear functions such that if $x\le -M$, then $|f(x)-g(x)|\le B$ and if $x\ge M$, then $|f(x)-h(x)|\le B$. Note that
	\[ |\fdiv_k(f(x))-\tfrac{f(x)}{k}|<k \]
	for every $x\in\Z$. Therefore, for $x\le -M$, we have
	\begin{align*}
		\left|\fdiv_k(f(x))-\tfrac{g(x)}{k}\right|&\le \left|\fdiv_k(f(x))-\tfrac{f(x)}{k}\right|+\left|\tfrac{f(x)}{k}-\tfrac{g(x)}{k}\right|\\
		&\le k+\tfrac{B}{k}\le k+B
	\end{align*}
	and in the same way, we can show that $|\fdiv_k(f(x))-\tfrac{h(x)}{k}|\le k+B$. for $x\ge M$. Since $x\mapsto \tfrac{g(x)}{k}$ and $x\mapsto \tfrac{h(x)}{k}$ are linear functions, this shows that $x\mapsto\fdiv_k(f(x))$ is pseudo-linear.
\paragraph{Sum}
Suppose the output gate is a $+$-gate, and its two input functions are pseudo-linear.
	Thus, our circuit computes the function $s\colon\Z\to\Z$ with $x\mapsto f(x)+f'(x)$, where $f,f'\colon\Z\to\Z$ are pseudo-linear functions. Hence, we have constants $M,B,M',B'\ge 0$ and linear functions $g,h,g',h'\colon\Z\to\Z$ such that:
	\begin{itemize}
		\item for $x\le -M$, we have $|f(x)-g(x)|\le B$,
		\item for $x\ge M$, we have $|f(x)-h(x)|\le B$,
		\item for $x\le -M'$, we have $|f'(x)-g'(x)|\le B'$, and
		\item for $x\ge M'$, we have $|f'(x)-h'(x)|\le B'$.
	\end{itemize}
	Observe that then
	\begin{align*}
		|s(x)-(g(x)+g'(x))|&=|f(x)+f'(x)-g(x)-g'(x)| \\
		&\le |f(x)-g(x)|+|f'(x)-g'(x)| \\
		&\le 2B
	\end{align*}
	for $x\le -M$. Similarly, we have $|s(x)-(h(x)+h'(x))|\le 2B$ for $x\ge
	M$.  Since $x\mapsto g(x)+g'(x)$ and $x\mapsto h(x)+h'(x)$ are linear
	functions, this shows that $s$ is pseudo-linear as well.
\paragraph{Multiplication by constant}
	Suppose the output gate is the multiplication by the constant $a\in\Z$, and the function computed by the circuit below its input gate is the pseudo-linear $f\colon\Z\to\Z$. Let $M,B\ge 0$ and let $g,h\colon\Z\to\Q$ be linear functions such that if $x\le -M$, then $|f(x)-g(x)|\le B$ and if $x\ge M$, then $|f(x)-h(x)|\le B$. Note that
\begin{align*}
	|a\cdot f(x)-a\cdot g(x)|&=|a|\cdot |f(x)-g(x)|\le |a|\cdot B
\end{align*}
for $x\le -M$. Likewise, we have $|a\cdot f(x)-a\cdot h(x)|\le |a|\cdot B$ for $x\ge M$.
Since the functions $x\mapsto a\cdot g(x)$ and $x\mapsto a\cdot h(x)$ are linear, this shows that $x\mapsto a\cdot f(x)$ is pseudo-linear.
\paragraph{Constant One}
The ``constant 1'' function is linear itself, and thus pseudo-linear.
\end{proof}

Hence, if $\Eq \in \{\max,\fdiv_m\mid m\in\Z\}^\circ$, then every function in ${\fBasis}^\circ = \{\Eq,\max, \fdiv_m \mid m \in \Z \}^\circ$ must be
pseudo-linear. Since $\fBasis$ is Presburger-complete, it follows that all Presburger-definable functions are pseudo-linear, in particular $f \colon \Z \to \Z$ with $f(2x) = 2x$ and
$f(2x+1) = 3x$. This is clearly Presburger-definable. However, we show: 

\begin{restatable}{lemma}{notPseudoLinear}\label{not-pseudo-linear}
	The function $f\colon\Z\to\Z$ with $f(2x)=2x$ and $f(2x+1)=3x$ for $x\in\Z$ is not pseudo-linear.
\end{restatable}
\begin{proof}
	Essentially, we will argue that if $f$ were pseudo-linear, then the expression $f(2x+1)-f(2x)$ could
	only attain finitely many values, but it equals $x$.  

	If $f$ is pseudo-linear, then in particular, there are $M,B\ge 0$ and a linear $g\colon\Z\to\Q$, say $g(x)=ax+b$, such that for $x\ge M$, we have $|f(x)-g(x)|\le B$. This implies
	\begin{align*}
		|f(x+1)-f(x)|&\le|f(x+1)-g(x+1)|+|g(x+1)-f(x)| \\
		&\le B+|g(x)+a-f(x)| \\
		&\le 2B+|a|
	\end{align*}
	if $x\ge M$.
	Thus, the expression $f(x+1)-f(x)$ can only assume finitely many values
	when $x$ ranges over $[M,\infty)$. However, we have $f(2x+1)-f(2x)=3x-2x=x$, a contradiction. 
\end{proof}
This implies that $\Eq \notin \{\max,\fdiv_m\mid m\in\Z\}^\circ$, meaning
$\Eq$ is necessary.

\subsection*{Division functions are necessary}
\begin{proposition}\label{each-prime-necessary}
	For every prime $p\in\Z$, we have $\fdiv_p\notin(\fBasisPrime\setminus\{\fdiv_p\})^\circ$.
\end{proposition}

\begin{proof}[Proof of \cref{each-prime-necessary}]
If the function $\fdiv_p$ were definable using
$\fBasisPrime\setminus\{\fdiv_p\}$, then so would the set $D_p\subseteq\Z$ of
integers divisible by $p$. But then there is a finite set $P$ of primes, with
$p\notin P$, such that $D_p$ is definable using $F=\{\max, E, \fdiv_q\mid q\in
P\}$. By \cref{minimal-period}, $\lcm(P)^e$ is a period of $D_p$, for some
$e\in\N$. But this implies that the smallest period of $D_p$, which is $p$,
divides $\lcm(P)^e$, in contradiction to $p\notin P$.
\end{proof}

Hence, we have proved that all the functions in $\fBasisPrime$ are needed to achieve Presburger-completeness, thereby proving its minimality.

\subsection{Proof of Remark~\ref{definable-skolem}}\label[appsec]{app:definable-skolem}
Here, we prove:
\definableSkolem*

First, note that
there is a Presburger-definable well-order on $\Z^m$. For example,
pick an arbitrary linear order on the $2^m$ orthants in $\Z^m$, and
order the vectors inside each orthant lexicographically. Suppose
$\text{lex}(\bar{w},\bar{y})$ is a Presburger formula defining such a
well-order on $\Z^m$, i.e. for every $u, v \in \mathbb{Z}^m$,
$\text{lex}(u,v)$ holds iff $u$ is equal to or ordered before $v$ in
this order.  Now, for any quantifier-free formula
$\varphi(\bar{x},\bar{y})$, consider the formula
$\Psi(\bar{x},\bar{y})$ defined as follows:
\begin{multline*}
(\lnot \exists \bar{z}  \ \varphi(\bar{x},\bar{z})) \implies \bar{y} = \bar{0}  \ \land \\
(\exists \bar{z} \ \varphi(\bar{x},\bar{z})) \implies (\varphi(\bar{x},\bar{y}) \land \forall \bar{w} \varphi(\bar{x},\bar{w}) \implies \text{lex}(\bar{y},\bar{w}))	
\end{multline*}
This formula states that if $\varphi(\bar{x},\bar{z})$ is not
satisfiable for any $\bar{z}$, then $\bar{y}$ must be $\bar{0}$;
otherwise $\bar{y}$ must be assigned the lexicographically smallest
tuple of values that makes $\varphi$ true. Thus,
$\Psi(\bar{x},\bar{y})$ uniquely defines a function.  Morover, by
virtue of its definition, this is also a Skolem function for $\bar{y}$
in $\forall \bar{x} \exists \bar{y}: \varphi(\bar{x},\bar{y})$. The required
quantifier-free $\PA$ formula $\psi(\bar{x},\bar{y})$ is obtained by
eliminating quantifiers from $\Psi(\bar{x},\bar{y})$.

\section{Additional material on Section~\ref{sec:synthesis}}
\label[appsec]{app:synthesis}

\subsection{Comparison with existing complexity upper bounds for Skolem function synthesis}\label[appsec]{app:comparison}
In this section, we discuss two related lines of work on Presburger functions that we are aware of which would lead to alternative Skolem function synthesis algorithms.

The first by Ibarra and Leininger~\cite{DBLP:journals/siamcomp/IbarraL81} is in the context of representing Presburger functions and shows that a set of functions similar to our $\fBasis$ is sufficient to express all Presburger-definable functions. Their methods can be used for Skolem synthesis, but this would yield a \emph{quadruply-exponential} upper bound: Starting from a (quantifier-free) specification, one would first convert it into a definition of a Skolem function as in \cref{sec:prelim}, which yields a $\Pi_1$-formula. Ibarra and Leininger's construction would then convert (i)~the $\Pi_1$ formula into a semilinear representation, (ii)~the semilinear representation into a counter machine~\cite[Thm. 2.2 and Lemma 4.1]{DBLP:journals/tcs/GurariI81}, (iii)~the counter machine into an SL program~\cite[Thm. 3]{DBLP:journals/jacm/GurariI81}, and finally (iv)~the SL program into a circuit~\cite{DBLP:journals/siamcomp/IbarraL81}. For step (i), only a doubly exponential upper bound is known, and for (ii),(iii), the authors of \cite{DBLP:journals/tcs/GurariI81,DBLP:journals/jacm/GurariI81} each provide an exponential upper bound. Overall, this yields a quadruply-exponential upper bound.

A second line of work is the translation of Presburger-definable functions into
$L_+$-programs by
Cherniavsky~\cite[Thm.~5]{DBLP:journals/siamcomp/Cherniavsky76}
.  This would
only yield a \emph{triply-exponential} upper bound for an $L_+$ program (which
would not even be a circuit).  Starting from a quantifier-free specification,
one would first convert it into a function as above, yielding a
$\Pi_1$-formula. In addition, Cherniavsky's construction requires the function
to only have one output. Reducing the number of output variables requires
existentially quantifying the other outputs, turning our $\Pi_1$-formula into a
$\Sigma_2$-formula. Then, the approach requires the formula to be
quantifier-free; even applying recent
techniques~\cite{DBLP:conf/icalp/HaaseKMMZ24,DBLP:conf/icalp/0001MS24} to a
$\Sigma_2$-formula, this would yield a doubly-exponential formula. Finally,
Cherniavsky's construction of an $L_+$ program itself is exponential (and
translating the program into a circuit might incur another blowup). Overall,
this only yields a triply-exponential upper bound.

\subsection{Proof of Theorem~\ref{thm:exponential-presburger-circuit}}\label[appsec]{app:exponential-presburger-circuit}
\thmExponentialPresburgerCircuit*

In the proof of \cref{thm:exponential-presburger-circuit}, we will simplify the description of constructed circuits, by also allowing the function $\Cp$ in gates, where
	\[ \Cp(x,y)=\begin{cases} y & \text{if $x\ge 0$} \\ 0 & \text{otherwise}\end{cases} \]
Using gates labeled with $\Cp$ is clearly just syntactic sugar, because
$\Cp$-gates can be replaced by $E$-gates and $\max$-gates: We have
$\Cp(x,y)=\Eq(\min(\max(x+1,0),1)-1,y)$ and $\min(x,y)=-\max(-x,-y)$ for any
$x,y\in\Z$.

\begin{remark}\label{c-vs-e-max}
	Removing both $\max$ and $\Eq$ from the set of atomic functions
	and instead allowing $\Cp$ would yield the same expressive power of Presburger circuits.
To see this, recall the definition of $F^\circ$ for a set $F$ of functions (\cref{def:f-circuit}).
First, note that $\{\max,\Eq\}^\circ \supseteq \{\Cp\}^\circ$, because 
\[ \Cp(x,y)=\Eq(\min(\max(x+1,0),1)-1,y) \] and $\min(x,y)=-\max(-x,-y)$.
Similarly, we also have that $\max(x,y) = \Cp(x-y,x) + \Cp(y-x-1,y)$
and $\Eq(x,y) = \Cp(\min(\Cp(x,-x)+\Cp(-x,x),0),y)$ and so $\{\max,\Eq\}^\circ = \{\Cp\}^\circ$.
Hence the $\Cp$ function should be thought of as simply a syntactic sugar in the place of $\Eq$ and $\max$.
Thus, $\{\Cp,\fdiv_m \mid m \in \Z\}$ is also a Presburger-complete collection of functions.
\end{remark}

Let us first observe that it suffices to prove \cref{thm:exponential-presburger-circuit} for quantifier-free formulas over the signature $(\Z;+,\le,0,1)$ (i.e. without modulo constraints). This is because given a quantifier-free formula $\varphi(\bar{x},\bar{y})$ with modulo-constraints, we can construct in polynomial time an equivalent existential formula $\exists \bar{z}\colon \varphi'(\bar{x},\bar{y},\bar{z})$, where $\varphi'$ has no modulo-constraints. If we can construct Presburger circuits for quantifier-free formulas over $(\Z;+,\le,0,1)$, then we can view $\varphi'(\bar{x},\bar{y},\bar{z})$ as having input $\bar{x}$ and output $(\bar{y},\bar{z})$, construct a circuit for a Skolem function for $\varphi'$. Then, one obtains a a Presburger circuit for a Skolem function for $\varphi$ by projecting away the output variables $\bar{z}$ and only outputting $\bar{y}$.

Therefore, we now assume that our input formula $\varphi(\bar{x},\bar{y})$ is
quantifier-free over the signature $(\Z;+,\le,0,1)$. Suppose
$\bar{x}=(x_1,\ldots,x_n)$ and $\bar{y}=(y_1,\ldots,y_m)$. By bringing
$\varphi$ in disjunctive normal form, we obtain a disjunction
$\bigvee_{i=1}^r\varphi_i$, where (i)~each $\varphi_i$ is a conjunction of
atoms, (ii)~each $\varphi_i$ has polynomial size in $\varphi$, and (iii)~$r$ is
at most exponential in the size of $\varphi$. Now each $\varphi_i$ can be
written as $A_i\bar{y}\le B_i\bar{x}+\bar{c}_i$ for some $A_i\in\Z^{\ell\times
m}$, $B_i\in\Z^{\ell\times n}$, and $\bar{c}_i\in\Z^{\ell}$. To simplify
notation, we assumed that the number $\ell$ of inequalities is the same for
each $i$ (this can easily be achieved by introducing trivial inequalities).

The idea of our circuit construction is to find the smallest $i$ such that for
our given $\bar{x}$, the system $A_i\bar{y}\le B_i\bar{x}+\bar{c}_i$ has a
solution $\bar{y}$, and then output such a solution. To check whether this
system has a solution (and to find one), we use \cref{affine-transformations}.
It implies that for every $i$, there is an exponential-sized set $P_i$ of pairs
$(D,\bar{d})$ where $D\in\Q^{n\times\ell}$ and $\bar{b}\in\Q^n$ with
$\fracnorm{D},\fracnorm{\bar{d}}$ being bounded exponentially, such that there is
a solution $\bar{y}$ if and only if there is one of the form
$D(B_i\bar{x}+\bar{c}_i)+\bar{d}$ for some $(D,\bar{d})\in P_i$.  Since all the
sets $P_i$ contain at most exponentially many pairs, we may assume that each
$P_i$ contains $s$ elements, for some $s\in\N$ that is bounded exponentially in
$\varphi$. Moreover, we order all the pairs in $P_i$ and write
\[ P_i=\{(D_{i,j},\bar{d}_{i,j}) \mid j=1,\ldots,s\}. \]
Our circuit will check whether $A_i\bar{y}\le B_i\bar{x}+\bar{c}_i$ has a
solution by trying all pairs in $P_i$. Then, when it has found the smallest $i$
for which there is a solution, it outputs
$D_{i,j}(B_i\bar{x}+\bar{c}_i)+\bar{d}_{i,j}$, where $j$ is the smallest $j$
for which this expression is a solution. In slight abuse of terminology, we say that the pair $(i,j)$ is a \emph{solution} if 
\begin{align}
	A_i(D_{i,j}B_i\bar{x}+\bar{c}_i)+\bar{d}_{i,j})&\le B_i\bar{x}+\bar{c}_i,\label{solution-inequalities} \\
	D_{i,j}(B_i\bar{x}+\bar{c}_i)+\bar{d}_{i,j}&\in\Z^{m}, \label{solution-modulo}
\end{align}
in other words, when $D_{i,j}(B_i\bar{x}+\bar{c}_i)+\bar{d}_{i,j}$ is an integral solution to $A_i\bar{y}\le B_i\bar{x}+\bar{c}_i$ for $\bar{y}$.
Let $\ll$ be the lexicographic ordering on $[1,r]\times[1,s]$, meaning $(i,j)\ll (i',j')$ if and only if (a)~$i<i'$ or (b)~$i=i'$ and $j<j'$. Then our circuit finds the $\ll$-minimal solution.

This means, we consider the function:
\begin{align*}
	F_{i,j}(\bar{x})=\begin{cases} 1 & \text{if $(i,j)$ is the minimal solution} \\
		0 & \text{otherwise}
	\end{cases}
\end{align*}
and construct a circuit for $F_{i,j}$. To this end, we first construct a circuit for the function
\begin{align*}
	G_{i,j}(\bar{x})=\begin{cases} 1 & \text{if $(i,j)$ is a solution} \\ 
		0 & \text{otherwise}
	\end{cases}
\end{align*}
For this, in turn, we construct a circuit that checks whether \eqref{solution-inequalities} is satisfied, and another to check \eqref{solution-modulo}. 

We write the system \eqref{solution-inequalities} of inequalities as $\bar{a}_k^\top\bar{x}\le\bar{b}_k$ for $k=1,\ldots,\ell$ for some vectors $\bar{a}_k,\bar{b}_k\in\Q^n$. Each of these vectors has polynomially many bits, so we can construct a polynomial-sized circuit for the function
\begin{align*}
	I_{i,j}(\bar{x})=\Cp\left(\sum_{k=1}^\ell \Cp(\bar{b}_k-\bar{a}_k^\top\bar{x},1)-\ell,1\right)
\end{align*}
which returns $1$ if and only if $\bar{a}_k^\top\bar{x}\le\bar{b}_k$ for each $k\in\{1,\ldots,\ell\}$. To check \eqref{solution-modulo}, we write the entries of the vector $D_{i,j}(B_i\bar{x}+\bar{c}_i)+\bar{d}_{i,j}$ as $\bar{e}_k^\top\bar{x}+f_k$ for some $\bar{e}_k\in\Q^n,f_k\in\Q$ for $k=1,\ldots,m$. Note that the condition ``$\bar{e}_k^\top\bar{x}+f_k\in\Z$'' is a modulo constraint and can be rewritten as $\bar{g}_k^\top\bar{x}+h_k\equiv 0\pmod{m_k}$ for some $\bar{g}_k\in\Z^n$ and $h_k,m_k\in\Z$. We thus construct a polynomial-sized circuit for the function
\begin{align*}
	M_{i,j}(\bar{x})=\Eq\left(m-\sum_{k=1}^m \Eq(\fdiv_{m_k}(\bar{g}_k^\top\bar{x}+h_k),1),1\right),
\end{align*}
which returns $1$ if and only if $m_k$ divides $\fdiv_{m_k}(\bar{g}_k^\top\bar{x}+h_k$ for each $k=1,\ldots,m$; and otherwise returns $0$. 

With $I_{i,j}$ and $M_{i,j}$, we can construct a circuit for $G_{i,j}$, since:
\[ G_{i,j}(\bar{x})=\Eq(I_{i,j}(\bar{x})+M_{i,j}(\bar{x})-2,1). \]
This allows us to construct a circuit for $F_{i,j}$, since
\[ F_{i,j}(\bar{x})=\Eq\left(\sum_{\stackrel{(t,u)\in[1,r]\times[1,s],}{(t,u)\ll(i,j)}} G_{t,u}(\bar{x}), G_{i,j}(\bar{x})\right). \]
This is because the first argument to $\Eq$ is zero if and only if $G_{t,u}(\bar{x})=0$ for all $(t,u)\ll(i,j)$. If that first argument is zero, then we evaluate the second argument. The latter, in turn, is $1$ if and only if $G_{i,j}(\bar{x})=1$.

Finally, with a circuit for $F_{i,j}$, we can now construct a circuit for a
Skolem function for $\varphi$. Here, we use the fact that if for some $\bar{x}$, there exists a $\bar{y}$ with $\varphi(\bar{x},\bar{y})$, then by definition of $F_{i,j}$, there is exactly one pair $(i,j)\in[1,r]\times[1,s]$ such that $F_{i,j}(\bar{x})=1$. Moreover, if $F_{i,j}(\bar{x})$ holds, then we can pick $D_{i,j}(B_i\bar{x}+\bar{c}_i)+\bar{d}_{i,j})$ for $\bar{y}$. Therefore, the following is a Skolem function for $\varphi$:
\[ f(\bar{x})=\sum_{i=1}^r\sum_{j=1}^s \Eq(1-F_{i,j}(\bar{x}), D_{i,j}(B_i\bar{x}+\bar{c}_i)+\bar{d}_{i,j}). \]

Let us now estimate the size of the circuit. Since $m$ and $\ell$ are
polynomial in the input, and all matrix entries of each $D_{i,j}$ and each
$\bar{d}_{i,j}$ have polynomially many bits, the circuits for each $I_{i,j}$
and each $M_{i,j}$ are polynomial-sized. Therefore, the circuit for each
$G_{i,j}$ is also polynomial-sized. The circuit for each $F_{i,j}$, however, is
exponential, because $r$ and $s$ are exponential in the size of $\varphi$.
Finally, in the circuit for $f$, the subcircuit for computing
$D_{i,j}(B_i\bar{x}+\bar{c}_i)+\bar{d}_{i,j})$ is again polynomial-sized, and
thus each $\Eq$ term in $f$ only adds polynomially many gates. Then, the sum
is a single linear combination gate with exponentially many wires to lower
gates. In total, we obtain exponentially many gates.

\subsection{Proof of Observation~\ref{obs:Bool:short}}\label[appsec]{app:Bool}
In this subsection, we prove:
\booleanCircuitSize*
To prove \cref{obs:Bool:short}, we will first state it more formally (as \cref{obs:Bool}), which requires some notation. We define:
\newcommand{\SPA}{\mathsf{S}_{\mathsf{PA}}}
\newcommand{\SBool}{\mathsf{S}_{\mathsf{bool}}}
\begin{align*}
	\begin{split}
	\SPA(n) = \max\{|\ckt| \mid \text{$\ckt$ is a minimum-size Presburger circuit } \\
	~~~~~\text{Skolem function for a $\PA$ formula of size $\le n$}\}
	\end{split} \\
	\begin{split}
	\SBool(n) = \max\{|\ckt| \mid \text{$\ckt$ is a minimum-size Boolean circuit} \\
	~~~~~\text{Skolem function for a Boolean formula of size $\le n$}\}
	\end{split}
\end{align*}
It has been pointed out before that if all Boolean formulas admit a polynomial-size circuits for Skolem functions, then $\NP\subseteq\Ppoly$~\cite[Theorem 1]{AkshayCGKS21}. Hence most likely, small Skolem function circuits do not always exist. However, apriori, it could be easier to prove a super-polynomial lower bound on Skolem function circuits than to refute $\NP\subseteq\Ppoly$. 

We observe here that the existence of small Boolean Skolem function circuits is
in fact \emph{equivalent} to $\NP\subseteq\Ppoly$. This means, proving a
super-polynomial lower bound for Skolem function circuits is as difficult as
refuting $\NP\subseteq\Ppoly$.

For this, recall that $\mathsf{P}/\mathsf{poly}$ is the 
class of all languages $L$ that could be decided by polynomial-sized Boolean circuits $\{\ckt_n\}_{n \in \mathbb{N}}$, one for each possible length of the string. More precisely, there is a polynomial $p$
such that each Boolean circuit $\ckt_n$ has $n$ input gates, is of size at most $p(n)$ and
$\ckt_n$ outputs 1 for a string $x \in \{0,1\}^n$ iff $x \in L$. Based on this complexity class,
we make the following observation.
\begin{restatable}{observation}{obsBool}\label{obs:Bool}
	$\SBool(n)$ is bounded by a polynomial if and only if $\mathsf{NP} \subseteq 
	\mathsf{P}/\mathsf{poly}$.
\end{restatable}
The ``only if'' was shown in \cite[Theorem 1]{AkshayCGKS21}. The ``if'' holds
because small circuits for SAT can be used to compute Skolem functions
bit-by-bit. See \cref{app:Bool} for details.

\begin{proof}
	Suppose $\SBool(n)$ is bounded by a polynomial. Under this assumption,
	it was proven in~\cite[Theorem 1]{AkshayCGKS21},
	that $\mathsf{NP} \subseteq \mathsf{P}/\mathsf{poly}$.
	For the converse direction, suppose $\mathsf{NP} \subseteq  \mathsf{P}/\mathsf{poly}$.
	This implies that Boolean satisfiability can be decided by polynomial-sized circuits. 
	It is well known that if Boolean satisfiability can be solved by polynomial-sized circuits,
	then there are polynomial-sized circuits $\{C_n\}_{n \in \mathbb{N}}$ such that not only does $C_n$ correctly solve satisfiability instances of size $n$,
	but $C_n$ also outputs a satisfying assignment for satisfiable formulas of size $n$~\cite[Chapter 6]{AB}.
	Now, consider the following Turing machine: On input $\psi = \forall \bar{x} \exists \bar{y} \varphi(\bar{x},\bar{y})$ where $\varphi$ is a formula of length $n$, the circuit $C_n$ and an assignment $A$ to the variables $\bar{x}$, it 
	\begin{itemize}
		\item First fixes this assignment
		in the formula $\varphi$ to get a formula $\varphi_X(\bar{y})$ depending only on $\bar{y}$.
		\item Then it runs the circuit $C_n$ on $\varphi_X(\bar{y})$ to get an assignment $B$ for $\bar{y}$
		and outputs it.
	\end{itemize}
	
	Note that by assumption on the circuit $C_n$, it follows that this Turing machine, always outputs an assignment $Y$ such that $\exists \bar{y} \varphi(A,\bar{y}) \iff \varphi(A,B)$ is true. 
	Hence, this Turing machine is a Skolem function for $\bar{y}$ in $\psi$.
	Furthermore, it runs in polynomial time in the size of $X, \psi$ and $C_n$. 
	By standard complexity arguments, any polynomial-time Turing machine
	can be converted into polynomial-sized Boolean circuits and so 
	there is a polynomial-sized Boolean circuit which for inputs $X, \psi$ and $C_n$
	acts as a Skolem function for $\bar{y}$ in $\psi$. This means that if we fix $\psi$ and $C_n$,
	then we get a polynomial-sized Boolean circuit Skolem function for $\bar{y}$ in $\psi$.
	Since $\psi$ was any arbitrary formula, it follows that $\SBool(n)$ is bounded by a polynomial.
\end{proof}

\subsection{Proof of Theorem~\ref{thm:exp-lower-bound}}\label[appsec]{app:exp-lower-bound}
The proof of \cref{thm:exp-lower-bound} relies on \cref{prop:pa-formula-as-function}, which lets us convert 
formulas into circuits. If $\varphi(\bar{x})$ has $n$
free variables $\bar{x}=(x_1,\ldots,x_n)$, then its \emph{characteristic function}
is the map $\xi_\varphi\colon \Z^n\to\{0,1\}$, where $\xi_\varphi(\bar{u})=1$
if and only if $\varphi(\bar{u})$ holds.
\begin{restatable}{proposition}{PAFormulaAsFunction}\label{prop:pa-formula-as-function}
	Given an existential $\PA$ formula $\varphi$, we can construct a
	Presburger circuit for $\xi_\varphi$ in time $\bigO{2^{|\varphi|}}$. If
	$\varphi$ is quantifier-free, the time bount becomes 
	$\bigO{|\varphi|}$.
\end{restatable}
\begin{proof}
	The first statement follows from the second, together with
	\cref{thm:exponential-presburger-circuit}: Given an existential
	$\varphi$, we view the quantified variables as output variables, and
	apply \cref{thm:exponential-presburger-circuit} to the resulting
	$\Pi_2$ formula. The resulting Presburger circuit allows us to compute
	an assignment of the quantified variables such that if
	$\varphi(\bar{x})$ is satisfied, then with these values. We can
	therefore use the second statement to check whether the output of the
	circuit makes $\varphi$ true. This proves our first statement. Thus, it
	remains to prove the second statement.

	We can construct it by structural induction on the tree
	representation of $\varphi$. For the base case
	(i.e. $\varphi$ is an atomic formulas in $\PA$), let $t$ denote the
	term $\sum_{i=1}^n a_i x_i + a_0$. Referring to
	\cref{app:exponential-presburger-circuit} for the definitions of $\Cp$, if $\varphi$ is $t \ge 0$, we use the Presburger circuit
	for $\Cp(t, 1)$.  If $\varphi$ is $t \equiv r \pmod M$, we use the
	circuit for $\Eq(\fdiv_M(t-r) - \fdiv_M(t-r-1) - 1, 1)$. Finally,
        if $\varphi$ is $t \not\equiv r \pmod M$, we use the circuit
        for $1-\Eq(t-r - M\fdiv_M(t-r),1)$. 
	
	Given a circuit for $\xi_\phi$ for a formula $\varphi$, the circuit for $\xi_{\neg\varphi}$
	is simply the circuit for $1 - \xi_\varphi$. Finally, given circuits for
	$\xi_{\varphi_1}$ and $\xi_{\varphi_2}$ for two formulas
	$\varphi_1, \varphi_2$, the circuits
	for conjunction and disjunction of $\varphi_1$ and $\varphi_2$ are
	those for $\Cp(\xi_{\varphi_1}-1,\xi_{\varphi_2})$ and $1 - \Cp(-\xi_{\varphi_1}, 1 - \xi_{\varphi_2})$ 
	respectively. In this way, we can construct the circuit for the formula $\varphi$.

	Since the circuits corresponding to the atomic formulas can be
	constructed in time linear in the sizes of the atomic formulas
	themselves, and since each induction step introduces
	exactly one $C$ function and at most a constant number of linear
	terms, the overall circuit can clearly be constructed in time
	$\bigO{|\varphi|}$.
\end{proof}

\begin{proof}[Proof of~\cref{thm:exp-lower-bound}]
	In \cite[Section 6]{DBLP:conf/icalp/HaaseKMMZ24}, the authors construct
	a family of existential formulas $\varphi_n(x)$ such that the set
	$\llbracket\varphi_n\rrbracket$ defined by $\varphi_n$ has a minimal
	period of $2^{2^{\Omega(n)}}$.
	
	Suppose $\varphi_n(x)=\exists\bar{y}\colon \psi(\bar{y},x)$. Consider the $\Pi_2$ formula
	\[ \mu_n\equiv\forall x\colon \exists \bar{y}\colon \psi_n(\bar{y},x). \]
	Let $\ckt_n$ be a Presburger circuit for a Skolem function for $\mu_n$.
	Such a Skolem function for $\mu_n$ yields, given $x\in\N$, a vector
	$\bar{y}$ such that if $\varphi_n(x)$, then $\psi_n(\bar{y},x)$. Using
	$\psi$, we can easily turn $\ckt_n$ into a circuit $\ckt'_n$ of size
	polynomial in $\ckt_n$, such that $\ckt'_n$ defines the set $\llbracket
	\varphi_n\rrbracket$: On input $x\in\Z$, $\ckt'_n$ first uses $\ckt_n$ to
	compute $\bar{y}$, and then simulate the circuit for $\psi_n$ from Proposition~\ref{prop:pa-formula-as-function}
	to output $\psi_n(\bar{y},x)$.

	Let $M_n\subseteq\Z$ be the set of divisors $m\in\Z$ for which
	$\fdiv_m$ gates occur in $\ckt'_n$. Moreover, let $e_n$ be the number of
	$\fdiv$ gates in $\ckt'_n$. Then, by \cref{minimal-period}, the number
	$\lcm(M_n)^{e_n}$ is a period of $\llbracket\varphi_n\rrbracket$.  This
	implies that the minimal period of $\llbracket\varphi_n\rrbracket$,
	which is at least $2^{2^{\Omega(n)}}$, divides $\lcm(M_n)^{e_n}$.
	Hence, $\lcm(M_n)^{e_n}\ge 2^{2^{\Omega(n)}}$. However, 
	$\lcm(M_n)^{e_n}$ is at most exponential in  $|\ckt'_n|$, and thus at most
	exponential in $|\ckt_n|$. Therefore, for some $c>0$, we have
	$2^{|\ckt_n|^c}\ge 2^{|\ckt'_n|}\ge \lcm(M_n)^{e_n}\ge 2^{2^{\Omega(n)}}$,
	hence $|\ckt_n|$ must be at least exponential.
\end{proof}

\subsection{Proof of Observation~\ref{hardness-one-one-long}}\label[appsec]{app:hardness-one-one-long}
Here, we prove:
\hardnessOneOneLong*
We first rephrase \cref{hardness-one-one-long} more formally.
Let%
\begin{align*}
	&\mathsf{S}_{\mathsf{1PA}}(n) = \max\{|\ckt| \mid \text{$\ckt$ is a minimum-size Presburger circuit } \\
	&\text{ for $\PA$ formulas over one input/output variable of size $\le n$}\}
\end{align*}

Then \cref{hardness-one-one-long} can be phrased as follows:
\begin{restatable}{observation}{hardnessOneOne}\label{hardness-one-one}
	If $\mathsf{S}_{\mathsf{1PA}}(n)$ is bounded by a polynomial then so is $\SBool(n)$. 
	Consequently, if $\mathsf{S}_{\mathsf{1PA}}(n)$ is bounded by a polynomial, then 
	$\mathsf{NP} \subseteq 
	\mathsf{P}/\mathsf{poly}$.
\end{restatable}
Here, recall the definition of $\SBool(n)$ from \cref{app:Bool}.
\begin{proof}
	First, we will give a polynomial-time reduction from $\Pi_2$ Boolean formulas
	to $\Pi_2$ PA formulas over one input and one output variable.
	Let $\forall \bar{x} \ \exists \bar{y} \ \psi(\bar{x},\bar{y})$ be a $\Pi_2$ Boolean formula where $\bar{x} = x_1,\dots,x_n$ and $\bar{y} = y_1,\dots,y_m$ and $\psi$ is a 3CNF formula. 
	
	Choose $n+m$ distinct primes $p_1,\dots,p_n, q_1,\dots, q_m$. By the prime number theorem, 
	we have $n+m$ distinct primes in the range $[0,O(n+m)]$, which can be found and verified in polynomial time, since they have logarithmically many bits. We will now construct a
	Presburger formula
	$\forall a \ \exists b \ \varphi(a,b)$ over one input variable and one output variable in the following way.
	
	For each clause $C_i = \ell_i^1 \lor \ell_i^2 \lor \ell_i^3$ in $\psi$, $\varphi$ will have a clause
	of the form $F(C_i) := F(\ell_i^1) \lor F(\ell_i^2) \lor F(\ell_i^3)$ where $F(\ell_i^j)$ is defined as follows.
	\begin{itemize}
		\item If $\ell_i^j$ is $x_k$ for some $k$, then $F(\ell_i^j)$ is $a \equiv 0 (\bmod \ p_k)$.
		\item If $\ell_i^j$ is $\overline{x_k}$ for some $k$, then $F(\ell_i^j)$ is $a \not \equiv 0 (\bmod \ p_k)$.
		\item If $\ell_i^j$ is $y_k$ for some $k$, then $F(\ell_i^j)$ is $b \equiv 0 (\bmod \ q_k)$.
		\item If $\ell_i^j$ is $\overline{y_k}$ for some $k$, then $F(\ell_i^j)$ is $b \not \equiv 0 (\bmod \ q_k)$.
	\end{itemize}
	
	Notice that the size of the formula $\varphi$ is $O((n+m)^2)$.
	Intuitively any assignment $X$ of $\bar{x}$ in the formula $\psi$ corresponds 
	to the following set of numbers for the variable $a$: $S(X) := \{k : \forall 1 \le i \le n, \ p_i \text{ divides } k \text{ if and only if } X(x_i) = 1\}$. Since each $p_i$ is a prime, it is easy to see that $S(X) \neq \emptyset$ for any $X$  and also that $S(X) \cap S(X') = \emptyset$ for any two distinct assignments $X'$.
	Furthermore, the union of $S(X)$ over all possible assignments covers all of the natural numbers.
	
	Similarly, any assignment $Y$ of $\mathbf{y}$
	corresponds to the following set of numbers for the variable $b$: $S(Y) := \{k : \forall 1 \le i \le m, q_i \text{ divides } k \text{ if and only if } Y(y_i) = 1\}$. Once again, it is easy to see
	that $S(Y) \neq \emptyset, S(Y) \cap S(Y')$ for $Y \neq Y'$ and also that the union of $S(Y)$ over all possible assignments covers all of the natural numbers.
	From the construction of the Presburger formula $\varphi$ it is then clear that
	$\psi(X,Y)$ is true for any two assignments $X,Y$ iff $\varphi(e,r)$ is true
	for any $e \in S(X), r \in S(Y)$.

	Now, suppose $\mathsf{S}_{\mathsf{1PA}}(n)$ is bounded by a polynomial of the form $n^c$ for some fixed $c$. So, in particular, this means for the $\Pi_2$ formula $\forall a \ \exists b \ \varphi(a,b)$
	there is a polynomial-sized PA circuit $C_\varphi$ which acts as a Skolem function for this $\Pi_2$ formula.
	Using this circuit, we now synthesize a Boolean circuit of polynomial size for the
	for the Boolean formula $\forall \bar{x} \exists \bar{y} \psi(\bar{x}, \bar{y})$. 
	To this end, consider the following polynomial-time Turing machine: On input $\forall \bar{x} \exists \bar{y} \psi(\bar{x}, \bar{y}), \forall a \exists b \varphi(a,b)$, the circuit $C_\varphi$ and an
	assignment $X$ to the variables $\bar{x}$, it
	\begin{itemize}
		\item First computes the smallest number $N$ in $S(X)$ which is given by exactly the product of the primes in the set $\{p_i : X(x_i) = 1\}$. 
		\item Then, it runs the circuit $C_\varphi$ on the number $N$ and produces an output number $M$.
		\item It then converts $M$ into an assignment of $\bar{y}$ in the following manner:
		$Y(y_i) = 1$ iff the prime $q_i$ divides $M$.
		\item Finally, it outputs $Y$.
	\end{itemize}
	
	By the discussion above and by the assumption that $C_\varphi$ is a Skolem function for $\forall a \exists b \varphi(a,b)$, it follows that this Turing machine always outputs an assignment $Y$
	such that $\exists \bar{y} \psi(X,\bar{y}) \iff \psi(X,Y)$ is true. Hence, when we fix the inputs $\forall\bar{x}\exists\bar{y}\psi(\bar{x},\bar{y})$ and $C_\varphi$ (hence, only leave $X$ as input), then the Turing computes
	a Skolem function fo $\bar{y}$ in $\forall \bar{x} \ \exists \bar{y} \ \psi(\bar{x},\bar{y})$.
	Now, by the same arguments as the ones given in~\cref{obs:Bool}, it follows
	that it is possible to convert this Turing machine into a polynomial-sized
	Boolean circuit that acts as a Skolem function for $\bar{y}$ in $\forall \bar{x} \exists \bar{y} \psi(\bar{x}, \bar{y})$. 
	Hence, $\SBool$ is bounded by a polynomial. 
\end{proof}

\section{Additional material on Section~\ref{sec:semantic-nf}}
\label[appsec]{app:semantic-nf}
\subsection{Proof of Proposition~\ref{prop:ModTameUniv}}\label[appsec]{app:ModTameUniv}
\propModTameUniv*
\begin{proof}
  Let $\psi$ be a maximal conjunctive sub-formula of $\varphi$, and
  let $M^{\psi}$ be the least common multiple (lcm) of all moduli that
  appear in modular constraints involving $y$ in $\psi$. Clearly,
  $M^\psi \le \prod_{M \in \mathfrak{M}} M$. For every atomic formula
  of the form $a y + t_x \equiv b \pmod M$ in $\psi$, where $t_x$ is a
  linear term in $\bar{x}$ and $a \in \mathbb{Z}$, we replace it with
  the semantically equivalent formula $\bigvee_{0\le r <
    M^\psi}\big((y\equiv r \pmod {M^\psi})\wedge (\mu t_x \equiv \mu
  (b - a r) \pmod {M^\psi})\big)$, where $\mu = \frac{M^\psi}{M}$.
  Since the replacement is done only at the leaves of the sub-tree
  rooted at the maximal conjunction node corresponding to $\psi$, the
  sub-tree resulting from the substitution represents a maximal
  conjunctive sub-formula of the new formula $\varphi'$.  By virtue of
  the construction, this sub-formula is also $y$-modulo tame. By
  repeating the above process for all maximal conjunctive sub-formulas
  of $\varphi$, we obtain a new formula $\varphi'$ that is
  semantically equivalent to $\varphi$, and is $y$-modulo tame.  It is
  easy to see that the above technique for transforming $\varphi$ to
  $\varphi'$ takes $\bigO{|\varphi|\big(\prod_{M\in
      \mathfrak{M}}M\big)}$ time.
\end{proof}

\subsection{Lemma~\ref{prop:sk-fns-of-disj} and its proof}\label[appsec]{app:sk-fns-of-disj}
Next, we present a
helper lemma that is useful for proving the main result of the
subsection, and is of independent interest.%
. The lemma essentially explains how a $\PA$ circuit for the Skolem function of a disjunction of $\PA$ specifications can be obtained by efficiently combining the $\PA$ circuits for the Skolem functions at each of the disjuncts.

\begin{restatable}{lemma}{lemDisj}\label{prop:sk-fns-of-disj}
Let $\varphi(\bar{x},y)$ be the formula $\bigvee_{i=1}^k
\varphi_i(\bar{x},y)$, where each $\varphi_i$ is a $\PA$ formula.  Let
$\ckt_i$ be a Presburger circuit representing a Skolem function $f_i$
for $y$ in $\forall \bar{x} \exists y:\, \varphi_i$, for each $i \in
\{1, \ldots r\}$.  Then, a Presburger circuit for a Skolem function
$f$ for $y$ in $\forall \bar{x} \exists y:\, \varphi$ can be
constructed in time $\bigO{k^2 + |\varphi| + \sum_{i=1}^k
  |\circuit_i|}$.
\end{restatable}

  \begin{proof}
  Let $\circuit_{\varphi}$ be a Presburger circuit for the
  characteristic function $\xi_{\varphi_i}(\bar{x},y)$ of
  $\varphi_i(\bar{x},y)$.  Since Presburger functions are closed under
  composition, $\xi_{\varphi_i}(\bar{x}, f_i(\bar{x}))$ is a
  Presburger function.  Moreover, a Presburger circuit for
  $\xi_{\varphi_i}(\bar{x}, f_i(\bar{x}))$ can be constructed in time
  $\bigO{|\varphi_i| + |\circuit_i|}$.

  We now construct the Presburger function $\gamma(\bar{x}) =
  \sum_{i=1}^k 2^i \xi_{\varphi_i}(\bar{x}, f_i(\bar{x}))$.  For every
  $\bar{x} \in \mathbb{Z}^n$, the $i^{th}$ bit in the binary
  representation of $\gamma(\bar{x})$ is $1$ iff $\varphi_i(\bar{x},
  f_i(\bar{x}))$ holds.  Thus, $\gamma(\bar{x})$ encodes the truth
  values of all $\varphi_i(\bar{x}, f_i(\bar{x}))$ for $1 \le i \le k$
  in a single function.  Since $2^i$ is a constant for each $i$, and
  since constants are represented in binary in Presburger functions,
  it follows that a Presburger circuit $\circuit_\gamma$ for $\gamma$
  can be constructed in time $\bigO{k^2 +
    \sum_{i=1}^k(|\circuit_{\varphi_i}| + |\circuit_i|)}$, i.e. in
  $\bigO{k^2 + |\varphi| + \sum_{i=1}^k |\circuit_i|}$.

  Finally, we construct the Presburger function $f(\bar{x}) =
  \sum_{i=1}^k C\big(\gamma(\bar{x}) - 2^i, f_i(\bar{x}) -
  f_{i-1}(\bar{x})\big)$, where $f_{-1}(\bar{x})$ is defined to be
  $0$.  We claim that for every $\bar{x} \in \mathbb{Z}^n$, if there
  is some $y \in \mathbb{Z}$ such that $\varphi(\bar{x}, y)$ holds,
  then $\varphi(\bar{x}, f(\bar{x}))$ holds as well.  In other words,
  $f(\bar{x})$ is a Skolem function for $y$ in $\forall \bar{x}
  \exists y:\, \varphi(\bar{x},y)$.  To prove this, we consider an
  arbitrary $\bar{x} \in \mathbb{Z}^n$ and assume that there exists $y
  \in \mathbb{Z}$ such that $\varphi(\bar{x},y)$ holds.  By definition
  of $\varphi$, there exists some $i \in \{1, \ldots k\}$ such that
  $\varphi_i(\bar{x}, y)$ holds.  Since $f_i(\bar{x})$ is a Skolem
  function for $y$ in $\forall \bar{x} \exists y:\,
  \varphi_i(\bar{x},y)$, it follows that $\varphi_i(\bar{x},
  f_i(\bar{x})$ holds as well.  Let $i^\star$ be the smallest $i$ in
  $\{1, \ldots k\}$ such that $\varphi_i(\bar{x}, f_i(\bar{x}))$
  holds.  From the definition of $\gamma$, it follows that
  $2^{i^\star} \le \gamma(\bar{x}) < 2^{i^\star + 1}$.  Therefore,
  $C\big(\gamma(\bar{x}) - 2^i, f_i(\bar{x}) - f_{i-1}(\bar{x})\big)$
  evaluates to $f_i(\bar{x}) - f_{i-1}(\bar{x})$ for all $i \in \{1,
  \ldots i^\star\}$, and to $0$ for all $i \in \{i^\star+1, \ldots
  k\}$.  Hence, $f(\bar{x})$ evaluates to
  $\sum_{i=1}^{i^\star}\big(f_i(\bar{x}) - f_{i-1}(\bar{x})\big) =
  f_{i^\star}(\bar{x})$. Since $\varphi(\bar{x}, f(\bar{x})) =
  \bigvee_{i=1}^r \varphi_i(\bar{x}, f_i(\bar{x})$, and since
  $\varphi_{i^\star}(\bar{x}, f_{i^\star}(\bar{x}))$ holds, it follows
  that $\varphi(\bar(x), f(\bar{x})$ holds as well.

  From the definition of $f(\bar{x})$, it is easy to see that a
  Presburger circuit for this function can be constructed in time
  $\bigO{k^2 + |\circuit_\gamma| + \sum_{i=1}^k |\circuit_i|}$, which
  in turn is in $\bigO{k^2 + |\varphi| + \sum_{i=1}^k |\circuit_i|}$.
\end{proof}

\subsection{Proof of Theorem~\ref{presburger-circuit-one-output}}\label[appsec]{app:presburger-circuit-one-output}

\paragraph{Some syntactic sugar}
Before we go into the proof of \cref{presburger-circuit-one-output}, let us introduce some syntactic sugar in Presburger circuits that will be useful.
A particularly useful class of Presburger functions that comes in
handy in various contexts is ``if-then-else'', or $\ite$, functions.
If $\varphi$ is a Presburger formula, and $f_1$ and $f_2$ are
Presburger functions, we use $\ite(\varphi, f_1, f_2)$ as shorthand
for $\Eq(1-\xi_\varphi,f_1)+\Eq(\xi_\varphi,f_2)$.  Here, $\xi_\varphi$ is the characteristic function of the set defined by $\varphi$, see the remarks at the beginning of \cref{app:exp-lower-bound}. 

Notice that $\ite(\varphi, f_1, f_2)$ evaluates to $f_1$ if $\varphi$ holds;
else it evaluates to $f_2$.  Furthermore, the size of $\ite(\varphi, f_1, f_2)$
is linear in $|\varphi|+|f_1|+|f_2|$.  With slight abuse of notation, and when
there is no confusion, we also use $\ite(f_1 \mathord{=} f_2, f_3, f_4)$ as
shorthand for $\Eq(f_1-f_2, f_3-f_4) + f_4$, and $\ite(f_1 \ge f_2, f_3, f_4)$
as shorthand for $\Cp(f_1-f_2, f_3-f_4) + f_4$. Here, $\Cp$ is the function
introduced at the beginning of \cref{app:exponential-presburger-circuit}. The
size of each of these functions is clearly linear in $|f_1| + |f_2| + |f_3| +
|f_4|$. Notice that neither $f_1$ nor $f_2$ may be terms in the syntax of
Presburger arithmetic, hence neither $f_1=f_2$ nor $f_1 > f_2$ may be formulas
in Presburger arithmetic.

Equipped with $\ite$, we are ready to present the proof of \cref{presburger-circuit-one-output}.
\paragraph{Start of the proof}
  W.l.o.g., $\varphi(\bar{x},y)$ can be written as
  $\bigvee_{i=1}^k \varphi_i(\bar{x},y)$ for some $r \ge 1$, where the
  top-most connective of each $\varphi_i$ (or label of the root of the
  sub-tree representing $\varphi_i$) is $\land$. The $\PA$ circuit for the Skolem function of a disjunction of $\PA$ specifications can be obtained by efficiently combining the $\PA$ circuits for the Skolem functions at each of the disjuncts (see  Lemma~\ref{prop:sk-fns-of-disj}). Thus, it suffices to focus on Skolem functions for maximal conjunctive formulas $\varphi_i$.

The $y$-modulo-tameness of $\varphi_i$ means there is
a single $M\in \mathbb{N}$ such that all modulo constraints involving
$y$ are of the form $y\equiv r \mod M$ for some $r\in [0,M-1]$.  Let
$R\subseteq \{0, 1, \ldots M-1\}$ be the set of residues that appear
in some modulo constraint involving $y$ in $\varphi_i$, i.e., for each
$r\in R$, there is an atomic formula $y\equiv r\mod M$ in
$\varphi_i$. Then $|R|\leq |\varphi_i|$.

  \newcommand{\lflag}{\mathrm{lf}}
  \newcommand{\uflag}{\mathrm{uf}}
  \newcommand{\lbound}{\mathrm{lb}}
  \newcommand{\ubound}{\mathrm{ub}}

We show how to obtain a set of integer intervals with end-points
parameterized by $\bar{x}$, such that the value of any Skolem function
for $y$ must lie within one of these intervals if $\varphi_i(\bar{x},
y)$ is to be satisfied.  Formally, a \emph{bounded interval} $I$ is an
ordered pair of Presburger functions, written as
$[\alpha(\bar{x}), \beta(\bar{x})]$.  We write $\set{I}:=\{z \in \mathbb{Z}
\mid \alpha(\bar{x}) \le z \le \beta(\bar{x})\}$.
Observe that if $\alpha(\bar{u}) > \beta(\bar{u})$ for some $\bar{u} \in
\mathbb{Z}^n$, then $\set{I(\bar{u})} = \emptyset$; such an interval
is called an \emph{empty interval}.  Given a set of bounded intervals
$L = \{[\alpha_{1}(\bar{x}), \beta_{1}(\bar{x})], \ldots
[\alpha_{s}(\bar{x}), \beta_{s}(\bar{x})]\}$, we abuse notation and
use $\set{L}$ to denote
$\bigcup_{i=1}^s \set{[\alpha_i(\bar{x}), \beta_i(\bar{x})]}$.  In
order to represent \emph{unbounded intervals}, we use a $4$-tuple
$B
= \langle \lflag(\bar{x}), \lbound(\bar{x}), \uflag(\bar{x}), \ubound(\bar{x})\rangle$,
where $\lflag(\bar{x})$ and $\uflag(\bar{x})$ are characteristic
functions of suitable Presburger formulas, and $\lbound(\bar{x})$ and
$\ubound(\bar{x})$ are Presburger functions giving the upper bound of
a left-open interval and lower bound of a right-open interval,
respectively. Let $\set{B}$ be $\{v \in \mathbb{Z} \mid $
either $(v \le \lbound(\bar{x}) ~\text{and}~ \lflag(\bar{x}) = 1)$, or
$(v \ge \ubound(\bar{x}) ~\text{and}~ \uflag(\bar{x}) = 1) \}$.  Thus,
for an assignment $\bar{u}$ of $\bar{x}$, depending on the values of
$\lflag(\bar{u})$ and $\uflag(\bar{u})$, $\set{B}$ may contain only a left-open
interval or only a right-open interval, or an interval containing all
integers, or even an empty interval. For brevity, we omit
$\bar{x}$ as arguments of $I, \alpha, \beta, \lflag, \lbound, \uflag$
and $\ubound$ when there is no confusion.

\begin{restatable}{claim}{clintervals}\label{cl:intervals}
  For each $r\in R$, there exist
    bounded intervals
    $L^r=\{[\alpha^r_1(\bar{x}),\beta^r_1(\bar{x})],\ldots,
    [\alpha^r_{k_r}(\bar{x}),\beta^r_{k_r}(\bar{x})]\}$, and unbounded
    intervals
    $B^r=\langle\lflag^r(\bar{x}),\lbound^r(\bar{x}),\uflag^r(\bar{x}),\ubound^r(\bar{x})\rangle$,
    s.t. for all assignments $\bar{u}$ to
    $\bar{x}$, \begin{align}& \{v\in\mathbb{Z}\mid \varphi_i(\bar{u},v)=1\}= \set{L^r}\cup \set{B^r} \label{eq:intervals}.  \end{align}
    Also, $k_r \le |\varphi_i|$ and $\PA$ circuits for
    $\alpha_i^r, \beta_i^r, \lflag^r, \lbound^r, \uflag^r, \ubound^r$
    can be constructed in time polynomial in $|\varphi_i|$.
    \end{restatable}

  {\em Proof of Claim.} Fix $r\in R$ and consider values of $y$
  s.t. $y \equiv r \pmod M$ holds. Since $\varphi_i$ is $y$-modulo
  tame, all constraints of the form $y \equiv r \pmod M$ in
  $\varphi_i$ evaluate to true, while all constraints of the form
  $y \equiv r' \pmod M$, where $r \neq r'$, evaluate to false.  We now
  inductively construct the set of bounded intervals $L^r$ and
  unbounded intervals $B^r$ bottom-up at each level in the tree
  representation of $\varphi_i$, such that
  Condition~(\ref{eq:intervals}) holds at each level. At the leaves,
  each inequality involving $y$ can be converted into a lower or upper
  bound using $\fdiv_M$. That is, if we have $\sum a_j x_j +b y
  +c \geq 0$ where each $a_j, b, c$ are constants, then we can rewrite
  this as $y\geq \fdiv_b (-c -\sum a_i x_i)$ if $b>0$, else we write
  it as $y\leq \fdiv_{-b} (c +\sum a_i x_i)$. Each of these is an
  unbounded interval; accordingly, we set $\uflag^r(\bar{x})$ (resp.
  $\lflag^r(\bar{x})$) to $1$, and use $\fdiv_b (-c -\sum a_i x_i$
  (resp. $\fdiv_{-b} (c +\sum a_i x_i)$) for $\ubound(\bar{x})$
  (resp. $\lbound(\bar{x})$).  Furthermore, $L^r$ is empty at the
  level of the leaves. It is then easy to see that
  Condition~(\ref{eq:intervals}) holds for the formula represented by
  each leaf of the tree for $\varphi_i$.

   When moving up the tree, say from level $i-1$ to level $i$, let
   $(L^r_{lc},B^r_{lc})$ represent the intervals at the left child of
   a node $p$ at level $i$, and $(L^r_{rc},B^r_{rc})$ represent the
   intervals at its right child.  We describe below how to combine the
   intervals at the children to obtain the corresponding intervals
   $(L^r_p, B^r_p)$ at the parent $p$.  Let $|L^r_{lc}|$ denote the
   count of bounded intervals in $L^r_{lc}$, and similarly for
   $L^r_{rc}$ and $L^r_p$. Our construction satisfies the following
   additional invariants:
\begin{enumerate} \item[(I1)] $|L^r_p| \leq
   |L^r_{rc}| + |L^r_{lc}|$, i.e. the count of bounded intervals at a
   node grows no faster than the total count of bounded intervals at
   its children \item[(I2)] the total size of Presburger circuits
   representing bounds in $L^r_p$ is in
   $\bigO{|L^r_{lc}|^2.|L^r_{rc}|^2 + |L^r_{lc}|.|L^r_{rc}|.S}$, where
   $S$ is the maximum size of a Presburger circuit representing a
   bound in $L^r_{lc}$ and $L^r_{rc}$. \item[(I3)] the total size of
   Presburger circuits representing flags and bounds in $B^r_p$ is in
   $\bigO{T^r_{lc} + T^r_{rc}}$, where $T^r_{lc}$ (resp. $T^r_{rc}$)
   is the total size of Presburger circuits representing flags and
   bounds in $B^r_{lc}$ (resp. $B^r_{rc}$)\end{enumerate}

   We have two cases to consider, corresponding to node $p$ being
   labeled $\vee$ or $\wedge$.  If $p$ is labeled $\vee$, we obtain
   $L^r_{p}$ as $L^r_{rc}\cup L^r_{lc}$ by taking the union of the
   bounded intervals from each child. The lower/upper bound and flag
   information at the parent is defined as:
   $\lflag^r_p=\min(\lflag^r_{rc}+\lflag^r_{lc},1)$,
   $\uflag^r_p=\min(\uflag^r_{rc}+\uflag^r_{lc},1)$. And $\lbound^r_p$
   is defined to be $\max(\lbound^r_{rc},\lbound^r_{lc})$ if
   $\lflag^r_{rc}+\lflag^r_{lc}>1$, and otherwise it is propagated
   from whichever $\lflag^r$ is 1. Similarly for $\ubound^r_p$ which
   is the min of both $\ubound^r$ if both $\uflag$s are 1, else it is
   propagated from whichever is 1. Note that doing this ensures
   Condition~(\ref{eq:intervals}) and invariants (I1), (I2) and (I3) at node $p$.

If $p$ is labeled $\wedge$, the situation is a bit trickier. The
propagation of $B^r_p$ is similar to before, except that $\max$ are
replaced with $\min$ and vice versa. However the definition of $L^r_p$
requires care. If we take the $s$ intervals for left child and $t$
intervals for right child, a naive procedure to obtain the pairwise
intersections will satisfy Condition~(\ref{eq:intervals}) but it may
result in $s t$ intervals, violating Invariant~(I1).  To ensure
Invariant (I1), we use Lemma~\ref{lem:conj-of-lists} inspired from
sorting networks~\cite{CLRS,AKS83} that takes $s t$ lists of intervals
and outputs a list of $\max(s, t)$ coalesced-and-sorted disjoint
intervals.  Note that Condition~(\ref{eq:intervals}) and invariants (I1), (I2) and (I3) are now
satisfied even when the node $p$ is labeled $\wedge$.  By inductively
continuing this construction, we obtain a list of at most polynomially many
intervals, each represented by polynomial sized Presburger circuits,
when we reach the root of the tree representing $\varphi_i$.

This completes the proof of the Claim~\ref{cl:intervals}. Now it remains to define the Skolem functions. From the Claim, we can easily obtain $F_i$, as a Presburger circuit for a Skolem function for $y$ in $\varphi_i$ by choosing deterministically some point in the intervals. Recall that $\ite$ was defined earlier and the ceiling function can also be easily defined. 
  \begin{align*}
    f_i=\ite(&\max_r F^r_i\geq 1,\max_r F^r_i,\ite(\min_r F^r_i=0,0,\min F^r_i))\\
    \text{where,}&\text{ for all } r\in R,\\
    f^r_i=\ite &(\lflag^r(\bar{x})=1,M(c^r)+r, \\
    &(\ite(\uflag^r(\bar{x})=1,M(d^r)+r,\\
    &\quad (\ite(a^r_1<b^r_1,Ma^r_1+r,\\
    &\quad\quad (\ite\ldots \\
    &\quad \quad\quad (\ite(a^r_{k_r}<b^r_{k_r},Ma^r_{k_r}+r,0))\ldots))))))\\
    \text{where,} &\text{ for all } 1\leq j\leq k_r,\\
    &a^r_{j}=\lceil{\fdiv_M(\alpha_{j}^r(\bar{x})-r)\rceil},\\
    &b^r_j=\fdiv_M(\beta_j^r(\bar{x})-r),\\
    &c^r=\fdiv_M(\lbound^r-r),\\
    &d^r=\lceil{\fdiv_M(\ubound^r-r)\rceil}
  \end{align*}
 To see that  $f_i$ as defined above are Skolem functions note that for any valuation to $\bar{x}$, if $F^r_i(\bar{x})$ gives a non-zero value  when we compute value of $y$, $y\equiv r\mod M$ is true for only one $r$ and in its corresponding $L_r$, we can choose any non-empty interval, which we do using the nested $\ite$. If all intervals are empty then we then fix $0$ as the output. Now, in $f_i(\bar{x})$ we just check if any $F^r_i$ gave a non-zero value and if so, we pick that (with a preference of the max over min), and if all are 0, then we just pick 0.

It remains to show Lemma~\ref{lem:conj-of-lists} which we used in the proof above. We state this Lemma with proof in \cref{app:conj-of-lists}.
\begin{restatable}{lemma}{lemConjOfLists}\label{lem:conj-of-lists}
Let $L_1 = \{[\alpha_{1,1}, \beta_{1,1}], \ldots [\alpha_{1,s},
  \beta_{1,s}]\}$ and $L_2 = \{[\alpha_{2,1}, \beta_{2,1}], \ldots
[\alpha_{2,t}, \beta_{2,t}]\}$ be two sets of intervals, where
$\alpha_{i,j}$ and $\beta_{k,l}$ are Presburger functions, represented
as Presburger circuits. There exist Presburger functions $g_1, h_1,
\ldots g_{\max(s,t)}, h_{\max(s,t)}$ such that $\set{L_1} \cap
\set{L_2} = \bigcup_{i=1}^{\max(s,t)} \set{[g_i, h_i]}$.  Moreover,
Presburger circuits for all $g_i$'s and $h_i$'s can be constructed in
time $\bigO{s^2.t^2 + s.t.S}$, where $S = \max\big(\max_{i=1}^s
|\alpha_{1,i}|, \max_{i=1}^s |\beta_{1,i}|, \max_{i=1}^t
|\alpha_{2,i}|, \max_{i=1}^t |\beta_{2,i}|\big)$.
\end{restatable}

We end this subsection by observing that the proof of Theorem~\ref{presburger-circuit-one-output}, in fact, shows something stronger, namely that every $\PA$-definable Skolem function can be obtained as a $\PA$-circuit using the computation of $L^r$ and $B^r$ above (see Corollary~\ref{all-PA-one-output}).

\subsection{Proof of Lemma~\ref{lem:conj-of-lists}}
\label[appsec]{app:conj-of-lists}
\newcommand{\scomp}{\ensuremath{\mathsf{IComp}}}
\newcommand{\sortnet}{\ensuremath{\mathsf{CandS}}}
\newcommand{\bubbleup}{\ensuremath{\mathsf{BubbleUp}}}
\newcommand{\bubbledown}{\ensuremath{\mathsf{BubbleDown}}}
  
Recall from Section~\ref{sec:semantic-nf} that we represent intervals using
ordered pairs of Presburger functions.  For non-empty intervals $I_1 =
[\alpha_1, \beta_1]$ and $I_2 = [\alpha_2, \beta_2]$, we say that
$I_1$ and $I_2$ are \emph{coalescable} if $\alpha_2 \le \beta_1 + 1$.
In such cases, $\set{[\min(\alpha_1, \alpha_2),
    \max(\beta_1,\beta_2)]}$ represents the coalesced interval.  For
arbitray intervals $I_1$ and $I_2$, we say that $I_1 \prec I_2$ if one
of the following holds:
\begin{itemize}
\item $I_1$ is an empty interval
\item $\alpha_1 \le \beta_1 < \alpha_2 \le \beta_2$, and $I_1$
  and $I_2$ are not coalescable.
\end{itemize}
It is easy to see that $\prec$ is a transitive relation.  A sequence
of intervals $\big([\alpha_1,\beta_1], \ldots [\alpha_k, \beta_k]
\big)$ is said to be \emph{coalesced-and-sorted w.r.t. $\prec$} iff
for all $1 \le i < j \le k$, we have $[\alpha_i, \beta_i] \prec
[\alpha_j, \beta_j]$.  As an example, the sequence $\big([1,0], [4,2],
[3, 7], [9, 10], [12, 20]\big)$ is coalesced-and-sorted, where the
first two intervals are empty.  On the other hand, the sequence
$\big([3, 7], [8, 10], [1, 0], [12, 20]\big)$ is \emph{not}
coalesced-and-sorted because $[3, 7]$ can be coalesced with $[8,10]$
to yield $[3,10]$, and additionally, $[8, 10] \not\prec [1,0]$.

\lemConjOfLists*
\begin{figure*}[h!]
  \scalebox{0.8}{
    \begin{tikzpicture}[font=\small,thick,>=latex]

\def\dist{0.8}   %
\def\nlines{6}   %
\pgfmathsetmacro\H{\nlines+4}
\draw[thin] (-1,1) rectangle (17, -\dist*\H - 0.5);

\foreach \i in {1,2,3,4,6,7,8} {
  \pgfmathsetmacro\y{-(\i-1)*\dist}
  \draw[->] (-1.5,\y) -- (-0.5,\y);

  \ifnum\i=4
    \node[left] at (-1.2,\y-1) {$\large{\vdots}$};
  \else
    \ifnum\i=1
      \node[left] at (-1.5,\y) {$I_1$};
    \else
      \ifnum\i=2
        \node[left] at (-1.5,\y) {$I_2$};
      \else
        \ifnum\i=8
           \node[left] at (-1.5,\y) {$I_k$};
        \fi
      \fi
    \fi
  \fi    
      
}

\pgfmathsetmacro\klb{-(\nlines+1.7)*\dist}
\pgfmathsetmacro\kru{0.5}
\draw[thick] (-0.5,\klb) rectangle (1.5, \kru);
\node (candSk) at (0.5,-3.0) {\large CandS$_k$};

\foreach \i in {1,2,3,4,6,7,8} {
  \pgfmathsetmacro\y{-(\i-1)*\dist}
  \draw[->] (1.5,\y) -- (17.5,\y);

  \ifnum\i=4
    \node[right] at (1.9,\y-1) {$\large{\vdots}$};
  \else
    \ifnum\i=1
      \node[above] at (1.9,\y) {$\widehat{I}_1$};
    \else
      \ifnum\i=2
        \node[above] at (1.9,\y) {$\widehat{I}_2$};
      \else
        \ifnum\i=8
           \node[above] at (1.9,\y) {$\widehat{I}_k$};
        \fi
      \fi
    \fi
  \fi    
      
}

\draw[->] (-1.5,-8*\dist) -- (-0.5,-8*\dist);
\draw[->] (-0.5,-8*\dist) -- (17.5,-8*\dist);
\node[left] at (-1.5, -8*\dist) {$I_{k+1}$};

\foreach \i in {1,2,3,4,6,7,8} {
  \pgfmathsetmacro\yl{-(\i)*\dist-0.2}
  \pgfmathsetmacro\yu{-(\i-1)*\dist +0.2}
  \pgfmathsetmacro\xx{6.5-(\i-1)*0.5}

  \ifnum\i=4
    \node[above] at (\xx,\yl) {$\large{\iddots}$};
  \else
    \draw[line width = 1.5mm] (\xx, \yl) -- (\xx, \yu);
  \fi
}

\pgfmathsetmacro\klb{-(\nlines+1.7)*\dist}
\pgfmathsetmacro\kru{0.5}
\draw[thick] (2.8,-8*\dist-0.5) rectangle (7, 0.7);
\node (BUk) at (5,-7.5) {\large BubbleUp$_{k+1}$};

\node[above] at (7.3,  0)          {$\widetilde{I}_1$};
\node[above] at (7.3, -1*\dist)    {$\widetilde{I}_2$};

\node[above] at (7.3, -4.8*\dist)    {$\large{\vdots}$};
\node[above] at (7.4, -8*\dist)    {$\widetilde{I}_{k+1}$};

\foreach \i in {1,2,3,4,6,7,8} {
  \pgfmathsetmacro\yl{-(\i)*\dist-0.2}
  \pgfmathsetmacro\yu{-(\i-1)*\dist +0.2}
  \pgfmathsetmacro\xx{16.8-(\i-0.6)*0.6}

  \ifnum\i=4
    \node[above] at (\xx,\yl) {};%
  \else
    \draw[line width = 1.5mm] (\xx, \yl) -- (\xx, \yu);
  \fi
}

\foreach \i in {1,2,3,4,6,7} {
  \pgfmathsetmacro\yl{-(\i)*\dist-0.2}
  \pgfmathsetmacro\yu{-(\i-1)*\dist +0.2}
  \pgfmathsetmacro\xx{15.8-(\i-0.6)*0.6}

  \ifnum\i=4
    \node[above] at (\xx,\yl) {$\large{\iddots}$};
  \else
    \draw[line width = 1.5mm] (\xx, \yl) -- (\xx, \yu);
  \fi
}

\foreach \i in {1,2,3,4,6} {
  \pgfmathsetmacro\yl{-(\i)*\dist-0.2}
  \pgfmathsetmacro\yu{-(\i-1)*\dist +0.2}
  \pgfmathsetmacro\xx{14.8-(\i-0.6)*0.6}

  \ifnum\i=4
    \node[above] at (\xx,\yl) {};%
  \else
    \draw[line width = 1.5mm] (\xx, \yl) -- (\xx, \yu);
  \fi
}

\foreach \i in {1,2,3} {
  \pgfmathsetmacro\yl{-(\i)*\dist-0.2}
  \pgfmathsetmacro\yu{-(\i-1)*\dist +0.2}
  \pgfmathsetmacro\xx{8.7+(\i-0.4)*0.4}

    \draw[line width = 1.5mm] (\xx, \yl) -- (\xx, \yu);
}

\foreach \i in {1,2,3,4} {
  \pgfmathsetmacro\yl{-(\i)*\dist-0.2}
  \pgfmathsetmacro\yu{-(\i-1)*\dist +0.2}
  \pgfmathsetmacro\xx{9.7+(\i-0.4)*0.4}

  \ifnum\i=4
    \node[above] at (\xx,\yl) {$\huge{\ddots}$};
  \else
    \draw[line width = 1.5mm] (\xx, \yl) -- (\xx, \yu);
  \fi
}

\foreach \i in {1,2,3} {
  \pgfmathsetmacro\yl{-(\i)*\dist-0.2}
  \pgfmathsetmacro\yu{-(\i-1)*\dist +0.2}
  \pgfmathsetmacro\xx{10.7+(\i-0.4)*0.4}

    \draw[line width = 1.5mm] (\xx, \yl) -- (\xx, \yu);
}

\pgfmathsetmacro\klb{-(\nlines+1.7)*\dist}
\pgfmathsetmacro\kru{0.5}
\draw[thick] (8.2,-8*\dist-0.5) rectangle (16.8, 0.6);
\node (BUk) at (11,-7.5) {\large BubbleDown$_{k+1}$};

\node[below] at (6, -11*\dist) {\Large CandS$_{k+1}$};

\draw[->] (9, -11*\dist) -- (10, -11*\dist) node[right] {$I_{\mathrm{comp}_h}$};
\draw[->] (9, -11.6*\dist) -- (10, -11.6*\dist) node[right] {$I_{\mathrm{comp}_l}$};
\draw[line width = 1.5mm] (9.5,-10.8*\dist) -- (9.5,-11.8*\dist);
\node[right] at (11, -11.3*\dist) {\Large IComp block};

    \end{tikzpicture}

  }
  
    \caption{Sorting network inspired coalese-and-sort network}
 \label{fig:sortnet}
\end{figure*}
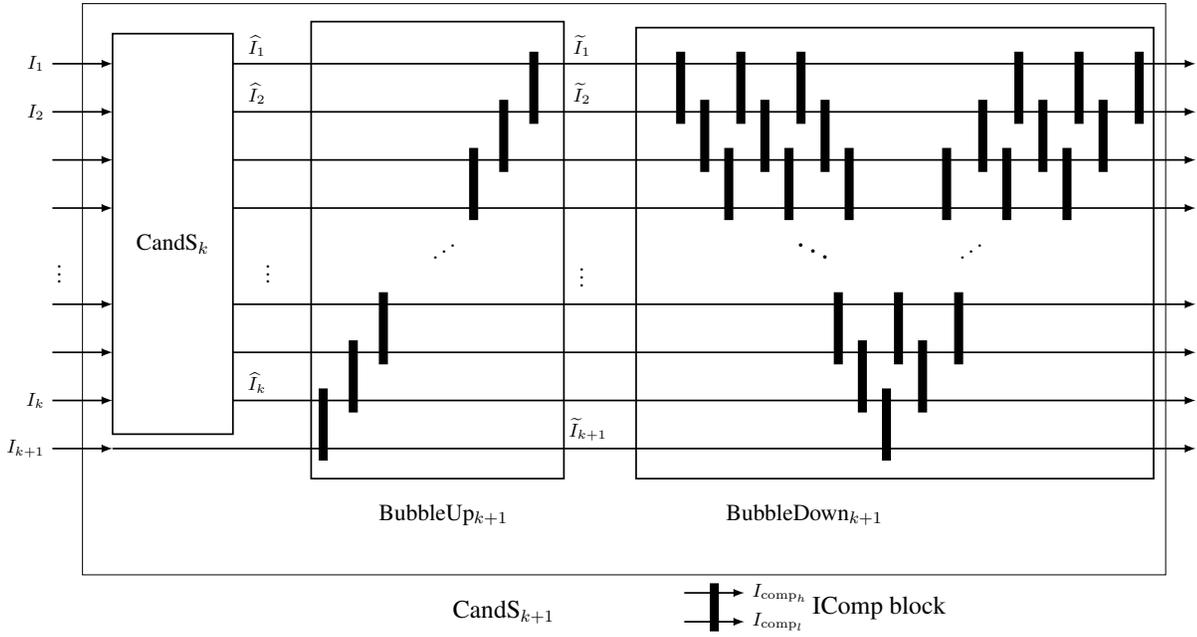

Towards proving Lemma~\ref{lem:conj-of-lists}, we describe an
algorithm that constructs Presburger circuits for all $g_i$'s and $h_i$'s
in $\bigO{s^2.t^2 + s.t.S}$ time, where $S = \max\big(\max_{i=1}^s
  |\alpha_{1,i}|, \max_{i=1}^s |\beta_{1,i}|, \max_{i=1}^t
  |\alpha_{2,i}|, \max_{i=1}^t |\beta_{2,i}|)$.
  The algorithm works in two phases. In the first phase, we intersect
  each interval in $L_1$ with each interval in $L_2$ to obtain $s.t$
  intervals.  In the second phase, we coalesce and sort the resulting
  $s.t$ intervals w.r.t. the $\prec$ relation, and retain the
  ``highest'' $\max(s,t)$ intervals.  Since the intersection of a
  union of $s$ intervals and a union of $t$ intervals cannot have more
  than $\max(s,t)$ intervals, such an intersection can always be
  represented as the union of at most $\max(s,t)$ non-coalescable
  intervals.  The lemma follows immediately.

It is easy to see that the lower (resp. upper) end-point of each of
the non-coalescable $\max(s,t)$ intervals mentioned above is a lower
(resp. upper) end-point of the given $s + t$ intervals.  Hence,
computing the ``highest'' $\max(s,t)$ intervals effectively amounts to
choosing the right pairs of end-points of these intervals from the
end-points of the given intervals.  This is done using a
coalesce-and-sort network described below.

For every pair of intervals $[\alpha_{1,i}, \beta_{1,i}]$ and
$[\alpha_{2,j}, \beta_{2,j}]$, the intersection interval is given by
$[\max(\alpha_{1,i},\alpha_{2,j}), \min(\beta_{1,i},\beta_{2,j})]$. It
is easy to see that this works in all cases, even if one of the
intervals is empty.  Constructing the corresponding Presburger
circuits takes time at most linear in $|\alpha_{i,i}| + |\beta_{1,i}|
+ |\alpha_{2,j}| + |\beta_{2,j}|$, i.e. linear in $S$, where $S =
\max\big(\max_{i=1}^s |\alpha_{1,i}|, \max_{i=1}^s |\beta_{1,i}|,
\max_{i=1}^t |\alpha_{2,i}|, \max_{i=1}^t |\beta_{2,i}|)$.

In order to coalesce-and-sort the $s.t$ intervals obtained above, we
discuss below an adaptation of sorting networks~\cite{CLRS}.
Specifically, we consider an interval comparator (or basic
coalesce-and-sort gadget) $\scomp$ that takes two intervals
$[\alpha_1, \beta_1]$ and $[\alpha_2, \beta_2]$ as inputs, and
produces two intervals as output.  We use $\scomp_h$
(resp. $\scomp_l$) to denote the ``high'' (resp. ``low'') interval
output of $\scomp$.  We want $\scomp$ to behave as follows.
\begin{itemize}
  \item If any input is an empty interval, then $\scomp_l = [1,0]$ (empty interval) and $\scomp_h$ is the other (possibly empty) input
    interval.
  \item If both inputs are non-empty intervals,
    \begin{itemize}
      \item If $[\alpha_i, \beta_i] \prec [\alpha_j, \beta_j]$, then
        $\scomp_l = [\alpha_i, \beta_i]$ and $\scomp_h = [\alpha_j,
        \beta_j]$.
      \item In all other cases, $\scomp_l = [1,0]$ (empty interval)
        and $\scomp_h$ is the interval resulting from coalescing
        the input intervals.
    \end{itemize}
\end{itemize}
It is easy to see that the interval comparator can be constructed
using Presburger functions of size polynomial in $\sum_{i=1}^2
(|\alpha_i| + |\beta_i|)$. Specifically, if
$\scomp_l([\alpha_1,\beta_1], [\alpha_2, \beta_2]) = [\lambda_l,
  \mu_l]$, then
\begin{itemize}
  \item $\lambda_l := \ite(\alpha_1 > \beta_1, 1, \ite(\alpha_2
    > \beta_2, 1, \ite(\beta_1 < \alpha_2, \alpha_1,
    \ite(\beta_2 < \alpha_1, \alpha_2, 1))))$
  \item $\mu_l := \ite(\alpha_1 > \beta_1, 0, \ite(\alpha_2
    > \beta_2, 0, \ite(\beta_1 < \alpha_2, \beta_1,
    \ite(\beta_2 < \alpha_1, \beta_2, 0))))$
\end{itemize}
Similarly, if $\scomp_h([\alpha_1,\beta_1], [\alpha_2,\beta_2]) =
[\lambda_h, \mu_h]$, the Presburger functions for $\lambda_h$ and
$\mu_h$ can be analogously defined.

We now construct a coalesce-and-sort network inductively using the
comparator $\scomp$ defined above.  With only two input intervals to
coalesce-and-sort, $\scomp$ itself serves as the network.  For $k \ge
2$, let $\sortnet_k$ (for Coalesce-and-Sort) represent the network for
$k~(\ge 2)$ inputs. We build $\sortnet_{k+1}$ inductively from
$\sortnet_k$, as shown in Fig.~\ref{fig:sortnet}.  Here the module
$\bubbleup_{k+1}$ consists of $k$ copies of $\scomp$ connected such
that the $k+1^{st}$ interval input $I_{k+1}$ can bubble up (possibly
after coalescing with other intervals) to its rightful position in the
sequence of output intervals.  The module $\bubbledown_{k+1}$ consists
of $\bigO{k^2}$ copies of $\scomp$ connected such that all empty
intervals can bubble down to the later part of the output sequence of
intervals.

\begin{claim}\label{clm:sort-net-property}
  For every $k \ge 2$, let $\sortnet_k\big((I_1, \ldots I_k)\big) =
  (\widehat{I}_1, \ldots \widehat{I}_k)$. Then, the following hold:
  \begin{enumerate}
    \item For every $1 \le j < i \le k$, we have $\widehat{I}_i \prec
      \widehat{I}_j$, and
    \item $\bigcup_{i=1}^k \set{I_i} ~=~ \bigcup_{i=1}^k \set{\widehat{I}_i}$.
  \end{enumerate}
\end{claim}
We prove the claim inductively for $k \ge 2$.  For the base case
($k=2$), we have $\sortnet_2\big((I_1, I_2)\big) =
\big(\scomp_l(I_1,I_2), \scomp_h(I_1, I_2)\big)$.  The claim follows
immediately from the definition of $\scomp$ above.  Next, we assume
that the claim holds for some $k \ge 2$.  Since $\widehat{I}_i \prec
\widehat{I}_j$ for $1 \le i < j \le k$, it follows from the definition
of $\prec$ that if $\set{\widehat{I}_j} = \emptyset$, then all
$\set{\widehat{I}_i}$ for $1 \le i < j$ are also $\emptyset$.
Furthermore, all non-empty intervals in $\widehat{I}_1, \ldots
\widehat{I}_k$ are disjoint and non-coalescable.

For the inductive step, we refer to the construction of
$\sortnet_{k+1}$ in Fig.~\ref{fig:sortnet}.  Let
$\big(\widetilde{I}_1, \ldots \widetilde{I}_{k+1}\big)$ be the
sequence of outputs of the $\bubbleup_{k+1}$ module in
Fig.~\ref{fig:sortnet}.  From the inductive properties of
$\widehat{I}_1, \ldots \widehat{I}_k$ above, from the structure of
$\bubbleup_{k+1}$ and from the properties of $\scomp$, it can be
easily shown that the following hold:
\begin{itemize}
  \item $\bigcup_{i=1}^{k+1} \set{\widetilde{I}_i} ~=~ \bigcup_{i=1}^k
    \set{\widehat{I}_i} \cup \set{I_{k+1}}$ $~=~ \bigcup_{i=1}^{k+1}
    \set{I_i}$. This holds invariantly at every level of the
    $\bubbleup_{k+1}$ module.
  \item For every non-empty interval $\widetilde{I}_i$ and $\widetilde{I}_j$,
    where $1 \le j < i \le k$, $\widetilde{I}_i$ and $\widetilde{I}_j$ are
    not coalescable. This holds invariantly for $1 \le i < j \le l$ after
    every level $l$ of the $\bubbleup_{k+1}$ module.
  \item The sequence $\big(\widetilde{I}_1, \ldots
    \widetilde{I}_{k+1}\big)$ can be divided into at most two
    sub-sequences $\big(\widetilde{I}_1, \ldots
    \widetilde{I}_{p}\big)$ and $\big(\widetilde{I}_{p+1}, \ldots
    \widetilde{I}_{k+1}\big)$ such that each of the sub-sequences are
    coalesced-and-sorted w.r.t. $\prec$ and $\set{\widetilde{I}_p} =
    \emptyset$, although the entire sequence $\big(\widetilde{I}_1,
    \ldots \widetilde{I}_{k+1}\big)$ may not be
    coalesced-and-sorted. This happens if $I_{k+1}$ coalesces with
    $\widehat{I_{p+1}}$ for the first time at level $p$ of the
    $\bubbleup_{k+1}$ module.
\end{itemize}
Thus, the sequence $\big(\widetilde{I}_1, \ldots
\widetilde{I}_{k+1}\big)$ \emph{almost} satisfies the properties of our
claim, except perhaps for a contiguous stretch of empty intervals
starting from $\widetilde{I}_p$.  Significantly, all pairs of distinct
intervals in $\big(\widetilde{I}_1, \ldots \widetilde{I}_{k+1}\big)$
are non-coalescable and can be ordered w.r.t. $\prec$. Therefore, the output
of the $\bubbledown_{k+1}$ network correctly orders all of the
intervals $\widetilde{I}_1, \ldots \widetilde{I}_{k+1}$, ensuring that
the claim holds at the output of $\sortnet_{k+1}$.

Note that the complete coalesce-and-sort network can be constructed
using Presburger circuits in $\bigO{k^2}$ time where $k$ is the count
of intervals to be coalesced and sorted.  In our case, $k = s.t$;
hence the coalesce-and-srt network can be constructed in
$\bigO{s^2.t^2}$ time.  The overall circuit also has $s.t$ intervals
generated from the first phase of the algorithm, and each such
interval can be constructed in time $\bigO{S}$, where $S =
\max\big(\max_{i=1}^s |\alpha_{1,i}|, \max_{i=1}^s |\beta_{1,i}|,
\max_{i=1}^t |\alpha_{2,i}|, \max_{i=1}^t |\beta_{2,i}|)$.  Hence,
the overall time-complexity is in $\bigO{s^2.t^2 + s.t.S}$.

\subsection{All $\PA$-definable Skolem functions}
We observe that the proof of Theorem~\ref{presburger-circuit-one-output}, in fact, shows something stronger, namely that every $\PA$-definable Skolem function can be obtained as a $\PA$-circuit using the computation of $L^r$ and $B^r$ above. This follows from the Claim~\ref{cl:intervals}, where the set of all possible correct values of Skolem functions is being represented. As a result, if we determine the output of $f_i$ by using the given $\PA$-definable Skolem function, say $f^\star$, to choose the interval and point within these intervals, we get a $\PA$ circuit that computes $f^\star$. More formally,

\begin{corollary} \label{all-PA-one-output}
Let $\varphi(\bar{x},y)$ be a $y$-modulo-tame $\PA$ formula, where $M$
is the maximum modulus in any modular constraint involving $y$.  For
every Skolem function $f^\star(\bar{x})$ for $y$ in $\forall \bar{x} \exists
y\, \varphi(\bar{x}, y)$, there exist choice functions $\rho:
\mathbb{Z}^n \rightarrow \{0,1,\ldots M-1\}$, $\pi: \mathbb{Z}^n
\times \{0,1,\ldots M-1\} \rightarrow \mathsf{Intrvl}$ and $\sigma:
\mathbb{Z}^n \times \mathsf{Intrvl} \rightarrow \mathbb{Z}$ such that
$f^\star(\bar{x}) = \sigma\big(\bar{x}, \pi(\bar{x},\rho(\bar{x}))\big)$,
where
\begin{itemize}
  \item $\mathsf{Intrvl}$ is the set of intervals corresponding to
    $L^r$ and $B^r$ for $0 \le r < M$, as used in the proof of
    Theorem~\ref{presburger-circuit-one-output}.
  \item $\pi(\bar{u}, r)$ evaluates to an interval $I$ in
    $L^r$ or $B^r$.
  \item $\sigma(\bar{u}, I)$ evaluates to an integer in $\set{I}$.
\end{itemize}
Morover, if $f^\star(\bar{x})$ is Presburger definable, all the choice
functions $\rho, \pi$ and $\sigma$ above are also Presburger definable.
\end{corollary}

\subsection{Proof of Theorem~\ref{thm:optimal}}\label[appsec]{app:optimal}
\thmOptimal*

\begin{proof}
  W.l.o.g. let $\varphi(\bar{x},y)$ be of the form $\bigvee_{i=1}^k
  \varphi_i(\bar{x},y)$, where each $\varphi_i(\bar{x},y)$ is a
  maximal conjunctive formula.  Let $A$ be an algorithm that takes
  $\varphi(\bar{x},y) \in \mathfrak{S}$ and computes $\exists y:
  \varphi(\bar{x},y)$ in time polynomial in $|\varphi|$.  Let the
  corresponding quantifer-eliminated formula be denoted
  $\widetilde{\varphi}(\bar{x})$.
  
  We now construct the formula $\varphi^\star = \bigvee_{i=1}^k
  \big(\varphi_i \wedge \widetilde{\varphi}\big)$, and claim that
  $\varphi^\star$ is semantically equivalent to $\varphi$ and is in
  $\PASyn$.

  The semantic equivalence is easily seen from the observation that
  $\varphi \implies \exists y: \varphi$.  Notice also that since
  $\varphi$ is $y$-modulo tame, and since the quantifier-eliminated
  formula $\widetilde{\varphi}$ does not have $y$ in any of its atomic
  formulas, so $\bigvee_{i=1}^k(\varphi_i \wedge \widetilde{\varphi})$,
  i.e. $\varphi^\star$ is also $y$-modulo tame.

  Furthermore, $\exists^{\mathsf{local}} y: \varphi^\star$ is, by
  definition, $\bigvee_{i=1}^k\big((\exists^{\mathsf{local}} y:
  \varphi_i) \wedge \widetilde{\varphi}\big)$, since
  $\widetilde{\varphi}$ does not have any sub-formula involving $y$.
  For the same reason, $\exists y: \varphi^\star$ is also
  $\bigvee_{i=1}^k \big((\exists y: \varphi_i) \wedge
  \widetilde{\varphi}\big)$.  However, $\widetilde{\varphi}$ is
  semantically equivalent to $\bigvee_{i=1}^k \big(\exists y:
  \varphi_i)$.  Hence $\exists y: \varphi^\star$ is semantically
  equivalent to $\exists y:\varphi$, or equivalently to
  $\widetilde{\varphi}$.
  
  Now, if possible, suppose $\exists^{\mathsf{local}} y:
  \varphi^\star$ is true and $\exists y: \varphi^\star$ is false for
  some assignment $u$ of $\bar{x}$.  This implies that
  $\widetilde{\varphi}$ is false for assignment $u$ of $\bar{x}$.
  Hence $(\exists^{\mathsf{local}} y: \varphi)\wedge \widetilde{\varphi}$ is
  also false, contradicting our premise that $\exists^{\mathsf{local}} y:
  \varphi^\star$ is true.

  Clearly, the above approach of constructing $\varphi^\star$ takes
  time polynomial in $|\varphi|$ if algorithm $A$ computes
  $\widetilde{\varphi}$ in time polynomial in $|\varphi|$.
\end{proof}

\section{Additional material on Section~\ref{sec:syntactic-nf}}
\label[appsec]{app:syntactic-nf}
\subsection{Auxiliary lemma}\label[appsec]{app:syntactic-nf-preserve-integrality}
\begin{lemma}\label{preserve-integrality}
	Suppose $A\colon\Q^k\to\Q^\ell$ is an affine map and all denominators in the coefficients of $A$ divide $M\in\Z$. If $\bar{u},\bar{u}'\in\Z^k$ with $\bar{u}\equiv\bar{u}'\pmod{M}$, then
	$A(\bar{u})\in\Z^\ell$ if and only if $A(\bar{u}')\in\Z^{\ell}$.
\end{lemma}
\begin{proof}
	Since all denominators in $A$ divide $M$, we can write
	$A=\tfrac{1}{M}B$ for some affine map $B\colon\Q^k\to\Q^\ell$ that has
	only integer coefficients. Then we have $B(\bar{u})\equiv
	B(\bar{u}')\pmod{M}$ and hence the following are equivalent:
	(i)~$A(\bar{u})\in\Z^\ell$, (ii)~$M$ divides every entry of
	$B(\bar{u})$, (iii)~$M$ divides every entry of $B(\bar{u}')$,
	(iv)~$A(\bar{u}')\in\Z^\ell$.
\end{proof}
\subsection{Proof of Theorem~\ref{syntactic-implies-semantic-nf}}\label[appsec]{app:syntactic-implies-semantic-nf}
\syntacticImpliesSemanticNF*
\begin{proof}
	Suppose $\varphi(\bar{x},\bar{y})$ is in $\DNNF$. Then $\varphi$ is $y_i$-modulo-tame for each $i\in[1,m]$ by definition. Thus, it remains to show that for all $\bar{x}\in\Z^n$ and all $(y_1,\ldots,y_i)\in\Z^i$ the implication
	\[ \lexists y_{i+1},\ldots,y_m\colon\varphi(\bar{x},\bar{y})\to \exists y_{i+1},\ldots,y_m\colon\varphi(\bar{x},\bar{y}). \]
	holds. Call the left-hand side $\tau(\bar{x}^i)$. Suppose $\tau$ is satisfied
	for some $\bar{x}^i$. Among all maximal conjunctive subformulas of
	$\tau$, at least one, say $\tau'$, must be satisfied. Then $\tau'$ is
	obtained by locally quantifying $y_{i+1},\ldots,y_m$ in some maximal
	conjunctive subformula $\varphi'$ of $\varphi$.  Since $\varphi$ is in
	$\DNNF$, and $\varphi'$ is a maximal conjunctive subformula, there are $M\in\Z$, a vector $(\bar{r},\bar{s})\in[0,M-1]^{n+m}$, and 
	affine transformations $A_1,\ldots,A_m\colon \Q^{n+i}\to\Q^{m-i}$ such
	that $\varphi'$ is a positive Boolean combination of building blocks as
	in \eqref{piece-dnnf}, where $\psi$ is an atomic formula. Moreover, $M$ divides all denominators of coefficients in $A_i$ for $i\in[1,m]$, and we have $A_i(\bar{r}^i)\in\Z^{m-i}$. 

	Observe
	that after locally quantifying $y_{i+1},\ldots,y_m$ in such a buildling
	block, the resulting formula will still imply
	$\psi(\bar{x}^i,A_i(\bar{x}^i))$ and $\bar{x}^i\equiv \bar{r}^i$,
	because these subformulas do not contain any of the variables
	$y_{i+1},\ldots,y_m$. In particular, since $A_i(\bar{r}^i)\in\Z^{m-i}$, we may conclude $A_i(\bar{x}^i)\in\Z^{m-i}$ (see \cref{preserve-integrality}).

	This implies that
	$\varphi'(\bar{x}^i,A_i(\bar{x}^i))$ holds, because every building
	block whose counterpart in $\tau'$ that is satisfied, will be satisfied
	in $\varphi'$ by $(\bar{x}^i,A_i(\bar{x}^i))$. This implies that
	$(\bar{x}^i,A_i(\bar{x}^i))$ is an integral solution to the entire formula $\varphi$.
\end{proof}

\subsection{Proof of Theorem~\ref{construct-syntactic-nf}}\label[appsec]{app:construct-syntactic-nf}
Before we prove \cref{construct-syntactic-nf}, we need a version of \cref{affine-transformations} in the setting of linear inequalities \emph{and modulo constraints}.
\begin{proposition}\label{affine-transformations-conjunctions}
	Let $\psi(\bar{x},\bar{y})$ be a conjunction of atomic formulas over
	the variables $\bar{x}=(x_1,\ldots,x_n)$ and $\bar{y}=(y_1,\ldots,y_m)$, and let
	$\bar{b}\in\Z^m$ be a vector.  If the formula $\psi(\bar{x},\bar{b})$
	is satisfiable over $\Z^n$, then it is satisfied by an integral
	vector of the form $D\bar{b}+\bar{d}$, where $D\in\Q^{n\times m}$,
	$\bar{d}\in\Q^n$ with $\fracnorm{D}$ and $\fracnorm{\bar{d}}$ are
	exponential in the size of $\psi$.
\end{proposition}
\begin{proof}
	Let $M$ be the least common multiple of all moduli in $\psi$. Since
	$\psi(\bar{x},\bar{b})$ is satisfiable over $\Z^n$, there is a vector
	$\bar{r}\in[0,M-1]^n$ such that there is a solution $\bar{x}\in\Z^n$
	with $\bar{x}\equiv\bar{r}\pmod{M}$. There is also some
	$\bar{s}\in[0,M-1]^m$ with $\bar{b}\equiv\bar{s}\pmod{M}$. We may
	therefore assume that $\psi(\bar{x},\bar{y})$ is of the form 
	\[\varphi(\bar{x},\bar{y})~\wedge~ (\bar{x},\bar{y})\equiv(\bar{r},\bar{s})\pmod{M}, \]
	where $\varphi$ is a conjunction of linear inequalities. In this form, we can observe that $\bar{x}$ is a solution to $\psi(\bar{x},\bar{b})$ if and only if there exists a vector $\bar{z}\in\Z^m$ such that $\bar{x}=M\cdot\bar{z}+\bar{r}$ and $\varphi(M\bar{z}+\bar{r})$. Therefore, the conjunction
	\[ \varphi(M\cdot\bar{z}+\bar{r},\bar{b}) \]
	of linear inequalities in $\bar{z}$ has a solution. By \cref{affine-transformations}, it has an integral solution of the form $E\bar{b}+\bar{e}$, where $E\in\Q^{n\times m}$, $\bar{e}\in\Q^n$ with $\fracnorm{E}$ and $\fracnorm{\bar{e}}$ at most exponential in the size of $\psi$. But then $M\cdot (E\bar{b}+\bar{e})+\bar{r}=(ME)\bar{b}+(M\bar{e}+\bar{r})$ is an integral solution to $\psi(\bar{x},\bar{b})$, with size bounds as desired.
\end{proof}

\constructSyntacticNF*
\begin{proof}
	Let $\varphi$ be a quantifier-free PA formula. 
	Since we can write $\varphi$ as an exponential disjunction of
	polynomial-sized conjunctions of atoms (by bringing it into DNF), and a disjunction of a $\DNNF$\ formula is again in $\DNNF$, it suffices to perform the translation for a
	single conjunction of atoms. 

	Now \cref{affine-transformations-conjunctions} tells us that if for some $\bar{x}^i$, there is a $\bar{y}^i$ with $\varphi(\bar{x}^i,\bar{y}^i)$, then there is such a $\bar{y}^i$ that can be written as $D(\bar{x}^i)$ for some affine transformation $D\colon\Q^{n+i}\to\Q^{m-i}$, where the coefficients in $D$ have at most polynomially many bits.

	This means, for each $i=1,\ldots,m$, there is a list of exponentially many affine transformations $D_{i,1},\ldots,D_{i,s}$ (of polynomial bit-size) such that if $\varphi(\bar{x}^i,\bar{y}^i)$ has a solution $\bar{y}^i$, then it has one of the form $D_{i,j}(\bar{x}^i)$ for some $j\in[1,s]$.

	Let $\varphi_1,\ldots,\varphi_\ell$ be the atomic formulas in $\varphi$.
	Let $F$ be the set of functions $f\colon[1,m]\to[1,s]$, which we think of as assignments of an affine transformation $D_{i,f(i)}$ to each $i\in[1,m]$. Moreover, for each $f\in F$, let $M_f$ be the least common multiple of the denominators occurring in $D_{1,f(1)},\ldots,D_{m,f(m)}$ and of all modulo occurring in $\varphi$. Note that then, $M_f$ has polynomial bit-size. Finally, let $I_f\subseteq [0,M-1]^{m+n}$ be the set of all vectors $(\bar{r},\bar{s})\in[0,M-1]^{m+n}$ such that $D_{i,f(i)}(\bar{r}^i)\in\Z^{m-i}$.
	Now consider the formula $\tau=\bigvee_{f\in F}\bigvee_{(\bar{r},\bar{s})\in I_f}\bigwedge_{k=1}^\ell
	\tau_{f,\bar{r},\bar{s},k}$, where
	$\tau_{f,\bar{r},\bar{s},k}$ is the formula
	\begin{multline} \left(\varphi_k(\bar{x},\bar{y})\vee \bigvee_{i=1}^m \bar{y}^i=D_{i,f(i)}(\bar{x}^i)\right) \\
	\wedge \bigwedge_{i=1}^m\varphi_k(\bar{x}^i,D_{i,f(i)}(\bar{x}^i))\wedge (\bar{x},\bar{y})\equiv (\bar{r},\bar{s})\pmod{M_f}. \label{building-block-dnnf}\end{multline}
	Note that $\tau$ is almost in $\DNNF$---the only condition that might be violated is modulo-tameness: It is possible that modulo constraints inside some $\varphi_k$ are not w.r.t.\ $M_f$. However, this is easy to repair: One can replace each $\tau_{f,\bar{r},k}$ by a disjunction of (exponentially many) building blocks where these modulo constraints are written modulo $M_f$. This is possible because by construction of $M_f$, all the moduli in $\varphi$ divide $M_f$. This new formula will be in $\DNNF$ and equivalent to $\tau$. Thus it remains to show that $\tau$ is equivalent to $\varphi$.

	Suppose $\varphi(\bar{x},\bar{y})$ holds. Then there is a function
	$f\in F$ such that
	$\varphi(\bar{x}^i,D_{i,f(i)}(\bar{x}^i))$ and the vectors $D_{i,f(i)}(\bar{x}^i)$ for $i\in[1,m]$ are all integral. This means, there is a vector $(\bar{r},\bar{s})\in I_f$ with $(\bar{x},\bar{y})\equiv (\bar{r},\bar{s})\pmod{M_f}$.
	In particular,
	$\varphi_k(\bar{x}^i,D_{i,f(i)}(\bar{x}^i))$ holds for each
	$k=1,\ldots,\ell$. This implies that \eqref{building-block-dnnf} holds
	with the left disjunct in the large parenthesis.

	Conversely, suppose \cref{building-block-dnnf} is satisfied for some $f\in F$, $(\bar{r},\bar{s})\in I_f$, and some $\bar{x}\in\Z^n$ and some $\bar{y}\in\Z^m$. Clearly, if in the large parenthesis, we satisfy $\varphi_k(\bar{x},\bar{y})$ for every $k=1,\ldots,\ell$, then $\varphi(\bar{x},\bar{y})$ is satisfied by definition. So consider the case where for some $i\in[1,m]$, we satisfy $\bar{y}^i=D_{i,f(i)}(\bar{x}^i)$ rather than $\varphi_k(\bar{x},\bar{y})$. Then the second conjunct of $\tau_{f,\bar{r},\bar{s},k}$ tells us that in particular, the assertion $\varphi_k(\bar{x}^i,D_{i,f(i)})=\varphi_k(\bar{x}^i,\bar{y}^i)=\varphi_k(\bar{x},\bar{y})$ holds for each $k$. The latter means that $(\bar{x},\bar{y})$ satisfies $\varphi$.
\end{proof}

\subsection{Proof of Theorem~\ref{succinctness-semantic-syntactic}}\label[appsec]{app:succinctness-semantic-syntactic}
\succinctnessSemanticSyntactic*
\begin{proof}
	Consider $\Psi_n(x,y):=x<y\le x+2^n \wedge y\equiv 0\pmod{2^n}$. Then
	$\Psi_n$ is in $\PASyn$: It is $y$-modulo-tame, and for
	every $x\in\Z$, there is a $y\in\Z$ with $\Psi_n(x,y)$, so that the
	condition on local and global quantification holds as well.

	Suppose $\psi_n$ is a $\DNNF$ equivalent. 
	Since $\psi_n(x,y)$ is in $\DNNF$, there is a list of affine
	transformations $A_1,\ldots,A_\ell$ (one for each maximal conjunctive
	subformula), such that for every $x$ with $\exists y\colon\psi(x,y)$,
	we also have $\psi(x,A_j(x))$ for some $j\in[1,\ell]$. These
	transformations must syntactically appear in the formula $\psi_n$,
	hence $\ell\le|\psi_n|$. We claim that $\ell\ge 2^n$, which implies
	$|\psi_n|\ge 2^n$ and thus the
	\lcnamecref{succinctness-semantic-syntactic}.

	Observe that $\Phi_n$ defines a function $f\colon\Z\to\Z$, which maps
	$x$ to the smallest multiple of $2^n$ above $x$. This means, for every
	$x\in\Z$, there must be a $j\in[1,\ell]$ with $A_j(x)=f(x)$. Write
	$A_j(x)=a_jx+b_j$ for some $a_j,b_j\in\Q$. Note that if $a_j\ne 1$,
	then $A_j$ can only provide the correct value for some finite interval
	$[-M,M]$ of numbers $x$, because $a_jx+b_j=f(x)$ implies
	$|(a_j-1)x+b_j|=|a_jx+b_j-x|\le 2^n$.  Thus, for $|x|>M$, the only
	remaining $A_j$ are those with $A_j(x)=x+b_j$.  However, this means for
	every residue $r\in[0,2^n-1]$, there must be some $j\in[1,\ell]$ with
	$b_j=r$. This implies $\ell\ge 2^n$.
\end{proof}

\label{afterbibliography}
\newoutputstream{pagestotal}
\openoutputfile{main.pagestotal.ctr}{pagestotal}
\addtostream{pagestotal}{\getpagerefnumber{afterbibliography}}
\closeoutputstream{pagestotal}

\newoutputstream{todos}
\openoutputfile{main.todos.ctr}{todos}
\addtostream{todos}{\arabic{@todonotes@numberoftodonotes}}
\closeoutputstream{todos}
\end{document}